%% file: main.tex
\documentclass[12pt,reqno]{article}

\usepackage[english]{babel}
\usepackage[T1]{fontenc}
\usepackage[utf8]{inputenc}
\usepackage{authblk}
\usepackage{amsmath, amsthm}
\usepackage{graphicx}

\usepackage[capposition=top]{floatrow}
\usepackage{subcaption}
\usepackage{pdflscape}
\usepackage{caption}
\usepackage{rotating}

\usepackage{color}
\usepackage{tabularx,ragged2e,booktabs,caption}
\usepackage{adjustbox}
\usepackage[flushleft]{threeparttable}

\usepackage{csquotes}
\usepackage[style=authoryear,maxcitenames=3,uniquelist=false,maxbibnames=99,backend=bibtex]{biblatex}
\DeclareLanguageMapping{english}{english-apa}

\renewbibmacro{in:}{}

\renewbibmacro*{volume+number+eid}{%
  \printfield{volume}%
  \setunit*{\addnbspace}
  \printfield{number}%
  \setunit{\addcomma\space}%
  \printfield{eid}}
\DeclareFieldFormat[article]{number}{\mkbibparens{#1}}

\usepackage{longtable}
\usepackage[toc,page]{appendix}

\newtheorem{prop}{Proposition}

\makeatletter
\renewcommand*{\@fnsymbol}[1]{\ensuremath{\ifcase#1\or *\or 1 \else\@ctrerr\fi}}
\makeatother

\makeatletter
\renewcommand{\@seccntformat}[1]{\csname the#1\endcsname.\quad}
\makeatother

\begin{document}

\title{Quantifying Different Gains from Trade in Quality}

\author{Chen Yuting\footnote{School of Economics, Yokohama National University, 79-4 Hodogaya-ku, Yokohama, Japan \\Email: chen-yuting-bh@ynu.ac.jp }

}

\maketitle

\begin{abstract}

In this paper, I study to what extent countries differ in their preferences for quality and their technologies for improving quality. The paper also quantifies the contribution of those differences to the differences in gains from trade across countries. I adopt \textcite{ANTONIADES2015263}, which allows endogenous quality choice of firms. This paper extends that framework into multi-country and multi-sector setting. The new feature is that bilateral trade liberalization yields spillover effects on the competitiveness of another country. This paper structurally estimates the model under multi-country and multi-sector framework. It is found that richer countries generally have stronger valuation for quality. The quantification demonstrates that variations in the strength of quality preferences and in the technology of improving quality across countries add to their heterogeneities in market competitiveness. Subsequently, I simulate a 5\% increase in the trade barriers. If the quality channel is shut down, countries with stronger preferences for quality have larger degrees of underestimations in their losses from the trade barrier. Finally, gains from a universal rise in quality preference are unequal among countries, with larger economies generally gaining more than smaller economies.
\end{abstract}

\newpage
\section{Introduction}
Trade economists have paid particular attention to the role played by product quality in international trade. Under quality sorting, more efficient firms produce higher-quality goods, enter more competitive markets and charge higher prices,\footnote{See, for example, \textcite{baldwinharrigan2011}, \textcite{baller2015}, \textcite{doi:10.1093/restud/rdr030}, \textcite{Demir} and \textcite{ECKEL2015216}.} which enriches conventional efficiency-based trade theories.\footnote{Conventional heterogeneous firm trade theories include \textcite{melitz2003} and \textcite{MO2008}.} Additionally, quality differentiations across firms and markets can explain price variations in exports and imports.\footnote{See, for example, \textcite{Mandel2010HeterogeneousFA}.} Finally, the product quality is related to welfare gains from trade.\footnote{See, for example, \textcite{FLXY}}

This paper analyzes differences in quality preferences and technologies among countries and investigates the impact of these differences on gains from trade. Despite the growing literature linking quality, price and welfare, little has considered the contribution of quality either to cross-country differences in market competitiveness or to unequal gains from trade. In addition, this paper is the first to quantify the countries' heterogeneities in quality preferences and costs from both demand and supply sides. In this paper, I investigate this
problem using a combination of empirical, theoretical and structural analyses. Two findings emerge from this analysis: first, differences in quality preferences among countries enlarge the inequality in gains from trade; second, a universal rise in quality preferences brings higher gains to larger economies.

To provide micro-foundations, I use a recent Chinese firm-level export transaction dataset to explore firm-level evidence. Evidence from this dataset shows that firms charging higher prices sell to larger countries, earn higher revenues and enter more markets. The positive relationship between price, revenue, and market entry is strengthened in sectors of high research and development (R\&D) intensity. Typically, R\&D intensity can be viewed as a proxy for scope of quality differentiation in \textcite{mz2012}. Additionally, I choose two sectors with different levels of R\&D intensity - tobacco (R\&D intensity approximately zero) and pharmaceutical (R\&D intensity approximately 50\%). In the pharmaceutical sector, the positive relationship between price, revenue and market entry still holds. However, in tobacco sector, price is negatively correlated to revenue and the destination market size. The evidence reveals that the price is positively correlated destination market size, firm revenues and market entry in a subset of sectors only. This implies that price reveals quality only in selected industries and that some sectors are homogeneous with limited quality differentiation space.

Motivated by these stylized facts, I extend the framework of \textcite{ANTONIADES2015263} into the multi-country and multi-sector setting. The framework introduces quality differentiation into \textcite{MO2008} (MO, thereafter). According to \textcite{ANTONIADES2015263}, firms compete along the quality dimension as well as cost dimension so that more productive firms charge higher prices because they can produce higher-quality products. I also extend \textcite{ANTONIADES2015263} so that a firm can produce multiple products. The extension suggests that more productive firms produce more products with higher average qualities. On the demand side, quality preferences are homogeneous within a country and sector. On the supply side, the cost of improving quality is embedded in both variable and sunk costs that firms have to pay. 

In the multi-country setting in this paper, it is possible to assess spillover effects of the bilateral trade liberalization and the preference shock. The spillover effect is unique to this multi-country model. A bilateral trade liberalization between two countries can generate negative effect on competitiveness of other countries. The effect on the third country is higher, the larger the trade-liberalizing economies. Apart from that, a positive preference shock on one country could have positive effect on competitiveness of other countries. This occurs since quality differentiation scope in selling to that market are widened following positive preference shock. Firms in all countries respond by rising quality of goods sold to that market. The responses are higher for more productive firms. Therefore, the productivity threshold of entering the market with preference shock rises. As for other countries without shocks, the least productive firms cannot profitably export to the market with the positive preference shock. This makes this group of firms unable to cover the sunk cost of entry. Thus, the least productive firms are driven out of the market, in all countries. The productivity thresholds of other countries without preference shock could also rise. The larger the country with preference shock, the higher the effects on productivity of other countries. Finally, if the trade liberalization negatively affect competitiveness of a third country, higher quality preference alleviate the negative effect. If the trade liberalization positively affect competitiveness of a third country, higher quality preference strengthen the positive effect. This occurs since higher quality preference have positive selection effect. 

The multi-country setting of this paper enables computation of equilibrium. I bring the model to data to estimate the relevant parameters of quality in the model. An extended gravity model can be derived from the above theoretical framework. In addition to the bilateral trade costs and the importer/exporter fixed effects typically included in a gravity model estimation,\footnote{The original gravity model starts from \textcite{tinbergen} and is augmented by \textcite{AvW2003}.} I include (a non-linear function of) quality preferences and costs into the the gravity equation. According to this gravity model, bilateral trade flow is positively affected by the destination country's quality preferences and negatively affected by the origin country's costs of improving quality.\footnote{The details of this are discussed in the Section \ref{theory} and Section \ref{emprics}.} I parameterize that gravity equation to estimate quality preferences/costs parameters as well as trade costs. After the estimation, I recover two measures of the degree of competition for each country-sector pair. The first measure is endogenous competitiveness. This is the productivity threshold above which firms can make non-negative profits. Thus, before a firm makes a random draw of productivity, a firm will be, ex ante, less likely to enter a market where the cutoff productivity is high and competition is fierce. The second is the exogenous measure of competitiveness. This is the fundamental productivity level in a market.  Firms are more likely to draw a high productivity if the level of fundamental productivity is higher. I follow \textcite{CGMO2011} and estimate the endogenous competitiveness first and back out the exogenous competitiveness using free-entry conditions.\footnote{This is discussed in detail in Section \ref{emprics}.} 

The structural analysis suggests that quality preferences are related to income and that considering quality enlarges differences in competitiveness across countries. First, the preference for quality is positively correlated with GDP per capita in more than half of the sectors studied.\footnote{For the rest of sectors, GDP per capita has insignificant effect on quality preferences.} This indicates that for these industries, the preferences for quality are on average stronger in richer countries. Second, the distribution of preferences for quality vary across sectors. The variance of the estimates across countries is particularly large for certain sectors such as HS 85 (Electrical machinery and equipment), which is 56.19, and small for certain sectors such as HS 37 (Photographic or cinematographic goods), which is approximately 9.94.\footnote{As for the cost of quality, it is found that the mean is smaller than the preference parameters within the same sector and the variance are also smaller.} These findings are, to some extent, consistent with empirical evidence on the correlation between price and market size.\footnote{\textcite{Syversion07} among other, reports a negative correlation between market size and output price, while \textcite{Verhoogen08} and \textcite{hallak2009firms} report positive correlations.} For example, \textcite{amit} and \textcite{KY2016} finds that prices are a good proxy for quality in some industries but not in others. I compare the endogenous competitiveness estimated from this model with that estimated from the MO canonical model. Comparing across countries by sector, it is noticeable that endogenous competitiveness is more dispersed under this model. This more dispersed distribution can be explained by the additional heterogeneity introduced by differentials in both quality preferences and costs of improving quality among countries.

A counterfactual exercise is performed to examine the competitiveness-enhancing effect of quality. Exogenous competitiveness  and quality parameters are kept constant. I experiment with a 5\% universal increase in international trade costs. This exercise is implemented under both this model and the MO model. In general, the productivity cutoffs, which measure the endogenous competitiveness, decrease in both models. This implies that the average productivity is lower and the economy becomes less competitive. Quantitatively, the decline of the productivity cutoff is 25\% to 300\% larger in this model with quality than that it is in the MO setup (depending on the sector). Furthermore, the differences between this model and the MO model in degrees of the declines of productivity cutoffs rise in the magnitude of preferences for quality, in most sectors. If consumers value quality more in some countries, the difference between the rise of productivity cutoffs in this model and in the MO model will be larger in these countries than the difference in other countries where quality is not valued. This suggests that gains from trade across countries are more heterogeneous than in a canonical model, in which quality is considered.

For the second counterfactual exercise, I simulate a universal positive shock in preference for quality. Although quality preferences are exogenous to this model, they can change over time by promoting consumers' awareness of product quality. Most countries experience gains in productivity of less than 141.2\%. Approximately 25 countries experience a loss in productivity. The total gains in productivity is nearly 5 times the total losses in productivity, from the global scale. It is noticeable that large economies generally gain more than small economies. On average, 1\% rise in population size leads to 7.876\%  more gains in productivity. This finding is consistent with the model implication that larger economies have wider scope for quality differentiation. 

Quality in international trade has been intensively studied, but existing work pays limited attention to sources of quality variation from both the demand and the supply side and does not link them to heterogeneity among countries. \textcite{HALLAK2006238} estimates the demand for quality across countries on the demand side. Other studies focus more on firms' behavior in quality improvement (\textcite{BAS2015250}, among others). Nevertheless, \textcite{ANTONIADES2015263} reconciles both the demand and the supply sides into a single theoretical framework. This paper uses a more general framework and compares predictive differences between that quality-extended model and the MO model. 

The analysis of this paper provides new insight into the sources of unequal gains from trade from a quality perspective. That countries do not gain equally has been well documented. Gains from trade can be divergent among countries of different sizes (\textcite{MARKUSEN1981531}, among others) or incomes (\textcite{tw90}, among others). A recent study by \textcite{ANDERSON2016279} suggests that free trade agreement (FTA) can bring -0.3\% to 5\% gains to different countries. However, little literature provides explanations for the unequal gains across countries. The paper points out and quantifies the strengthening effect on competition brought by higher quality preferences and quantifies it. 


The paper relates to three strands of literature: heterogeneous firms' quality choices, pricing-to-market and non-homothetic preferences. Trade liberalization induces firms' quality-upgrading behavior (\textcite{FLY2015}, among others). Larger plants have higher output prices and use more expensive materials \parencite{KV2012}.\footnote{Other studies include \textcite{mz2012}.} This paper embraces the above findings by allowing quality choices to be endogenous to firms such that prices and markups vary with plant size and productivity. The model also includes findings by \textcite{dfr2014} that export unit values vary to a larger extent than quality-adjusted prices.\footnote{\textcite{abaer} finds that price deviate from purchasing power parity.} Finally, the study, by estimating quality preferences and relating to the country-level income, incorporates non-homothetic preferences proposed by \textcite{fgh}.\footnote{Similarly, \textcite{fieler} proposes that richer households consume higher-elasticity goods more thus leading to more trade flows between high income countries than between low-income ones.}

The remainder of the paper is organized as follows. Section \ref{stylizedfact} uses firm-product level trade statistics to explore the relationship between price, market features and firm exporting performances. I use these statistics to show the micro foundations of the model. Section \ref{theory} lays out the theoretical framework under a multi-country and multi-product setting. Section \ref{emprics} discusses the data and how relevant parameters are estimated. Section \ref{counterfactual} implements the counterfactual analysis. Finally, Section \ref{conclude} concludes.

\section{Stylized Facts}\label{stylizedfact}
\subsection{Data}
In order to document stylized facts regarding \textit{f.o.b.} export prices across destinations and across firms within the same destination, I use a unique micro-level data and two sets of macro-level data. The specific micro dataset used here is the China Customs Trade Statistics (CCTS) issued in 2013 by The General Administration of Customs. The advantage of this dataset is that it is highly disaggregated in recording the import/exports of Chinese firms. Additionally, it records the prices of exports/imports, origins/destinations of imports/exports, the product (HS8) of the transaction, and the mode (ordinary or processing) of international trade. I use the CEPII-Dist dataset to obtain bilateral characteristics of destination countries and China. The Penn World Table dataset is used to obtain population of destination countries.

\subsection{Empirical Findings}
In this subsection, I report three stylized facts concerning export prices across destinations and across firms within destinations. Although the existing literature such as \textcite{mz2012} has documented these findings, it is important to show that they hold in the more recent years. Moreover, these are the facts to be embraced in the model and that lay the micro foundation in the model setup.

\textit{On export prices across destinations} — Based on the entire customs export data in 2013, Table \ref{style1} reports the regression results using (log) export prices as the dependent variable and destination country’s population as the main explanatory variable, controlling for destination’s GDP per capita and distance to China. Columns 1-2 use the prices at the firm-HS8-country level and Columns 3-4 use HS8-country level. The coefficients on (log) population in all specifications are significantly positive, suggesting that export prices increase with the destination’s market size, consistent with \textcite{mz2012}. Thus, it is summarized as the following fact:\\
\textbf{Stylized Fact 1:} \textit{On average, firms set higher export prices for the same product in larger markets.}

\textit{On export prices and revenues} - Table \ref{style2} presents robust evidence that firms charging higher export prices earn greater revenues even within very narrowly defined destination-product markets. In Columns 1-2, I use prices at the firm-HS8-country level and Columns 3-4 I use prices at firm-product level. This relationship is highly statistically significant. Importantly, it is also markedly stronger for goods with greater scope for quality upgrading, as proxied by sectors’ R\&D intensity compiled by \textcite{KROSZNER2007187}. The magnitudes and signs of estimated coefficients are relatively robust to different specifications and different level of aggregations. The elasticity of export prices with respect to revenues is 0.15. A doubling in firm sales in a given market is thus associated with 20\% higher bilateral unit prices for the average product. That number is bigger for sectors with higher R\&D intensity. This yields the following fact:\\
\textbf{Stylized Fact 2:} \textit{On average, firms charge higher prices simultaneously earn greater revenues in each destination. The correlation between price and revenue is higher in R\&D intensive sectors.}

\textit{On prices and market entry} -As reported in Table \ref{style3}, exporters that supply more countries systematically charge higher average prices (Columns 1-2). Firms selling to more destinations also exhibit greater price dispersion across importers (Columns 3-4). In this table, I use prices (or price dispersions) at firm-HS8 product level. These results are both largely enhanced by products with substantial potential for quality differentiation. As columns (2) and (4) show, the patterns are stronger for sectors with larger scope for quality differentiation, which is proxied by sectoral R\&D intensity, defined as before. The finding is consistent with \textcite{MANOVA2017116} and \textcite{mz2012} which use earlier versions of CCTS. This finding suggests the following fact:\\
\textbf{Stylized Fact 3:}\textit{ On average, exporters entering more destinations and offering a wider range of export prices charge higher prices. The correlation between price and market entry is higher in R\&D intensive sectors.}

One caveat is that the above relationships can be sector-dependent. The above relationship exist only in a subset of sectors. In Figure \ref{fig:s1} and \ref{fig:s2}, I plot the above 3 stylized facts for sector HS 24 (Tobacco) and HS 6, respectively. The findings of the three stylized facts continue to hold in the sub-sample of HS 6, while they are reversed in HS 24. From Figure (\ref{fig:stylized1_hs24}), the average price decreases with destination market size. Firms charging higher prices earn lower revenues, as in Figure (\ref{fig:stylized2_hs24}). The relationship between the firm's price charged and the number of market it entered (Figure (\ref{fig:stylized3_hs24})) is not as significant as in sector HS 6 (\ref{fig:stylized3_hs30}). The comparison suggests that competition improves quality in only a subset of sectors. This is also consistent with the Table (\ref{style2}) -(\ref{style3}) where the relationship between price, revenue and market entry are "more positive" in R\&D intensive sectors. In sector 30 (Pharmaceuticals), the average R\&D intensity is approximately 58\% while that ratio of sector 24 is 0. 
\section{Theory}\label{theory}
I lay out theoretical frameworks that incorporate endogenous quality choices in the linear demand system. The demand side features heterogeneous quality preferences among different countries; In addition, quality upgrading requires additional variable and fixed costs such that producers endogenously choose the level of quality in their products. The model is built on a multi-country basis for the convenience of quantification in subsequent empirical studies. For simplicity, sector notations are dropped as there are no interactions between them.
\subsection{Setup}
I follow the framework of \textcite{ANTONIADES2015263} and \textcite{foster08}, the preference of a consumer in country $j$ is represented as (\ref{consumer_utility}). This utility function has the advantage that the price/quantity are linear in quality preference
\begin{eqnarray}\label{consumer_utility}
U_{j} = &q^{c}_{0}& + \sum_{s}\left[\alpha_{s}\int_{\omega\in \Omega_{s}}q_{s}(\omega)^{c}d\omega + \kappa_{js} \int_{\omega \in \Omega_{s}} z_{s}(\omega)q_{s}(\omega)^{c}d\omega\right] \nonumber\\ && -\sum_{s}\left[\frac{1}{2}\gamma_{s}\int_{\omega\in\Omega_{s}}(q_{s}(\omega)^{c})^{2}d\omega + \frac{1}{2}\eta_{s}\left(\int_{\omega\in\Omega_{s}}q_{s}(\omega)^{c}\right)^{2}\right]
\end{eqnarray}
where $q_{0}^{c}$ and $q_{s}(\omega)$ are individual $c$'s consumption in numeraire and differentiated goods (in sector $s$) respectively. The quality of each variety $\omega$ is given by $z_{s}(\omega)$ and the taste for quality by consumers in country $j$ in sector $s$ is $\kappa_{js}$. In the above utility function, $\kappa_{js}z_{s}(\omega)$ is equivalent to the variety-specific taste shifter in \textcite{foster08}.\footnote{The difference between this paper and \textcite{foster08} is that the latter regard the shifter as the idiosyncratic term.} If qualities of all varieties are zero or if the taste parameters are zero, the model becomes traditional \textcite{MO2008} model. Finally, $\alpha_{s}$ and $\eta_{s}$ capture the degree of substitution between each variety and the numeraire. The parameter $\gamma$ measures the degree of differentiation between varieties. $\kappa_{js}$ picks up the degree of preference for variety. Specifically, $\kappa_{js}$ are assumed to be positive and differ across countries and sectors.

This generates a linear inverse demand function that depends on both quantity and quality:
\begin{equation}\label{indemand}
p_{j}(\omega) = \alpha -\gamma q(\omega)^{c} +\kappa_{j}z(\omega)-\eta Q^{c}
\end{equation}
where $Q^{c} = \int_{\omega \in \Omega} q(\omega)^{c} d\omega$.
By re-arranging (\ref{indemand}) one can obtain the total demand in country $j$ of each variety:
\begin{equation}
q_{j}(\omega) = L_{j}q(\omega)^{c} = \frac{\alpha L_{j}}{\eta N_{j}+\gamma} -\frac{L_{j}}{\gamma} p_{j}(\omega) + \frac{L_{j}\kappa_{j}}{\gamma}z(\omega)+ \frac{\eta N_{j} L_{j}}{\gamma\left(\eta N_{j} + \gamma\right)}\bar{p}_{j} -\frac{\eta N_{j}L_{j}\kappa_{j}}{\gamma\left(\eta N_{j}+\gamma\right)}\bar{z}_{j}
\end{equation}
where $\bar{p}_{j} = \frac{1}{N_{j}}\int_{\omega \in \Omega} p_{j}(\omega) d\omega$ is the mean price over all varieties in country $j$, $\bar{z}_{j} = \frac{1}{N_{j}} \int_{\omega \in \Omega} z_{j}(\omega) d\omega$ is the mean quality over all varieties in $j$. This form of preference over quantity and quality ensures that demand function is linear in price and quality. 

On the supply side, firms have to incur costs in both production and quality improvements. For simplicity, labor is assumed to be the only factor of production. As with most heterogeneous firm trade models, firms have to pay a sunk cost $f_{e}$ before they draw the productivity $c$ (marginal cost of production). If a firm in country $i$ produce and sell to country $j$, it should choose both the quantity and the quality of its product, given the total cost function:
\begin{equation}\label{totalcost}
TC_{ij} (\omega) = q_{ij}(\omega)(\tau_{ij}c + \mu_{i}z_{ij}(\omega)) + \delta z_{ij}(\omega)^{2}  
\end{equation}
where $\tau_{ij}$ is the ice-berg trade cost from country $i$ to $j$ and $z_{ij}$ is the quality level of goods sold from $i$ to $j$. 
The first term implies that firms have to incur additional marginal cost of quality upgrading,which depends on both quality level $z_{ij}$ and quality improving technology at origin country $\mu_{i}$, when producing in country $i$; the second term picks up the fixed cost of improving quality. To generate closed form solutions, I assume that the fixed cost is a quadratic function of quality level. A more generalized function form in the fixed cost can be found in \textcite{FLY2015} in which CES preference is assumed. In the cost function (\ref{totalcost}), the marginal cost of upgrading quality differs across countries/sectors, since labor productivity are heterogeneous across countries. The multiplier in the fixed cost $\delta$ is assumed to be identical across countries, since the fixed costs of innovation across countries, which include high-skilled workers, tend to be homogeneous across countries.\footnote{This assumption is made since the highly skilled labor migrate to regions where wages are high, according to \textcite{migration_wage}. Therefore the wage for high skilled workers tend to converge over time.}
\subsection{Firms' Problem}
For a firm (of marginal cost $c$) from country $i$ selling to each (potential) market $j \in J$, it chooses quantity and quality to maximize profit:
\begin{equation}
\pi_{ij}(c) = \max_{q_{j}(\omega), z_{j}(\omega)} p_{ij}(\omega) q_{ij}(\omega) - TC_{ij}(\omega)
\end{equation}
Given the demand structure, the quantity of goods sold from country $i$ to country $j$ satisfies a linear function with respect to price $p_{ij}$: 
\begin{equation}\label{linear_demand}
q_{ij} (c, z) = \frac{L{j}}{\gamma}\left[c_{D}^{j}- p_{ij}(c, z) + \kappa_{j}z_{ij} \right]
\end{equation}
where $c_{D}^{j}$ is the cost threshold so that firms with marginal cost of production $c\leq c_{D}^{j}/\tau_{ij}$ can sell to market $j$ and make non-negative profits. Equation (\ref{linear_demand}) shows that demand is linear to both price and quality. This  makes it convenient for later analysis in  optimal pricing.

Given the demand function and the firms' problem, the optimal price and quantity as a function of quality $z$ is linear in marginal cost\footnote{Per unit trade cost  $\left(\tau_{ij} -1 \right)c$ does not depend on quality}:
\begin{subequations}
\begin{align}
p_{ij}(c, z) & = \frac{1}{2}(c_{D}^{j}+\tau_{ij}c) +\frac{1}{2}(\kappa_{j}+ \mu_{i})z_{ij}\\
q_{ij}(c, z) & = \frac{L_{j}}{2\gamma}(c_{D}^{j}-\tau_{ij}c) +\frac{L_{j}}{2\gamma}(\kappa_{j} - \mu_{i})z_{ij}
\end{align}
\end{subequations}
where $c_{D}^{j}$ is the cost-cutoff in country $j$, above which firms cannot profitably sell in market $j$. This is also the endogenous competitiveness of country $j$ (lower cutoff, more endogenous competition), which changes with factors such as ice-berg trade costs, quality parameters and market size.
Thus, that the profit of selling to market $j$ can be re-written as
\begin{equation}
\pi_{ij} (c, z) = \frac{L_{j}}{4\gamma}\left[(c_{D}^{j} -\tau_{ij}c) +(\kappa_{j}-\mu_{i})z_{ij}\right]^{2} - \delta z_{ij}^{2}
\end{equation}
Solving for the first-order conditions for profit, the optimal quality choice to sell from $i$ to $j$ of firm $c$ is written as a decreasing function of marginal cost so that more productive firms offer higher quality within a market
\begin{equation}\label{quality}
z_{ij}(c) = \rho_{ij}  (c_{D}^{j} - \tau_{ij}c)
\end{equation}
where $\rho_{ij} = L^{j}(\kappa_{j} - \mu_{i})/ \left(4\delta\gamma - L^{j}(\kappa_{j}-\mu_{i})^{2}\right)$. As Equation (\ref{quality}) shows, the quality is linearly related to marginal cost $c$. Consistent with other international trade in quality literature,\footnote{For example, see \textcite{FLXY} in which more productive firms offer higher quality products with lower quality-adjusted costs.} for the same firm $c$ from country $i$, it is expected to offer higher quality goods to larger markets and  countries with higher taste for quality. Higher quality is also expected from countries with higher labor productivity in quality improvement (i.e., lower $\mu_{i}$). It is noticeable that at the first glance, a lower cutoff resulting from trade liberalization seems to have negative impact on quality updates. However, a lower cutoff indicates that the average marginal cost is also lower, since firms with marginal cost above the new cutoff exit the market. Since firm-level quality $z(c)$ decreases with $c$, the average quality of products sold in the economy rises. That effect is higher the larger quality differentiation scope $\rho_{ij}$.

Given optimal quality function in (\ref{quality}) and price/quantity functions, it is evident that price and quantity can be further expressed as 
\begin{subequations}\label{pq}
\begin{align}
p_{ij}(c, z) & = \frac{1}{2}(c_{D}^{j}+\tau_{ij}c) +\frac{1}{2}(\kappa_{j}+ \mu_{i})\rho_{ij}(c_{D}^{j} - \tau_{ij}c)\\
q_{ij}(c, z) & = \frac{L_{j}}{2\gamma}(c_{D}^{j}-\tau_{ij}c) +\frac{L_{j}}{2\gamma}(\kappa_{j} - \mu_{i})\rho_{ij}(c_{D}^{j} - \tau_{ij}c)
\end{align}
\end{subequations}
So the for each firm with marginal cost of production $c$, the operating revenue of exporting from country $i$ to $j$ is:
\begin{equation}\label{revenue}
r_{ij} (c) = \frac{\rho_{ij} \left((c_{D}^{j})^{2} - (\tau_{ij}c)^{2}\right) + \delta\rho_{ij}^{2}(\kappa_{j}+\mu_{i})(c_{D}^{j} -\tau_{ij}c)^{2}}{\kappa_{j}-\mu_{i}}
\end{equation}
Equation (\ref{pq}) and (\ref{revenue}) implies that other things equal, more productive firms charge higher prices and earn greater revenues from each destination, if $(\kappa_{j}+\mu_{i})\rho_{ij}>1$. This is consistent with Stylized 2 in the prior section. 

The operating profit can be expressed as\footnote{Following \textcite{ANTONIADES2015263}, the fixed cost of upgrading quality is also regarded as sunk since this model is not dynamic.}:
\begin{equation}\label{profit}
\pi_{ij} = \frac{L^{j}}{4\gamma} \left[ 1+ (\kappa_{j}-\mu_{i})\rho_{ij}\right](c_{D}^{j} - \tau_{ij}c)^{2} 
\end{equation}
Profit function in (\ref{profit}) implies that countries of higher quality preferences offer higher profits, all else being equal. This is related to the third stylized facts. From (\ref{profit}), more productive firms earn higher profits and earn non-zero profits in a larger set of markets.

This profit function is a quality-adjusted form of \textcite{MO2008}, in which profits are magnified by a term larger than one. This quality-adjusted term is key to the later empirical analysis. 
\subsection{Equilibrium}
In equilibrium, the free entry condition of each country implies that the expected profit (earned from selling to each market) of a firm is zero. Such that that 
\begin{equation}\label{fe}
\sum_{j\in J} \int_{0}^{c_{D}^{J (c)}/\tau_{ij}} \pi_{ij} (c) dG_{i}(c) = f_{E}^{i}
\end{equation}
where $f_{E}$ is the sunk cost a firm have to pay prior to entry. It is assumed to be positive and can vary across countries and sectors. Equation (\ref{fe}) implies that operating profit from all markets are expected to merely compensate for the sunk cost. 

Following prior literature, it is assumed here that marginal cost draws follow Pareto distribution: $G_{i}(c) = \left(\frac{c}{c_{M}^{i}}\right)^{k}$ so that firms draw the marginal cost $c$ from the range $\left[0, c_{M}^{i}\right]$. The term $c_{M}^{i}$ is the exogenous competitiveness in market $i$. Lower values indicates that a firm in country $i$ is more likely to draw a lower cost. The power $k$ governs the Pareto distribution dispersion: higher $k$ implies that the distribution of $c$ is more concentrated. 

Given profit as in (\ref{profit}), one can re-write the equilibrium condition in (\ref{fe}) for country $i$ so that the summation of expected profit from each market equals the sunk cost paid.
\begin{equation}
\sum_{j\in J}\frac{L^{j}}{4\gamma} \left[ 1 +\left(\kappa_{j}-\mu_{i}\right)\rho_{ij}\right]  \tau_{ij}^{-k} (c_{D}^{j})^{k+2} = \frac{(k+1)(k+2)\gamma f_{e}}{2} \left(c_{M}^{i}\right)^{k}
\end{equation}
The same equilibrium condition can be written for each country $h$, in the matrix form. In the multi-country case, the vector of cost threshold $c_{D}$ for each country $j$ satisfies:
\begin{equation}\label{fe_matrix}
\mathbf{B}*\mathbf{L}*\mathbf{c_{D}}^{k+2} = 2\gamma\left(k+1\right)\left(k+2\right)\mathbf{f_{e}}\mathbf{c_{M}}^{k}
\end{equation}
where matrix $B$ is the the "quality adjusted" matrix of trade freeness. If there are a total of $J$ countries, the dimension is $J\times J$. This matrix is represented as the following form:

\[ 
\textbf{B} =\left[
\begin{array}{cccc}

    \left[ 1 +\left(\kappa_{1}-\mu_{1}\right)\rho_{11}\right] \tau_{11}^{-k}        & \left[ 1 +\left(\kappa_{2}-\mu_{1}\right)\rho_{12}\right] \tau_{12}^{-k}  & \dots & \left[ 1 +\left(\kappa_{J}-\mu_{1}\right)\rho_{1J}\right] \tau_{1J}^{-k} \\
    
    \left[ 1 +\left(\kappa_{1}-\mu_{2}\right)\rho_{21}\right] \tau_{21}^{-k}       & \left[ 1 +\left(\kappa_{2}-\mu_{2}\right)\rho_{22}\right] \tau_{22}^{-k} & \dots & \left[ 1 +\left(\kappa_{J}-\mu_{2}\right)\rho_{2J}\right] \tau_{2J}^{-k} \\
    
    \dots\\
    
     \left[ 1 +\left(\kappa_{1}-\mu_{J}\right)\rho_{J1}\right] \tau_{J1}^{-k}        & \left[ 1 +\left(\kappa_{2}-\mu_{J}\right)\rho_{J2}\right] \tau_{J2}^{-k}    & \dots & \left[ 1 +\left(\kappa_{J}-\mu_{J}\right)\rho_{JJ}\right] \tau_{JJ}^{-k}
\end{array}\right]
\]

It is noticeable that whether elements in $B$ are larger or smaller than one depends on the trade freeness, $\kappa$'s and $\mu$'s. If trade cost is small (i.e., $\tau_{ij}^{-k}$ is larger) and $\kappa_{j}$ is large, $i, jth$ element can exceed one. 

The matrix $L$ is of dimension $J \times J$ and is the diagonal matrix of market size $L_{i}$ for each country $i$:

\[ 
\textbf{L} =\left[
\begin{array}{cccc}

    L_{1}      & 0 & \dots & 0 \\
    
   0      & L_{2} & \dots & 0 \\
    
    \dots\\
    
     0        & 0   & \dots & L_{J}
\end{array}\right]
\]

The last term on the right-hand-side of (\ref{fe_matrix}) $\mathbf{c_{D}}^{k+2}$ is the vector of cutoffs (to the power of $k+2$), with the $ith$ element being $(c_{D}^{i})^{k+2}$. On the right hand side of Equation (\ref{fe_matrix}), $c_{M}^{k}$ is the vector of exogenous competitiveness (to the power $k$), with $ith$ element being $c_{M}^{i, k}$. Both $c_{D}^{k}$ and $c_{M}^{k}$ are of dimension $J\times 1$.

By inverting matrix $B$, cost threshold $c_{D}^{j}$ for each country $j$ can be computed as follows: 
\begin{equation}\label{cutoff}
\left(c_{D}^{j}\right)^{k+2} = \frac{2\gamma(k+1)(k+2)f_{e}}{\mid B\mid} \frac{\sum_{i}\mid C_{ij}\mid\left(c_{M}^{i}\right)^{k}}{L_{j}}
\end{equation}
where $\mid B \mid$ is the determinant of matrix $B$ and $C_{ij}$ is the cofactor of $B_{ij}$. A closer examination of the matrices and Equation (\ref{cutoff}) implies that higher market size $L_{j}$ leads to a lower calculated cutoff, all else being equal.\footnote{This also holds under multi-product setting in Appendix \ref{multiproduct}.}
If ice-berg trade cost is symmetric, i.e. $\tau_{ij} = \tau_{ji}$, bilateral trade liberalization can have heterogeneous effects on cost cutoffs of different countries. Following this argument, bilateral trade liberalization can have negative externalities on countries without trade liberalization.
\begin{prop}
If ice-berg trade cost is symmetric, i.e. $\tau_{ij} = \tau_{ji}$, the effect of a bilateral trade liberalization between $i$ and $j$ on cutoffs is not universal i.e. $\exists$ $k$, $k'$ $\in \lbrace 1,...,J \rbrace$ such that $\frac{\partial c_{D}^{k}}{\partial \tau_{ij}} >0$ and $\frac{\partial c_{D}^{k'}}{\partial \tau_{ij}} <0$. Additionally, $\exists$ $\textbf{B}$ such that $\frac{\partial c_{D}^{i}}{\partial \tau_{ij}}$, $\frac{\partial c_{D}^{j}}{\partial \tau_{ij}}$ $>0$ and $\frac{\partial c_{D}^{k}}{\partial \tau_{ij}} <0$, $\forall k\neq i, j$. The larger the trade liberalizing economies, the higher the effects on the cutoff of country $k \in \lbrace 1,...,J \rbrace$. 
\end{prop}
\begin{proof}
See Appendix \ref{apptwocountry}
\end{proof}
The intuition for the negative externalities is comparable to the special case of two-country world. It is driven by the long-term entry.In the long run, the market entries in trade-liberalizing economies rise. This drives up the competitiveness, i.e. the cost cutoff decrease as a result. For other countries without trade liberalization, the profit from selling to these markets are higher, relative to countries with trade liberalization. This results in decreased entries and higher cost threshold.

Another implication from (\ref{cutoff}) is on the effect of a preference shock, i.e. a change in $\kappa_{j}$, on cutoffs of countries. 

\begin{prop}
The effects of a shock in preference for quality $\kappa_{j}$ of any country $j$ on cutoffs are universal, i.e. $\exists$ $\textbf{B}$ such that $\frac{\partial c_{D}^{k}}{\partial \kappa_{k}} <0$, $\forall$ $k$. The larger the country with preference shock, the higher the effects on the cutoff of country $k \in \lbrace 1,...,J \rbrace$
\end{prop}

\begin{proof}
See Appendix \ref{apptwocountry}
\end{proof}

The intuition for this lies in the selection effect. Quality differentiation scope in selling to that market are widened following positive preference shock. Firms in all countries respond by rising quality of goods sold to that market. The responses are higher for more productive firms. Therefore, the productivity threshold of entering the market with preference shock rises. As for other countries without shocks, the least productive firms cannot profitably export to the market with the positive preference shock. This makes this group of firms unable to cover the sunk cost of entry. Thus, the least productive firms are driven out of the market, in all countries. The productivity thresholds of other countries without any preference shock also rise.

Additionally, the average bilateral \textit{f.o.b} price $\bar{p}_{ij}$ and trade value $r_{ij}$ can be computed by aggregating over $c\in \left[0, c_{D}^{i}\right]$\footnote{Trade value is computed based on \textit{f.o.b} price.}

\begin{equation}\label{price}
\bar{p}_{ij} = \frac{1}{2} \left[ \frac{2k+1}{k+1} + \frac{1}{k+1}\left(\kappa_{j}+\mu_{i}\right)\rho_{ij}\right]\frac{c_{D}^{j}}{\tau_{ij}}
\end{equation}
This is related to the first stylized fact, which links price with market size. From (\ref{price}), one can observe that market size affect price both by decreasing the cost cutoff $c_{D}^{j}$ and by increasing the scope for quality differentiation. If the latter dominates, the relationship between price and market size is similar to Figure (\ref{fig:stylized1_hs30}), otherwise, it is closer to (\ref{fig:stylized1_hs24}). 

The total trade value from $i$ to $j$  
\begin{equation}\label{tradeflow}
r_{ij} = \frac{kN_{i}^{E}(c_{m}^{i})^{-k}}{2\gamma} L^{j}(\tau_{ij})^{-(k+1)} \left(c_{D}^{j}\right)^{k+2}\left[1+(\kappa_{j}-\mu_{i})\rho_{ij}\right] \left(\frac{1}{k(k+2)}+\frac{(\kappa_{j}+\mu_{i})\rho_{ij}}{k(k+1)(k+2)}\right)
\end{equation}
where $N_{i}^{E}$ denotes the expected number of entrants in country $i$. From Equation (\ref{price}) and (\ref{tradeflow}), the quality components have two effects: higher $\mu_{i}$ raises prices and reduces the trade values because of the additional (variable) costs incurred in manufacturing higher-quality goods in the origin country. Higher quality preferences in the destination country raise the willingness to pay and thus average price and trade value rise. Another countervailing effect is that, higher quality preferences drive down the cost cutoffs in the importing country thus to some extent lowering prices and values. 

Finally, one can compute other aggregate variables, which are the number of varieties, entries, price indexes and welfare. The expected number of varieties in a country $j$ (i.e., the number of producers from domestic and abroad servicing market $j$) is similar to the M.O model multiplied by a quality-adjustment term that is negatively related to $\kappa$ (of the market $j$) and positively affected by $\mu$ (of each sourcing country). That can be characterized as:

\begin{equation}\label{varieties}
N_{j} = \frac{2(k+1)\gamma}{\eta}\frac{\alpha-c_{D}^{j}}{c_{D}^{j}} \frac{1}{\overline{1+(\kappa_{j}-\mu)\rho}}
\end{equation}
where $\overline{\left[1+(\kappa_{j}-\mu)\rho\right]} = \lbrace \left[1+(\kappa_{j}-\mu_{1})\rho_{1j}\right]N_{1j} + \left[1+(\kappa_{j}-\mu_{2})\rho_{2j}\right]N_{2j} + ... \rbrace / N_{j}$. $N_{kj}$ denotes the number of varieties sold from country $k$ to country $j$. Equation (\ref{varieties}) indicates that a higher preference for quality can be countervailing: $\kappa$ can directly lower the number of varieties because it imposes higher requirements for the quality firms offer to market $j$ such that fewer firms can meet the high-quality requirement. On the other hand, higher quality preferences lower the cutoffs as is previously argued, which then encourages entry by raising average profits. 

One can solve for the number of entrants by noting that the number of producers from origin to destinations depends on the exogenous competitiveness of the origin country and the cutoff in the destination country. Equation (\ref{varieties}) can be used to solve for the number of entrants in each country. Notice that the number of bilateral varieties $N_{ij}$ satisfies the following condition
\begin{center}
$N_{ij}=N_{i}^{E}*G(c_{D}^{ij})=N_{i}^{E}\tau_{ij}^{-k}c_{j}^{k}(c_{M}^{i})^{-k}$
\end{center}
For each country $j$, the total number of varieties from each sourcing country should equal Equation (\ref{varieties}), which can be represented in the following condition
\begin{equation}
\sum_{i\in J} N^{E}_{i}\left[1+(\kappa_{j}-\mu_{i})\rho_{ij}\right](\tau_{ij}c_{M}^{i})^{-k} =\frac{2(k+1)\gamma}{\eta}\frac{\alpha-c^{j}_{D}}{c^{j,k+1}_{D}}
\end{equation}

Under the multi-country set-up, refer to M.O model, re-arranging (\ref{varieties}) leads to the following condition:
\begin{equation}\label{entrants}
D*c_{M}^{k}*N^{E}=F
\end{equation}
If there are $J$ countries, $D$ has the dimension of $J\times J$ and is the transpose of matrix $B$:

\[ 
D =\left[
\begin{array}{cccc}

    \left[ 1 +\left(\kappa_{1}-\mu_{1}\right)\rho_{11}\right] \tau_{11}^{-k}        & \left[ 1 +\left(\kappa_{1}-\mu_{2}\right)\rho_{21}\right] \tau_{21}^{-k}  & \dots & \left[ 1 +\left(\kappa_{1}-\mu_{J}\right)\rho_{J1}\right] \tau_{J1}^{-k} \\
    
    \left[ 1 +\left(\kappa_{2}-\mu_{1}\right)\rho_{12}\right] \tau_{12}^{-k}       & \left[ 1 +\left(\kappa_{2}-\mu_{2}\right)\rho_{22}\right] \tau_{22}^{-k} & \dots & \left[ 1 +\left(\kappa_{2}-\mu_{J}\right)\rho_{J2}\right] \tau_{J2}^{-k} \\
    
    \dots\\
    
     \left[ 1 +\left(\kappa_{J}-\mu_{1}\right)\rho_{1J}\right] \tau_{J1}^{-k}        & \left[ 1 +\left(\kappa_{J}-\mu_{2}\right)\rho_{2J}\right] \tau_{2J}^{-k}    & \dots & \left[ 1 +\left(\kappa_{J}-\mu_{J}\right)\rho_{JJ}\right] \tau_{JJ}^{-k}
\end{array}\right]
\]

The second term on the left hand side $c_{M}^{-k}$ is the diagonal matrix of $c_{M}^{i,-k}$ and is also dimension of $J\times J$ written as 
\[ 
c_{M}^{k} =\left[
\begin{array}{cccc}

    c_{M}^{1,-k}      & 0 & \dots & 0 \\
    
   0      & c_{M}^{2,-k} & \dots & 0 \\
    
    \dots\\
    
     0        & 0   & \dots & c_{M}^{J,-k}
\end{array}\right]
\]

and the third term $N^{E}$ denotes the vector of entrants with dimension $J*1$. The $i-th$ element is $N^{E}_{i}$.

On the right-hand-side, matrix $F$ is of dimension $J*1$ and the $i-th$ element is
\begin{center}
$\frac{2\gamma(k+1)(\alpha-c_{D}^{i})}{\eta c_{D}^{i,k+1}}$
\end{center}

The average price in a market is obtained by aggregating over the average delivered price over destinations. Using the bilateral price in (\ref{price}), it is implied that the (weighted) expected price levels in country $i$ are as follows:
\begin{equation}\label{average_price}
p_{i} = \frac{2k+1+\overline{(\kappa_{i}+\mu)\rho}}{2(k+1)} c_{D}^{i}
\end{equation}
where $\overline{(\kappa_{i}+\mu)\rho} = \left[(\kappa_{i}+\mu_{1})\rho_{1i}N_{1i} + (\kappa_{i}+\mu_{2})\rho_{2i}N_{2i} +...\right]/N_{i}$. From (\ref{average_price}), similar argument can be applied: higher $\kappa$ can lower the cutoff thus lowering the price index. However, this raises prices, since the consumers could have a higher willingness to pay for high quality. 

Finally, the welfare in country $i$ is as follows:

\begin{eqnarray}\label{utility}
U_{i} &=& 1 + \sum_{s} \left[\frac{1}{2} \frac{N_{i}(\alpha_{s}-\bar{p}_{i}+\beta\bar{z}_{i})^{2}}{\gamma_{s} + \eta_{s} N}+\frac{N_{i}}{\gamma_{s}} \left(\frac{1}{2} \sigma_{p}^{2} +\frac{1}{2}\beta^{2}\sigma_{z}^{2} -\beta Cov(p,z)\right)\right] \nonumber\\&=& 1 + \sum_{s}\frac{1}{2\eta_{s}}(\alpha_{s}-c_{D}^{i})\left[\alpha_{s}-\frac{k+1}{k+2}c_{D}^{i}+\frac{\overline{(\kappa_{is}-\mu_{s})\rho_{jis}}}{k+2}c_{D}^{i}\right]
\end{eqnarray}
where the quality component $\frac{\overline{(\kappa_{is}-\mu)\rho_{jis}}}{k+2}$ is defined in the similar manner as previous. Other things equal, higher valuation on quality always raises welfare: it directly raises welfare  by raising the utility obtained from consuming higher quality goods; it indirectly increases welfare by reducing the cost cutoff so that more productive firms enters and provide higher quality goods.\footnote{The third term in the square bracket has the opposite effect compared with the first two terms. I attempted to compute the last term (multiplied by $\frac{1}{2\eta} (\alpha-c_{D}^{i})$) and found that its proportion to the total welfare is less than 5\%.}

The above derivations indicate that the existence of quality can be double-sides. The subsequent proposition argues that if the quality scope is high, the quality preference amplifies the effect of the trade cost change.
\begin{prop}
The larger the scope for quality differentiation $1+(\kappa_{j}-\mu_{i})\rho_{ij} $ $(i, j \in \lbrace 1,...,J\rbrace)$, the more likely the selection effect of quality preference $\kappa_{j}$, i.e. $\frac{\partial^{2} (c_{D}^{h})^{k+2}}{\partial \kappa_{j}\partial \tau_{ij}^{-k}}$ smaller. 
\end{prop}
\begin{proof}
See Appendix \ref{apptwocountry}.
\end{proof}

\subsection{Discussion}

The theory sketched above is based on linear demand, and it addresses the point of pricing-to-market. Specifically, it relates market toughness to the behavior of firms, and consequently, to the economic aggregates such as price index and welfare. In such a setting, market toughness operates through two channels: an increase in competition and an increase in the scope for quality differentiation. The theory identifies the second channel through which market toughness affects firms' outcomes. An increase in market toughness (e.g., an increase in market size or a decrease in trade costs) raises the scope of quality differentiation because it makes it easier for firms to recover the fixed cost of innovation. Under such circumstances, each firm responds by raising quality, mark-ups, and prices. The (endogenous) relation between the scope of quality differentiation and market toughness is a key element of the model and constitutes an important deviation from past work.\footnote{These include works using CES framework, for example \textcite{Gervais}, \textcite{Mandel2010HeterogeneousFA} and \textcite{FLY2015}, etc.} For the most productive firms, quality, prices and profit rise as the innovation effect dominates the competition. These firms can earn positive profit from exporting to more competitive market. This features the advantages of exporters. The theory provides clarity on the relation between prices, productivity, market shares, and quality. In heterogeneous firms' trade models, if no quality is present, these models predict a negative correlation between prices and productivity. However, if quality is present, and if higher quality indicates higher prices, then the correlation between prices and productivity, and (possibly) between prices and firm size becomes positive. Furthermore, since these models produce a quality sorting along the productivity axis, then the correlation between prices and quality is also positive.

The model is also connected to prior theoretical frameworks. For example, \textcite{FLXY} modified CES utility by allowing positive baseline utility. The derived price is then positively related to the destination market's income. Variable markup is implied by this modification. The slight difference from this model is that quality does not depend on market size. The implications of the model are also consistent with other forms of extensions of Melitz-Ottaviano under quality framework. For example, using slightly different extensions of the MO framework with quality, \textcite{bmnw} confirms the dominance of the quality-enhancing effect of competition using French firm-level data. 

Furthermore, the setting can be extended to multi-product case. For instance, \textcite{ECKEL2015216} discovers that core products have lower costs so that firms have more incentives to invest in upgrading quality in these groups of products. I follow \textcite{mmo} and extend the current framework by allowing firms to sell more than one product to a market.\footnote{Details of setup and derivations are in Appendix \ref{multiproduct}.} The implications are isomorphic to this extension. Under this setup, it is concluded that higher product customization costs lead to a more competitive market. For individual firms, rising market competitiveness exerts a larger effect on quality improvements of products closer to the core.  

\section{Quantification}\label{emprics}
This section discusses the dataset used in the study and the methodologies to compute relevant parameters. I focus on sectors that are for final consumption, since producers and consumers can have different attitudes towards quality.\footnote{Sectors for intermediate use, capital investment and final consumption are recognized by Broad Economic Classification (BEC) conversion.} As there is no sectoral interaction here, I carry out estimations sector by sector. According to the theoretical model, the parameters to be estimated are as follows: 1) preference $\kappa_{i,s}$ and cost $\mu_{i,s}$ for each country $i$ and sector $s$;  2) cutoffs (endogenous competitiveness) $c_{D}^{i,s}$ for each country/sector; and 3) cost upper-bound (exogenous competitiveness) $c_{M}^{i,s}$. Finally, sector $s$ refers to HS 2 product level. As one step in the counterfactual analysis, I perform the quantification under quality model and M.O model.

These are the steps to estimate the model:

\textbf{Step 1.} Given the bilateral trade flow in the data, the bilateral trade cost $\tau_{ij}$ can be estimated using aggregate bilateral trade flow in (\ref{tradeflow}).

\textbf{Step 2.} Obtain the residuals from the first step estimation and compute preference for quality $\kappa_{i}$ and cost of improving quality $\mu_{i}$ for each country $i$, using Generalized Method of Moments.

\textbf{Step 3.} With the above set of parameters, the cost cutoff (endogenous competitiveness, $c_{D}^{i}$) of each country is projected using bilateral price in (\ref{price}).

\textbf{Step 4.} With parameters from the above three steps,the exogenous competitiveness (multiplied by fixed cost $f_{e}$) $c_{M}^{i}f_{e}$ can be computed from free entry condition in (\ref{fe_matrix}).

\subsection{Data}
The main dataset used here is the BACI world trade database provided by CEPII. The origin of the dataset is COMTRADE of the United Nations Statistical Division. The advantage of BACI data is that it reconciles records of both exporter and importer when there are inconsistencies in the transaction records from both sides (\cite{CEPII:2010-23}). Thus, the dataset is more accurate since the reliability of the reported data of both exporter and importer are evaluated and cross-proofed. Another advantage is that the dataset is highly disaggregated: it reports the bilateral trade value and quantity at the HS6 level for more than 5,000 HS6 sectors. 


To complement the dataset, I collect bilateral remoteness data to proxy for ice-berg trade cost. The major statistics is from GeoDist data compiled by CEPII. This dataset records bilateral information on distance, and other variables used in gravity equations to identify particular links between two countries. These variables include colonial relation, common languages, the contiguity (\cite{CEPII:2011-25}). Further, data on bilateral regional trade agreement and bilateral common currency relations are collected from de Sousa. These variables are used to proxy bilateral ice-berg trade cost. Finally, I proxy market size by population and this dataset is from Penn World Table.

I obtain Pareto distribution parameter from a global firm-level database. Here I use Orbis dataset for this purpose. The ORBIS database (compiled by the Bureau van Dijk Electronic Publishing, BvD) is
a commercial dataset, which contains administrative data on 130 million firms worldwide.
ORBIS is an umbrella product that provides firm-level data covering approximately 100+ countries, both
developed and emerging, since 2008. This dataset covers both public and private firms. I access the financial module to obtain firm-level variables. Following \textcite{NBERw21558}, I use total asset, employment and material cost (either recorded in the original dataset or imputed by subtracting the total cost of employees from the cost of goods sold) to estimate total factor productivity (TFP) at the firm level. To get $k$, I regress firm ranking (in TFP) on (computed) firm productivity.\footnote{I use \textcite{LP2003} to compute TFP. In the sample of firms, the majority consists of manufacturing.} This gives $k=3.38$. 

Prior to the empirical studies, some summary statistics for the feature of the data are displayed first. Table \ref{summary1} reports the number of observations for each HS 2 sector. A small number of observations in one sector indicates that zero trade occur very frequently. It is implied that sectors are heterogeneous in international trade transactions: HS 61 and HS 62 (clothing) has the most observations, followed by HS 85 (machinery/equipment). Table \ref{summary2} and Table \ref{summary3} report the number of HS 6 products a country imports (Table \ref{summary2}) or exports (Table \ref{summary3}). From Table \ref{summary2}, the USA imports the highest number of products, followed by Italy, Netherlands and United Kingdom. From Table \ref{summary3}, China and the USA export the highest number of products, followed by Netherlands and Italy and United Kingdom. To some extent, a large exporter can also be a large importer. Finally, it is evident in Table \ref{summary2} and Table \ref{summary3} that the distribution of the number of HS 6 products exported/imported is uneven across countries. Some countries import/export fewer than 100 products. 

Another stylized fact lies in the \textit{f.o.b} price. Figure (\ref{fig:20422}) displays the relation between (log of average) \textit{f.o.b} price and (log of) GDP per capita in the destination country for the sector of 020422 (meat sheep or goats; fresh, chilled or frozen). The price to a destination market is the average \textit{f.o.b} price across all source countries exporting to that destination. Figure (\ref{fig:940169}) presents for the sector of 940169 (Seats). These two sectors are chosen since they have many importing countries. From the graph, it is implied that sectors can be heterogeneous in relation between price and destination market income. Therefore, the 'quality sorting' channel can exist only in a subset of sectors. The finding is to some extent consistent with \textcite{mz2012}, which reveals that firms charge higher prices in richer destinations within a firm-product category, using Chinese Customs Trade Statistics. It is also consistent with with \textcite{KY2016}, which uses the same dataset as \textcite{mz2012} and finds that quality sorting (competition raise quality and price) exists in a subset of HS 2 industries while some other industries have efficiency sorting (competition lowers price). 

\subsection{Quality Preference and Cost Parameters}
I recover quality preference and cost parameters based on the estimation of the gravity equation. According to (\ref{tradeflow}), the bilateral trade flow can be decomposed into origin fixed effects, destination fixed effects, bilateral trade costs and a non-linear function of $\kappa_{j}$ and $\mu_{i}$. Specifically, I estimate the following gravity equation for each HS2 sector $s$ (I suppress sector notation here):
\begin{equation}\label{reg}
\ln r_{ij} = \delta_{i} + \delta_{j} -(k+1)\ln\tau_{ij}+ \ln \left[1+(\kappa_{j}-\mu_{i})\rho_{ij}\right] + \ln \left(\frac{1}{k(k+2)}+\frac{(\kappa_{j}+\mu_{i})\rho_{ij}}{k(k+1)(k+2)}\right)
\end{equation}
where origin fixed effects $\delta_{i} = \ln(N_{i}^{E}\left(c_{M}^{i}\right)^{-k})$ and destination fixed effects $\delta_{j} = \ln\left(L^{j}\left(c_{D}^{j}\right)^{k+2}\right)$. This form is the original \textcite{MO2008} gravity equation with additional quality terms. Equation (\ref{reg}) has the advantage of explaining residuals in the original MO model. The ice-berg trade cost is proxied by several bilateral variables, so that the specification of \ref{reg} can be re-written as
\begin{center}
$\ln r_{ij} = \delta_{i} +\delta_{j} +\beta_{1} Contig +\beta_{2} Comlang +\beta_{3} Colony +\beta_{4} Comcol +\beta_{5} Curcol +\beta_{6} Smctry + \beta_{7}LnDist + \beta_{8} RTA +\beta_{9}Comcur + \epsilon_{ij}$
\end{center}
In the above specification, $Contig =1$ if the exporter ($i$) and the importer ($j$) are contiguous, and $Contig =0$ otherwise. $Comlang=1$ if $i$ and $j$ share common official language, $Colony=1$ if $i$ and $j$ have ever in colonial relationship, $Comcol=1$ if $i$ and $j$ have a common colonizer after 1945, $Curcol=1$ if $i$ and $j$ are currently in colonial relation, $Smctry=1$ if $i$ and $j$ were the same country. $LnDist$ is the (log) distance between $i$ and $j$. \footnote{This is the simple distance, which is the distance between the most populated cities of the two countries.} The data of those variables are from CEPII GeoDist. Additionally, $RTA =1$ if $i$ and $j$ have reached any trade agreements and $Comcur =1$ if $i$ and $j$ are in the same currency union. This is from \textcite{DESOUSA2012917}.

An issue with the ordinary least squares (OLS) is the presence of zero trade as is discussed before. The above gravity equation can be subject to bias due to the existence of zero bilateral trade, which to some extent suggested by Table \ref{summary1}. As is argued by \textcite{helpman2008estimating}, disregarding country pairs that do not trade with each other can cause biased estimates on the data. Additionally, it is documented in their empirical study that half of countries do not trade with each other. To address such issue, I use Poisson Pseudo Maximum Likelihood (PPML) proposed by \textcite{silva2006log} by including zero trade flows. The same estimation strategy is also used in \textcite{CGMO2011}.

Table  \ref{trade_costs} reports results from PPML estimates for selected sectors. At the first glance, the coefficients are within expectations. Specifically, distance has negative effects on trade flows. Regional trade agreements, common language and colonial relations can have positive effects on bilateral trade. Common currency can have both positive and negative effects. It is noticeable that most other sectors not reported here have similar patterns in terms of trade cost proxies.

Given the estimated equation, I am able to back parameters of quality preferences and costs from the residuals. These parameters are recovered using residuals in specification (\ref{reg}). Although BACI covers more than 200 countries, only 168 of them have data on population in the Penn World Table. Thus, for each HS2 sector, I have at most 336 parameters to recover. There are much more residuals than parameters. Therefore, the system is over-identified and hence I employed a least square procedure to hunt for the optimal solution. \footnote{Specifically, the optimal $\kappa$ and $\mu$ satisfy: $\left[\kappa,\mu\right]=\frac{1}{N} argmin \sum_{ij} \left[\ln \left[1+(\kappa_{j}-\mu_{i})\rho_{ij}\right] + \ln \left(\frac{1}{k(k+2)}+\frac{(\kappa_{j}+\mu_{i})\rho_{ij}}{k(k+1)(k+2)}\right) - \epsilon_{ij} \right]^{2}$, where $N$ is the number of observations of trade flows in each sector.}

I estimate those parameters for each sector and discovered stylized patterns from them. Table \ref{kappa} reports the summary statistics of estimates of $\kappa$'s for each sector; Table \ref{mu} reports the summary statistics of $\mu$'s. First, a comparison of Table \ref{kappa}
and Table \ref{mu} indicates that $\kappa$ is on average larger than $\mu$ and tends to be more dispersed than $\mu$. This to some extent implies that (dis) tastes for quality are more heterogeneous among countries than the technology of quality-improving. Second, a closer examination of Table \ref{kappa} indicates that sectors differ in the dispersion of $\kappa$'s across countries. For instance, HS 50 (textiles) has relatively small dispersion in $\kappa$, while HS 85 (electric motors and generators) has a very high dispersion. Thus, in some sectors countries tend to be homogeneous in preferences, while in other sectors countries can be divergent. 

Table \ref{kappa_gdp} attempts to explore the relation between the parameters and GDP per capita. Illustrations of the positive relation for 40 sectors are in the two panels of Figure \ref{fig:gdp_kappa}. \footnote{Consistent with Table \ref{kappa_gdp}, I illustrate those sectors where GDP per capita has positive and significant effect on $\kappa$}. Table \ref{kappa_gdp} presents coefficients in regressing estimated $\kappa$ on GDP per capita for 63 HS 2 sectors. It is evident that the positive correlation between $\kappa$ and GDP per capita exists in a subset of industries, while in other sectors insignificant effects exist. Comparing with Table \ref{kappa}, it is implied that sectors with medium level of dispersion tend to have $\kappa$ correlated with GDP per capita. One caveat is that the formation of quality preferences $\kappa$ is exogenous to the model. This differentiate with \textcite{dfr2014} where quality preferences are modeled as endogenous to income. According non-homothetic preferences setup in \textcite{fk2016}, quality preference of a country can also depend on income inequality inside the country. In sum, consumer preferences result from multiple socio-economic conditions and it is out of the scope of the this paper. 

In addition, I also examine the correlation between cost of quality and country income. Table \ref{mu_gdp} reports coefficients on GDP per capita in regressing $\mu$ on GDP per capita. Compared with Table \ref{kappa_gdp}, significant effects appear in fewer sectors. In addition, 3 of these sectors exhibit negative and significant results. The correlation between the marginal cost and income is less explicit compared with preferences. The explanation for this has two sides: the positive association can be accommodated by the fact that richer countries have higher labor cost while the negative association can be explained by that richer countries have a comparative advantage in producing higher quality goods since the demand for higher-quality is larger.

\subsection{Endogenous and Exogenous Competitiveness}

After parameters of preferences and cost of quality are estimated in the last subsection, the cost cutoff for each country is computed using the parameters estimated in the last step. Specifically, I use average bilateral \textit{f.o.b} price in (\ref{price}), the calculated trade cost and $\kappa_{j}$ (and $\mu_{i}$) to back out the cost threshold for each country. However, it is worth noting that price calculated using  value divided by quantity is noisy even though units are converted to tons. This potentially results in multiple cost cutoffs for one country in one sector. Thus, I modify (\ref{price}) by allowing noisy terms. To account for this, I regress price on trade costs, nonlinear terms of $\kappa$ /$\mu$ and destination country fixed effects, as in the following specification:
\begin{equation}\label{cutspecification}
\ln p_{ij} = cons + \ln \left[ \frac{2k+1}{k+1} + \frac{1}{k+1}\left(\kappa_{j}+\mu_{i}\right)\rho_{ij}\right] - \beta_{1}\ln \tau_{ij} + \psi_{j} +\epsilon_{ij}
\end{equation}
In (\ref{cutspecification}), the exponential of the coefficients on each of the destination fixed effects are the cost cutoffs. \footnote{This follows from \textcite{AllenAtkin}, where preference parameters in that paper is recovered from good-level fixed effects.} This corresponds to the endogenous competitiveness in \textcite{CGMO2011}.\footnote{In \textcite{CGMO2011}, the cost cutoffs are computed using price index of countries in each sector. The practice is infeasible here due to data coverage issues.}

With the cutoffs obtained from the empirical implementations, I back out cost upper bounds (exogenous competitiveness) $c_{M}^{i}$ (multiplied by $f_{e}$) for each country $i$ in each sector $s$. For this set of parameters, I use Equation (\ref{cutoff}) to back out exogenous competitiveness for each country/sector ($(c_{M}^{i})^{k}f_{e}$) by the following relation: 

\begin{center}
$f_{e}*c_{M}^{k} =\frac{B*L*c_{D}^{k+2}}{2\gamma(k+1)(k+2)}$
\end{center}
where $c_{M}^{k}$ is the vector of $(c_{M}^{i})^{k}$ and $c_{D}^{k+2}$ is the vector of $(c_{D}^{i})^{k+2}$, which is estimated as in Equation (\ref{cutspecification}). Matrix $B$ is defined in the prior subsection. I also implement above two steps under M.O by eliminating all quality components in the relevant equations. Table \ref{cst1} and Table \ref{cst2} report results as summaries in estimations of exogenous competitiveness across sectors.

Exogenous competitiveness displays several characteristics. First, the levels of these statistics vary to a large extent across sectors. For example, for HS 50, the figures are in $10^{12}$, while in other sectors the mean is below one. The cause in the large variation across sectors can be attributed to the different degrees of variations in the multiplication $\left[1+(\kappa_{j}-\mu_{i})\rho_{ij}\right]$ for different sectors.\footnote{I perform the similar estimation of $c_{M}^{k}$ following multi-country version of \textcite{MO2008}, typically, the means and standard deviations are smaller for almost all sectors.} The different degrees of variation results from the dispersion differences in $\kappa$ and $\mu$ across sectors, which are analyzed in the prior subsection. Comparing with preference for quality, sectors having a large variation in $\kappa$ can also have large variations in $c_{M}^{k}$. Another explanation could be that the sunk cost $f_{e}$ vary across industries.

The similar estimation strategy is also implemented without quality considerations. This follows from multi-country \textcite{MO2008} and \textcite{CGMO2011}. Subsequently, I compute $c_{D}^{k+2}$ implied by the two models. The purpose of this practice is to compare the predicted endogenous competitiveness of the two models. Table \ref{cutoff_qua} and Table \ref{cutoff_noqua} summarize cutoffs computed from the two different models across sectors. Illustrations of the distribution of cutoffs implied by both models for all sectors are displayed in the two panels of Figure \ref{cutoff_dist251}. In comparison with the two models, several implications arise. Firstly, for most sectors, the levels of cutoffs are on the same scale: the means in cutoffs do not significant differ from each other. Second, for most sectors, the distribution of cutoffs are more dispersed under quality model than under the MO model: the variances in Table \ref{cutoff_qua} is larger and for most sectors, cutoffs under quality sorting have higher maximum values and smaller minimum values. Those comparisons suggest that adding quality sorting enlarges the inequality among countries in competitiveness. 

To further compare the two models, I summarize the differences in the predictions in cutoffs across countries for each sector. Figure \ref{fig:diff_dist251}  displays the kernel density of prediction differences  across countries of various sectors. For most industries, the majority of differences lie around zero. The distribution in differences vary across sectors, which is observed from the skewness. For sectors such as HS 2 and HS 3, the difference distributions are nearly normal so the probability of overestimation and  underestimation are nearly equal. Right skewness occurs in sectors such as HS 74 and HS 82, indicating that quality model over-predicts more than it under-predicts. Left skewness happens in sectors such as HS 4 and HS 10 so that more under-estimates of quality model exist in these sectors. In general, within most sectors, predictive differences between the two models do not deviate to a large extent from zero.

\section{Counterfactual Scenario}\label{counterfactual}
Having estimated the model, one can use it to simulate the effects of trade frictions/liberalizations. This is achieved by recomputing for each sector the (quality-adjusted) trade freeness matrix $B$ while keeping the exogenous competitiveness, shape parameters and preference/costs for quality parameters constant. The resulting matrix is then used to compute new endogenous competitiveness. Specifically, the following statistic is to be computed:

\begin{center}
$\hat{c}_{D}^{j} = \frac{c_{D}^{'j}}{c_{D}^{j}}-1$
\end{center}

where $c_{D}^{'j}$ is the cutoff after trade liberalization. I estimated the above statistics under both this model and the M.O model and compare the differences between the two models to discover the underlying regularities leading to the differences. 

Furthermore, I also experiment with preference and technological shocks. The purpose of the counterfactual is to explore the effect of preferences or technological progress (in improving quality) on endogenous competitiveness across countries and sectors. This practice is infeasible under the framework of \textcite{MO2008} or other efficiency-based settings to the best of my knowledge. Thus, it is a significant contribution of this model.

\subsection{International Trade Costs}
The first exercise is to examine the effect of bilateral trade cost changes. I would like to examine the change in cutoffs across countries and sectors. I simulate a 5\% increase in international trade cost while keeping intra-national trade costs constant. This counterfactual analysis is to examine the effect of protectionism on the endogenous competitiveness. This counterfactual analysis is relevant since the temptation of protectionism in the past decade has been large and increasing for a large number of policy makers. Thus there is possibility of introducing new external tariffs or other types of trade barriers. 

As is noted above, I focus on two aspects: the changes in cutoffs and the differences in predictions implied by the two models. Table \ref{counter-factual1} summarizes the mean, standard deviation, minimum and maximum in the changes of cutoffs. Table \ref{counter-factual1_noqua} reports those statistics implied by \textcite{MO2008}. It is worth noting that in almost all sectors here, an increase in international trade cost (while keeping intra-national trade cost constant) can lead to a decrease in cost cutoff for some countries, regardless of the models. This decline in cutoff (or rise in the endogenous competitiveness) is justified by the rising entry of local producers since the relative cost of intra-national trade is cheaper than international costs. In the two-country case, consider that in MO with two countries, the cutoff is computed as 

\begin{center}
$c_{D}^{i,k+2} = \frac{\gamma\phi}{L^{i}}\frac{1-\omega^{j}}{1-\omega^{i}\omega^{j}}$
\end{center}
where $\omega_{i}=\tau_{ji}^{-k}$ and indicates the freeness of trade. If the both $\omega^{j}$ and $\omega^{i}$ increase, the change in cutoff is then 
\begin{equation}\label{DeltaCutoff}
\Delta c_{D}^{i,k+2} = \frac{\gamma\phi}{L^{i}} \left(1-\omega^{i}\omega^{j}\right)^{-2}\left[(1-\omega^{j})\omega^{j}\Delta \omega^{i}-(1-\omega^{i})\Delta \omega^{j}\right]]
\end{equation}
so that the direction of change in cutoff depends on the magnitudes of $\omega^{i}$, $\omega^{j}$ and the degrees in the changes of both. Under the current model, $\omega^{i}$ and $\omega^{j}$ can be higher than 1, since in this model it is replaced with $\left[1+\left(\kappa_{j}-\mu_{i}\right)\rho_{ij}\right]\tau_{ij}^{-k}$ and thus the sign of $\Delta c_{D}^{k+2}$ can be negative when $\Delta \omega^{j}$ and $\Delta \omega^{i}$ are negative.

Nonetheless, the aggregate cutoffs increase, for both models. A closer examination of the two sets of results implies the following regularities. First, the mean changes are all positive for both models, with the changes implied by the quality models higher than the MO model for most sectors.  This difference can be partially driven by the higher maximum value in Table \ref{counter-factual1} (for most sectors).  Second, and related to the prior section, the variance of changes is larger under this model. This is connected with the prior section in that variance under this model is higher than under \textcite{MO2008} with additional quality components. Finally, under both models, the maximum of the change is larger than 1 for all sectors. Thus, the magnitude of effect on endogenous competitiveness is larger than the degree of trade cost rise. 

Finally, it is worth exploring the underlying forces driving these changes. The focal points lie in whether such difference is systematic across countries in some (or all) sectors. If such regularities exist, then one can argue that these changes are not randomly driven. For all sectors, I attempt to correlate with quality preference of each country.\footnote{I also attempt to link them with cost of quality of each country. However, no systematic regularities are found.} In the two panels of Figure \ref{fig:diff_beta}, I tried to plot the (log) quality preferences with the prediction differences (changes predicted by this model minus changes predicted by MO).\footnote{I find similar patterns in the following sectors: HS2, HS4, HS7, HS8, HS9, HS11, HS15, HS16, HS17, HS18, HS19, HS35, HS36, HS38, HS39, HS52, HS58, HS63, HS64, HS65, HS69, HS73, HS76, HS81, HS86, HS95 and HS96.} All else being equal, the higher the quality preference of the country, the larger the predicted difference between the current model and MO model. This result indicates that in models with quality preference differentials across countries, the disparity in terms of loss from higher trade barriers is larger than in models without quality considerations. This implication is to some extent consistent with simulations on the market size by \textcite{ANTONIADES2015263}, in which it is shown that larger market size results in larger cutoff decrease in countries with higher quality preferences.

\subsection{Preference/Cost Changes}

The second counterfactual exercise is to examine the effect of a universal preference or technology shocks on endogenous competitiveness. This analysis is infeasible under models of MO or other models based on the efficiency sorting of heterogeneous firms in international trade. It is perceived to be relevant since consumers' preferences are changing over time. Although quality preferences are exogenous to this model, they can be altered in several ways. An increase in quality preference can result from higher expectations of consumers for the products they purchase. For instance, in food industry, food safety issues motivate more consumers to pursue organic food, which is perceived to be high-quality, thus raising the willingness-to-pay (\cite{foodquality}). The pervasiveness of advertisements and other multi-media can also shape consumers' preference for high quality products. Apart from that technology progress drives down the cost of quality updates. 

Essentially, the two changes are consistent. As is shown in the model, higher preference and lower marginal cost in quality raise the scope for quality differentiation. Therefore, I combine the two exercises by adding/ subtracting 0.1 units in either $\kappa_{j}$/ $\mu_{i}$ in computing $\rho_{ij}$. All else being equal, individual producers entering each market raise product quality and charge higher prices. If other conditions remain unchanged, resulting cost cutoffs are lowered. In the long run, a trade diversion effect can occur in that a subset of countries can experience a rise in their cost cutoffs. This arises since the cutoff of each country depends on both its own $\kappa$ (and $\mu$) and others' $\kappa$ ($\mu$). This can be explained in detail in Equation (\ref{DeltaCutoff}) if one replaces $\omega_{i}$ with $\left[1+\left(\kappa_{i}-\mu_{j}\right)\rho_{ji}\right]\tau_{ji}^{-k}$. Therefore, the direction of change in cutoffs depends on the magnitudes of $\kappa$ and $\mu$ across countries. 

As is expected, trade diversion occurs in a few countries. Table \ref{cf21} and \ref{cf22} summarize the percentage changes in cost cutoffs across countries. Figure \ref{fig:cutoff_change_distribution} visualizes the changes using the world map. I compute the weighted cutoff changes for each country. The weight is calculated as an industry's share of total value of imports of a country. Although the majority experience rising market competitiveness revealed in declining cost cutoffs, 25 of them have the opposite change. Overall, the magnitude of increase in cost cutoffs is lower than decrease: the decrease in cost cutoffs ranges from around 500\% to around 0.6\%. Thus positive preference or technology shocks brings more positive effects on productivity improving globally. A further plot in Figure \ref{fig:cutoff_change_distribution} reveals that the more than half of the countries have their declines in cutoffs falling less than 141.2\%. The change in endogenous competitiveness falling in the range of $\left[-141.2,0.58\right]$ occupies the most area.

Large countries generally gain more than small economies. This can be supported from the optimal quality choice in (\ref{quality}). Loss in utility which is revealed from rise in cost cutoff occur mostly in small economies. The justification of this observation can be that small economies have smaller scopes for quality differentiation, thus firms have less incentive to sell high-quality goods. Faced with the same degree of preference rise, firms are diverted to sell higher-quality goods to larger markets. I prove that by regressing gains from trade on their (log) population and other country level controls in Table (\ref{size-gain}). They reveal negative relationship between the change in cost cutoffs and population size. This exercise implies that market size is negative associated with changes in cost cutoff, i.e., positively related to changes in productivity and welfare gains. On average, 1\% rise in population size leads to 7.876\%  more gains in productivity. However, one caveat is that the aggregate gain in productivity is affected by other factors such as compositions of imports of a country. 

\section{Conclusion}\label{conclude}

This paper extends a theory on quality with endogenous markups. Theoretical framework is of multi-country type, which is a generalization of two-country model commonly used in Melitz and Ottaviano framework. Different from competition on cost, the theory identifies that in some sectors and countries, firms can also compete on quality. Tough competition featured by larger market, lower trade cost and higher preference for quality are more likely to induce firms to improve quality. The selection effect is larger when quality differentiation scope is wide.

Empirical study is undertaken to compare this theory with efficiency framework. Structural estimation is used to identify relevant parameters. These parameters are later used to compute cutoffs and average prices under quality competition. The same steps are used to compute counterparts under efficiency competition. The structural estimation implies that considering quality differentials among countries enlarges heterogeneities in competitiveness. The counterfactual study points out that in most sectors, the higher the quality preference of a country, the larger the loss from rising trade barrier, compared with Melitz and Ottaviano. Finally, positive universal preference shocks generally bring more gains to larger countries.

Though the paper addresses the importance of considering quality preference differentials across countries, it still has insufficiencies in investigating this issue. This study can be extended to examine the spatial distribution in quality preferences and technologies, using more disaggregated data, such as China Customs Trade Statistics and China Inter-Provincial Input-Output Statistics. Furthermore, one can compare the impacts of international and intra-national trade costs on competitiveness of different regions in China, under both the current model and Melitz-Ottaviano model. 
\newpage

\printbibliography
\newpage

\section*{Appendix}
\appendix
\numberwithin{equation}{section}
\makeatletter 
\newcommand{\section@cntformat}{Appendix \thesection:\ }
\makeatother
\section{Multi-Product Extension}\label{multiproduct}
\subsection{Setup}
The theoretical model is based on single product. The setup in the model assumes that each firm produces one product only. However, it can be generalized to assume that each firm produces multiple products and sells to multiple countries. The utility function in (\ref{consumer_utility}) and the production cost in (\ref{totalcost}) remain unchanged. I follow \textcite{mmo} by assuming that there exist "product ladder" with increasing customization cost for products further from the core product. Specifically, I assume that each product $m$ produced by a firm with core marginal cost $c$ incurs the marginal cost of 

\begin{equation}\label{productcost}
v(m,c) = \lambda^{-m}c
\end{equation}
with $\lambda \in \left(0, 1\right)$.

In the above setup, more peripheral products require higher customization costs. In (\ref{productcost}), $m$ is positive and represents the distance from the core product of firm $c$: $m =0$ indicates that product $m$ is the core product and higher $m$ implies that product $m$ requires higher adjustment cost. Higher $m$ reflects decreasing product appeal. Thus, if a firm in country $i$ decides to sell product $m$ to country $j$, its total cost becomes
\begin{equation}
TC_{ij} (c, m) = q(c, m)(c + \mu_{j}z(c, m)) + \delta z(c, m)^{2}  
\end{equation}

and its problem becomes:
\begin{equation}
\pi_{ij}(c, m) = \max_{q_{j}(c, m), z_{j}(c, m)} p_{ij}(c, m) q_{ij}(c, m) - TC_{ij}(c, m)
\end{equation}
Similar to the case of single-product model, one can derive the optimal quality of product $m$ as
\begin{equation}\label{productquality}
z_{ij}(c, m) = \rho_{ij}  (c_{D}^{j} - \tau_{ij}\lambda^{-m} c)
\end{equation}
where $\rho_{ij} = L^{j}(\kappa_{j} - \mu_{i})/ \left(4\delta\gamma - L^{j}(\kappa_{j}-\mu_{i})^{2}\right)$. Therefore, within a firm, more peripheral products are sold with lower quality. More productive firms provide higher-quality core products. Additionally, products are affected differently within a firm. A simple algebra in  (\ref{productquality}) implies that competitiveness of market $j$ (revealed in trade costs and market size) imposes heterogeneous effects on products of different hierarchies. To see this, the first-order derivative of $z_{ij} (c, m)$ with respect to $\tau_{ij}$ is

\begin{center}
$\frac{\partial z_{ij} (c, m)}{\partial \tau_{ij}} = \rho_{ij} \left(\frac{\partial c_{D}^{j}}{\partial \tau_{ij}}-\lambda^{-m}c\right)$
\end{center}

That derivative with respect to $m$ is then
\begin{equation}\label{productsod}
\frac{\partial^{2} z_{ij} (c, m)}{\partial \tau_{ij} \partial m} = \rho_{ij}c\lambda^{-m}\ln \lambda \leq 0
\end{equation}
since $0<\lambda <1$. From (\ref{productsod}), second order derivative implies that products closer to the core (lower $m$) are affected by market competition more significantly. 

The number of products a firm can offer to the destination market depends on the firm-level marginal cost $c$. Firms with lower marginal cost $c$ generally produce more products to a destination market. If the cost cutoff to export from $i$ to $j$ is $c_{x}^{ij}$, the number of products the firm $c \leq c_{x}^{ij}$ offers is
\begin{center}
$M_{ij} (c) = \max \left\lbrace m \mid c \leq c_{x}^{ij}\lambda^{m} \right\rbrace + 1$
\end{center}

the total profit of firm $c$ in exporting to country $j$ is the sum of profits from all products:
\begin{equation}
\Pi_{ij} (c) = \sum_{m=0}^{M_{ij}(c)} \pi_{ij} (c, m)
\end{equation}
\subsection{Equilibrium}
Similar to single product setting, a firm draws its marginal cost $c$ prior to entry to the market with sunk cost $f_{e}$. The expected profit from selling to all markets equal to the sunk cost in equilibrium. The total expected profit can then be decomposed into profits from each product sold to each destination. This can then be expressed as
\begin{equation}
\sum_{j\in J} \int_{0}^{c_{x}^{ij}} \Pi_{ij}(c) dG(c) = \sum_{j\in J}\sum_{m=0}^{\infty}\left[ \int_{0}^{\lambda^{m}c_{x}^{ij}} \pi_{ij} (c, m) d G(c)\right] = f_{e}
\end{equation}
Again, I assume Pareto distribution of $c$. Thus the expected profit of a firm in country $i$ can be expressed as 
\begin{equation}\label{productfe}
\frac{2c_{M}^{-k}}{(k+1)(k+2)}\left(1-\lambda^{k}\right)^{-1} \sum_{j\in J} \frac{L_{j}}{4\gamma}\left[1+\left(\kappa_{j}-\mu_{i}\right)\rho_{ij}\right]\tau_{ij}^{-k}\left(c_{D}^{j}\right)^{k+2} = f_{e}
\end{equation}
The above form can be re-written for $J$ countries. Thus, one can write the condition (\ref{productfe}) in the matrix form in a similar manner as in single product setting. The cutoff $c_{D}^{j}$ satisfies the following:
\begin{equation}\label{productcutoff}
\left(c_{D}^{j}\right)^{k+2} = \frac{2\gamma(k+1)(k+2)(1-\lambda^{k})f_{e}}{\mid B\mid} \frac{\sum_{i}C_{ij}\left(c_{M}^{i}\right)^{k}}{L_{j}}
\end{equation}
where $\mid B \mid$ is the determinant of matrix $B$ and $C_{ij}$ is the cofactor of $B_{ij}$. Other parameters in (\ref{productcutoff}) are defined in the similar manner as in single product setting. It is implied from the above condition that endogenous competitiveness can also depend on product flexibility $\lambda$: if $(1-\lambda^{k})^{-1}$ is large, the cutoff is small and the market is more competitive. 

Other aggregate variables are derived in the similar approach as in single product setup. It is important to notice that the product flexibility can vary across countries and sectors. The aggregate bilateral trade value from $i$ to $j$ also depends on the product flexibility and larger flexibility implies higher bilateral trade value. The expected number of entrants are also determined by the product flexibility. 

\begin{equation}\label{producttradeflow}
r_{ij} = \frac{kN_{i}^{E}(c_{m}^{i})^{-k}}{2(1-\lambda^{k})\gamma} L^{j}(\tau_{ij})^{-(k+1)} \left(c_{D}^{j}\right)^{k+2}\left[1+(\kappa_{j}-\mu_{i})\rho_{ij}\right] \left(\frac{1}{k(k+2)}+\frac{(\kappa_{j}+\mu_{i})\rho_{ij}}{k(k+1)(k+2)}\right)
\end{equation}

For computing the number of entrants, the matrix $F$ becomes
\begin{center}
$\frac{2\gamma(k+1)(1-\lambda^{k})(\alpha-c_{D}^{i})}{\eta c_{D}^{i,k+1}}$
\end{center}

\section{Proofs of Propositions}\label{apptwocountry}
\subsection{Proof of Proposition 1}
By the free entry condition in \ref{fe_matrix}, we can obtain $J$ conditions in the following form:
\begin{equation}\label{freeentry}
\left[1+(\kappa_{1}-\mu_{i})\rho_{i1}\right]L_{1}\tau_{i1}^{-k} (c_{D}^{1})^{k+2} + ... + \left[1+(\kappa_{J}-\mu_{i})\rho_{iJ}\right]L_{J}\tau_{iJ}^{-k} (c_{D}^{J})^{k+2} = 2\gamma(k+1)(k+2)f_{e}c_{M}^{k}
\end{equation}
When $\tau_{ij} = \tau_{ji}$ and when there is a bilateral trade liberalization, the first-order condition with respect to $\tau_{ij} (\tau_{ji})$ implies the following:
\begin{equation}\label{prop11}
\sum_{d\in J} B_{dd'}L_{d'}\frac{\partial (c_{D}^{d'})^{k+2}}{\partial \tau_{ij}^{-k}} =0
\end{equation}
if $d \neq i, j$

and 

\begin{subequations}\label{prop12}
\begin{align}
\sum_{d'\in J} B_{id'} L_{d'}\frac{\partial (c_{D}^{d'})^{k+2}}{\partial \tau_{ij}^{-k}} + \frac{B_{ij}}{\tau_{ij}^{-k}}L_{j}(c_{D}^{j})^{k+2} = 0\\
\sum_{d'\in J} B_{jd'} L_{d'}\frac{\partial (c_{D}^{d'})^{k+2}}{\partial \tau_{ji}^{-k}} + \frac{B_{ji}}{\tau_{ji}^{-k}}L_{i}(c_{D}^{i})^{k+2} = 0
\end{align}
\end{subequations}
for countries $i$ and $j$ respectively. 

Since $\frac{B_{ji}}{\tau_{ji}^{-k}}L_{i}c_{D}^{i} >0$ and $\frac{B_{ij}}{\tau_{ij}^{-k}}L_{j}c_{D}^{j} >0$, from (\ref{prop12}), at least one of  $\frac{\partial (c_{D}^{d'})^{k+2}}{\partial \tau_{ji}^{-k}}$ (or $\frac{\partial (c_{D}^{d'})^{k+2}}{\partial \tau_{ji}^{-k}}$, $d' \in \lbrace 1,..., J\rbrace$) is negative. However, if all of them are negative, the condition in (\ref{prop11}) cannot be satisfied. 
The conditions in (\ref{prop11}) and (\ref{prop12}) can be written in the matrix form:
\begin{equation}
\textbf{B}\textbf{L}\textbf{c}' = \textbf{F}
\end{equation}
where $\textbf{B}$ and $\textbf{L}$ is defined the same as in Section \ref{theory}. $\textbf{c}'$ is the vector with the $d$ th element being $\partial (c_{D}^{d})^{k+2}/\partial \tau_{ij}^{-k}$. On the right hand side, $\textbf{F}$ is the vector with $i$th element being $-\frac{B_{ij}}{\tau_{ij}^{-k}}L_{j}c_{D}^{j}$, $j$ th element being $-\frac{B_{ij}}{\tau_{ij}^{-k}}L_{j}c_{D}^{j}$ and other elements being 0. Thus, the sign of $\partial (c_{D}^{d})^{k+2}/\partial \tau_{ij}^{-k}$  depend on those two non-zero elements, determinant and the cofactors of matrix $\textbf{B}$:
\begin{equation}
\partial (c_{D}^{d})^{k+2}/\partial \tau_{ij}^{-k} = -\frac{1}{\mid B \mid} \left(C_{di}\frac{B_{ij}}{\tau_{ij}^{-k}}L_{j}(c_{D}^{j})^{k+2} + C_{dj}\frac{B_{ij}}{\tau_{ij}^{-k}}L_{j}(c_{D}^{j})^{k+2}\right)
\end{equation}
where $C_{di}$ is the $di$ th element in the cofactor matrix of $\textbf{B}$.
\subsection{Proof of Proposition 2}
The proof of Proposition 2 is similar to Proposition 1. Taking the derivative of Equation (\ref{freeentry}) with respect to $\kappa_{j}$, one can obtain J equations in the following form:
\begin{equation}\label{fockappa}
\sum_{h\in J} B_{ih}L_{h}\frac{\partial (c_{D}^{h})^{k+2}}{\partial \kappa_{j}} + \frac{\partial B_{ij}}{\partial \kappa_{j}}L_{j}(c_{D}^{j})^{k+2} =0
\end{equation}
In the matrix form, (\ref{fockappa}) can be written as 
\begin{equation}
\textbf{B}\textbf{L}\frac{\partial \textbf{c}_{D}^{k+2}}{\partial \kappa_{j}} = \textbf{G}
\end{equation}
where $\textbf{G}$ is the vector with with $i$th element being $\frac{\partial B_{ij}}{\partial \kappa_{j}}L_{j}(c_{D}^{j})^{k+2} >0$. 

Thus, the vector of partial derivatives can be computed as:
\begin{equation}
\frac{\partial({c}_{D}^{i})^{k+2}}{\partial \kappa_{j}} = -\frac{L_{j}(c_{D}^{j})^{k+2}}{L_{i}\mid B \mid}\sum_{h\in J} \mid C_{ih} \mid \frac{\partial B_{hj}}{\partial \kappa_{j}}
\end{equation}
The sign of $\frac{\partial({c}_{D}^{i})^{k+2}}{\partial \kappa_{j}}$ depends on the determinant of $\textbf{B}$ as well as its cofactors.
By the definition of $\textbf{B}$, $\frac{\partial B_{hj}}{\partial \kappa_{j}} > 0$, $\forall$ $h$. Therefore, $\exists$ $\textbf{B}$ such that $\frac{\partial ({c}_{D}^{i})^{k+2}}{\partial \kappa_{j}} >0$, $\forall$ $i$, $j$. 

\subsection{Proof of Proposition 3}
Taking the derivative of (\ref{fockappa}) with respect to $\tau_{ij}$ ($\tau_{ji}$), one can obtain two sets of equations, if the origin country is $i$ or $j$, the following holds:

\begin{equation}
\sum_{h\in j} B_{i(j)h} L_{h}\frac{\partial^{2} (c_{D}^{h})^{k+2}}{\partial \kappa_{j}\partial \tau_{ij}^{-k}} + L_{j}\frac{\partial B_{i(j)j}}{\partial \tau_{ij}^{-k}}\frac{\partial (c_{D}^{j})^{k+2}}{\partial \kappa_{j}} +\frac{\partial^{2} B_{i(j)j}}{\partial \kappa_{j}\partial \tau_{ij}^{-k}}L_{j}(c_{D}^{j})^{k+2} + L_{j}\frac{\partial B_{i(j)j}}{\partial \kappa_{j}} \frac{\partial (c_{D}^{j})^{k+2}}{\partial \tau_{ij}^{-k}} =0
\end{equation}

where $\frac{\partial B_{ij}}{\partial \tau_{ij}^{-k}}>0$. If the origin country is not $i$ or $j$, the following shall hold:
\begin{equation}
\sum_{h\in J} B_{dh}L_{h}\frac{\partial^{2} (c_{D}^{h})^{k+2}}{\partial \kappa_{j}\partial \tau_{ij}^{-k}} + L_{j}\frac{\partial B_{dj}}{\partial \kappa_{j}}\frac{\partial (c_{D}^{j})^{k+2}}{\partial \tau_{ij}^{-k}} =0
\end{equation}
If $\frac{\partial (c_{D}^{j})^{k+2}}{\partial \tau^{-k}_{ij}}$, $\frac{\partial (c_{D}^{j})^{k+2}}{\partial \kappa_{j}}$ $<0$, $\exists$ $\textbf{B}$ such that $\frac{\partial^{2} (c_{D}^{h})^{k+2}}{\partial \kappa_{j}\partial \tau_{ij}^{-k}} <0$. 

From the above polynomial equations, the higher $B_{dh}$, the lower $\frac{\partial^{2} (c_{D}^{h})^{k+2}}{\partial \kappa_{j}\partial \tau_{ij}^{-k}}$, i.e. more like that $\frac{\partial^{2} (c_{D}^{h})^{k+2}}{\partial \kappa_{j}\partial \tau_{ij}^{-k}} <0$.

\section{Productivity Estimation}
I construct the firm-level measures based on the ORBIS enterprise database of Bureau van Dijk (BvD). This dataset provides comprehensive information on listed and de-listed private companies around the world. I use the financial module of the database. It provides firm-level financial report items including total revenues, employment, total assets, and research and development (R\&D) expenses.%
\footnote{Data are downloaded in US dollars.}

I follow long-established methods of estimating firm productivity as a residual of Cobb-Douglas production function. In this regard, both methodologies proposed by \textcite{op1996} (OP) and \textcite{LP2003} (LP) are possible candidates given the current dataset. I choose to estimate firm productivity based on LP, because the LP approach relies on intermediate inputs as a proxy rather than on investment, whose level may be non-positive and depends on the assumption of the depreciation rate. On the other hand, it is common that firms record positive use of materials/energy so that I preserve as many observations as possible.\footnote{I also use OP as a robustness check and estimated productivities are similar.} 

I choose to download recent 10 years of financial data for each firm. The missing values exist. Thus, this is an unbalanced panel. The procedure to estimate productivity is as follow. First, gross output, capital and total inputs are proxied by total revenues, total assets, and Costs of Goods Sold (COGS), respectively. The cost of material/energy (in short, material, henceforth) is calculated by COGS minus total wage payable, by the accounting definition of COGS, if material cost is unavailable. Second, these values are deflated to obtain the quantity counterpart.\footnote{The total sales revenues are deflated by Consumer Price Index (CPI), the total assets deflated by the index of fixed asset investment deflator, and the material normalized by Producer Price Index (PPI). Currently, I use the US CPI (Total All Items) and PPI (for All Commodities), and construct the index of fixed asset investment deflator from gross fixed investment flows. They are retrieved from the US Federal Reserve Bank of St.\ Louis website, https://fred.stlouisfed.org/.}
The number of employees are directly observable from the data. Given the observations on gross output, labor, material, and capital, I estimate the production function based on the Stata program $levpet$ using as instruments current capital, lagged material, lagged labor, lagged two year material and lagged capital. Because industries can vary in their production technologies, the estimation is done separately for each 3-digit NAICS sector.

\newpage

\begin{table}
\caption{Price and Destination Country Population}
\input{stylized1.tex}

\label{style1}
\end{table}

\begin{table}
\caption{Price and Revenue}
\input{stylized2.tex}

\label{style2}
\end{table}

\begin{table}
\caption{Price and Entry}
\input{stylized3.tex}

\label{style3}
\end{table}

\begin{figure}[ht]
\centering
\begin{subfigure}[b]{0.45\textwidth}
	\centering
	\includegraphics[width=\linewidth]{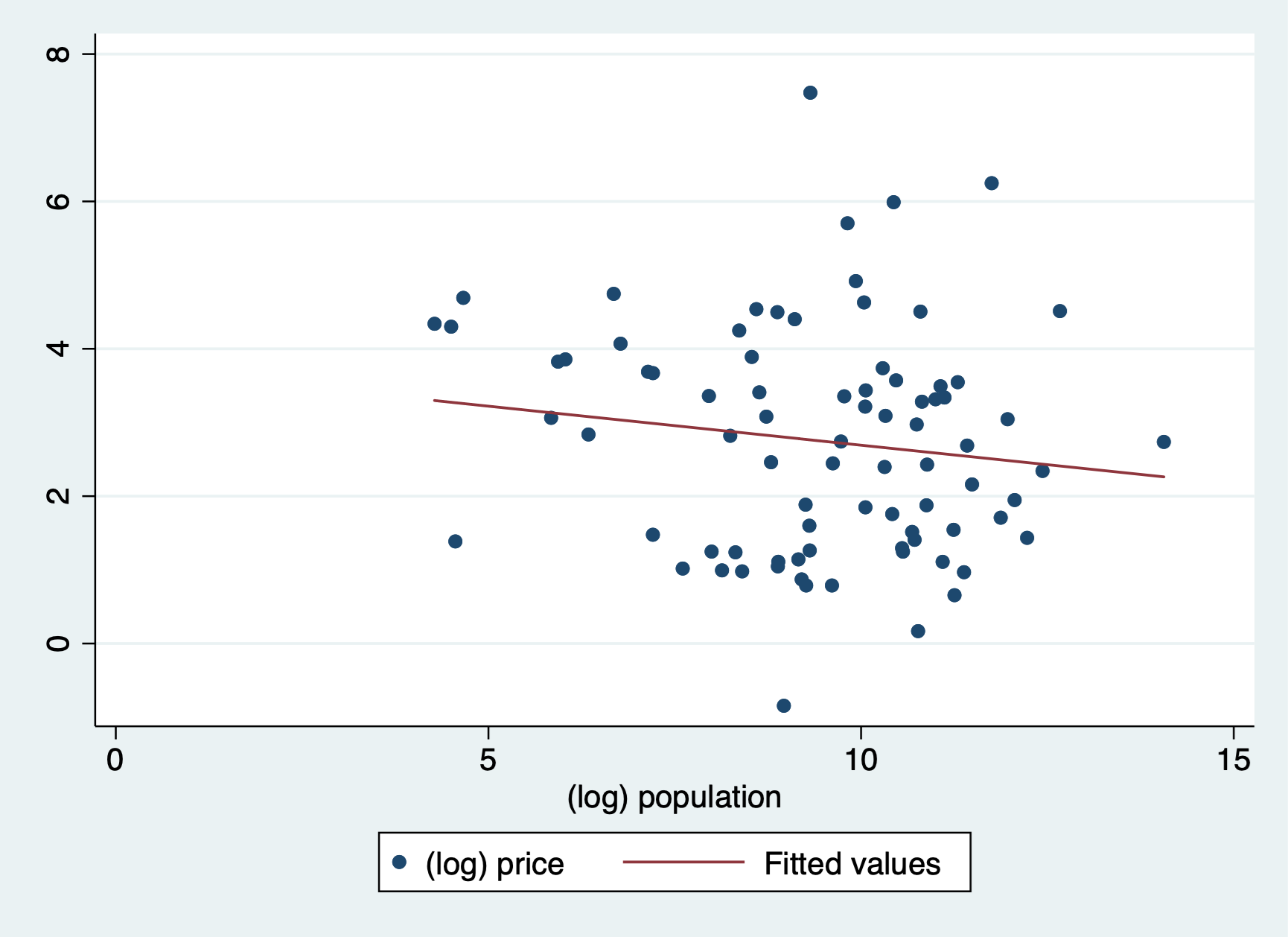}
	\caption{Price and Market Size: HS 24}
	\label{fig:stylized1_hs24}
\end{subfigure}
\hfill
\begin{subfigure}[b]{0.45\textwidth}
	\centering
	\includegraphics[width=\linewidth]{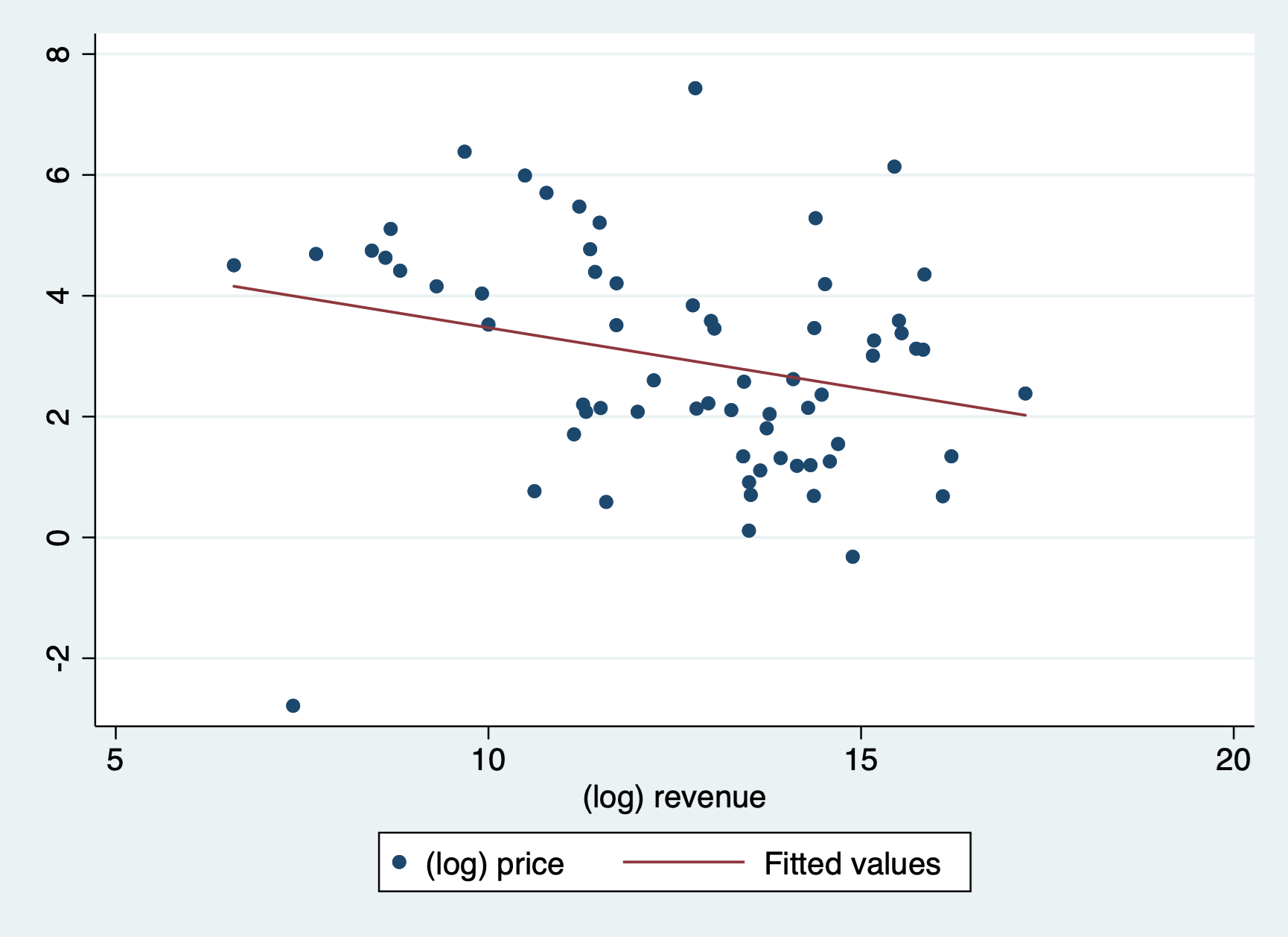}
	\caption{Price and Revenue: HS 24}
	\label{fig:stylized2_hs24}
\end{subfigure}
\hfill
\begin{subfigure}[b]{0.45\textwidth}
	\centering
	\includegraphics[width=\linewidth]{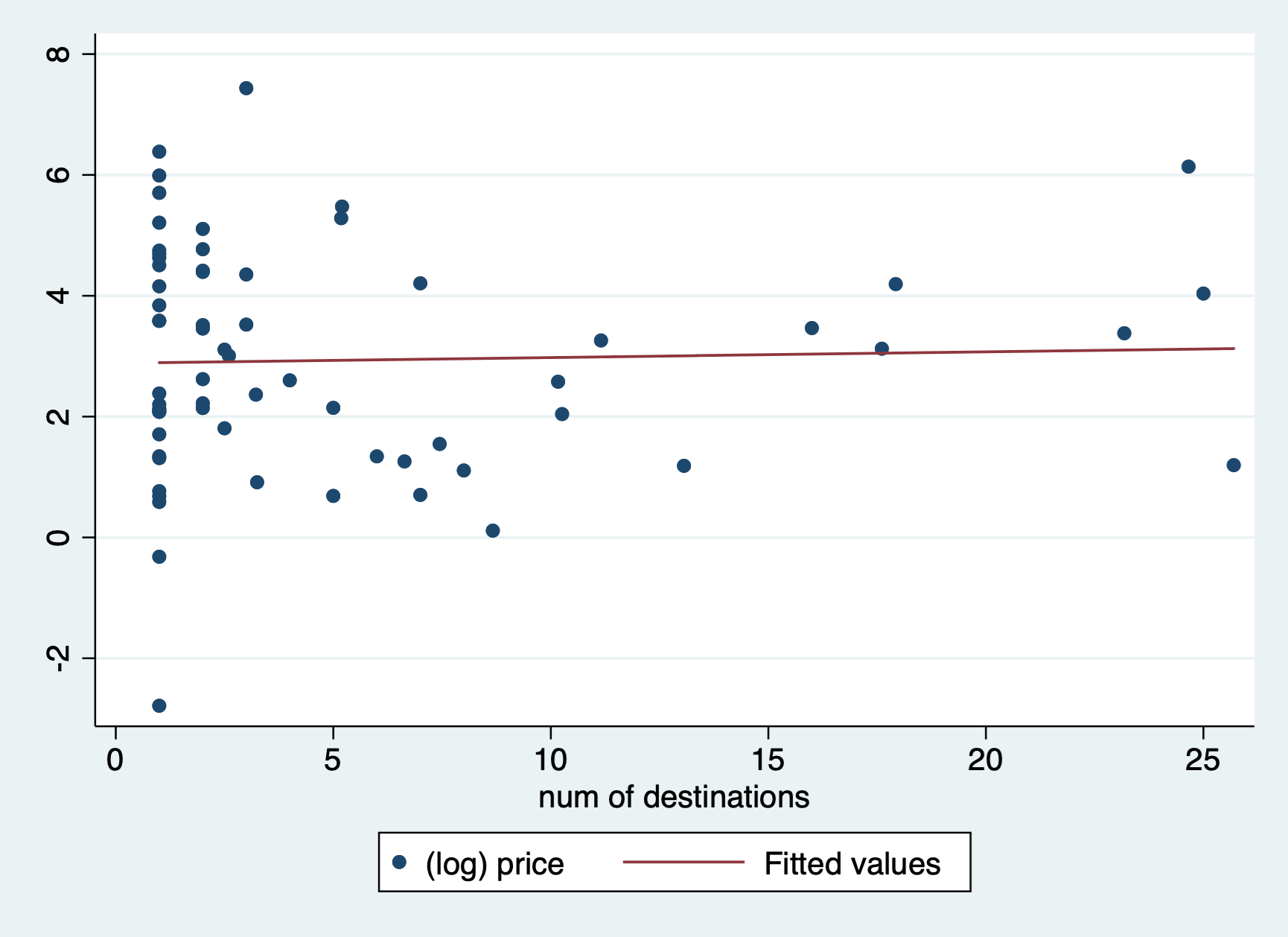}
	\caption{Price and Market Entry: HS 24}
	\label{fig:stylized3_hs24}
\end{subfigure}

\caption{Price, Revenue and Destination Market Characteristics HS 24}
\label{fig:s1}
\floatfoot{Note: The three figures show the correlation between average (log) price across firms exporting to a destination and destination market size (Panel a), average (log) price and average (log) revenue across destinations of each firm (Panel b), and average (log) price of a firm across destinations and average number of markets a firm enters (Panel c), of firms in sector HS 24 (Tobacco). Market size is proxied by population size and the data is from Penn World Table.}
\end{figure}

\begin{figure}[ht]
\centering
\begin{subfigure}[b]{0.45\textwidth}
	\centering
	\includegraphics[width=\linewidth]{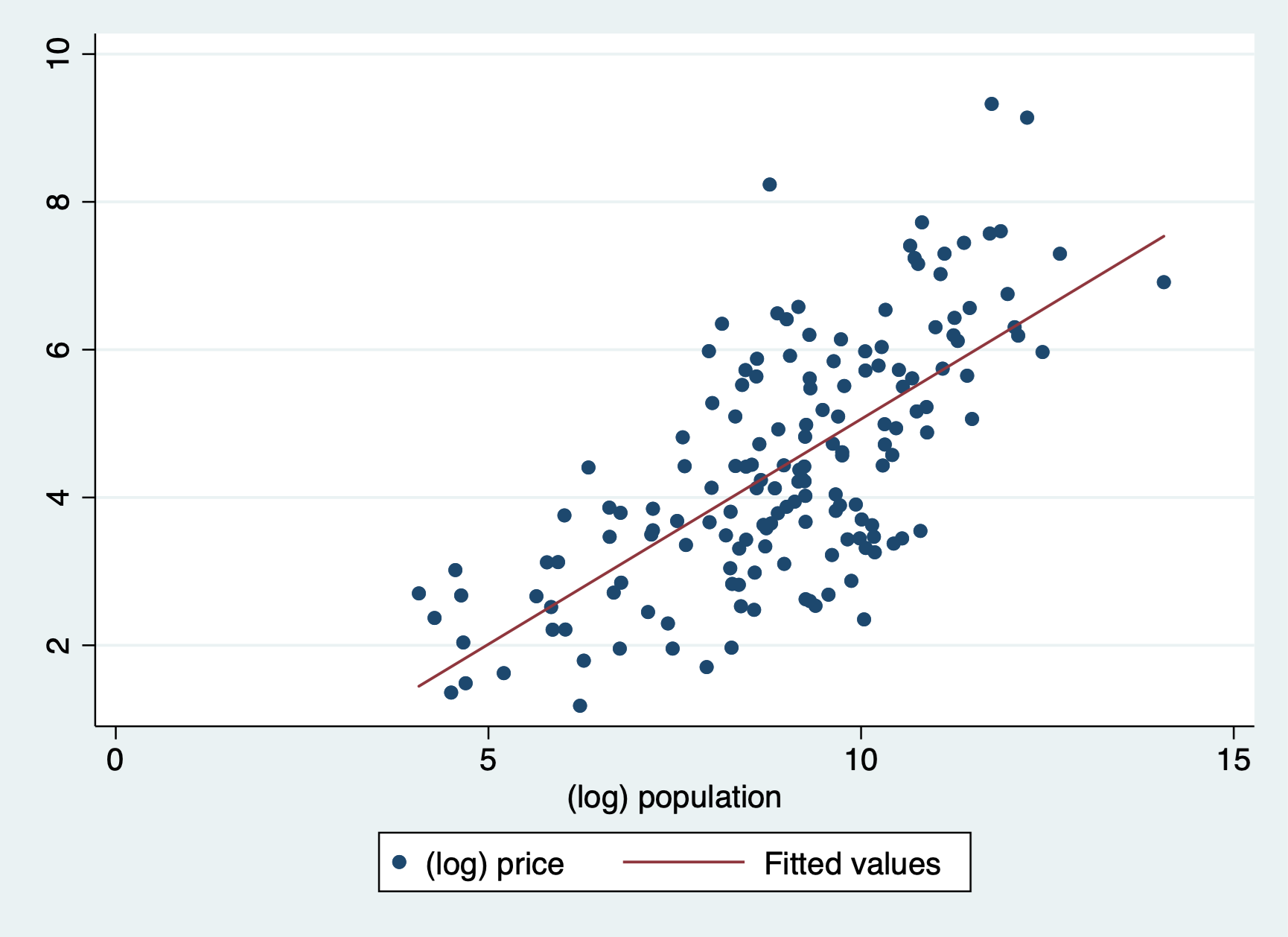}
	\caption{Price and Market Size: HS 30}
	\label{fig:stylized1_hs30}
\end{subfigure}
\hfill
\begin{subfigure}[b]{0.45\textwidth}
	\centering
	\includegraphics[width=\linewidth]{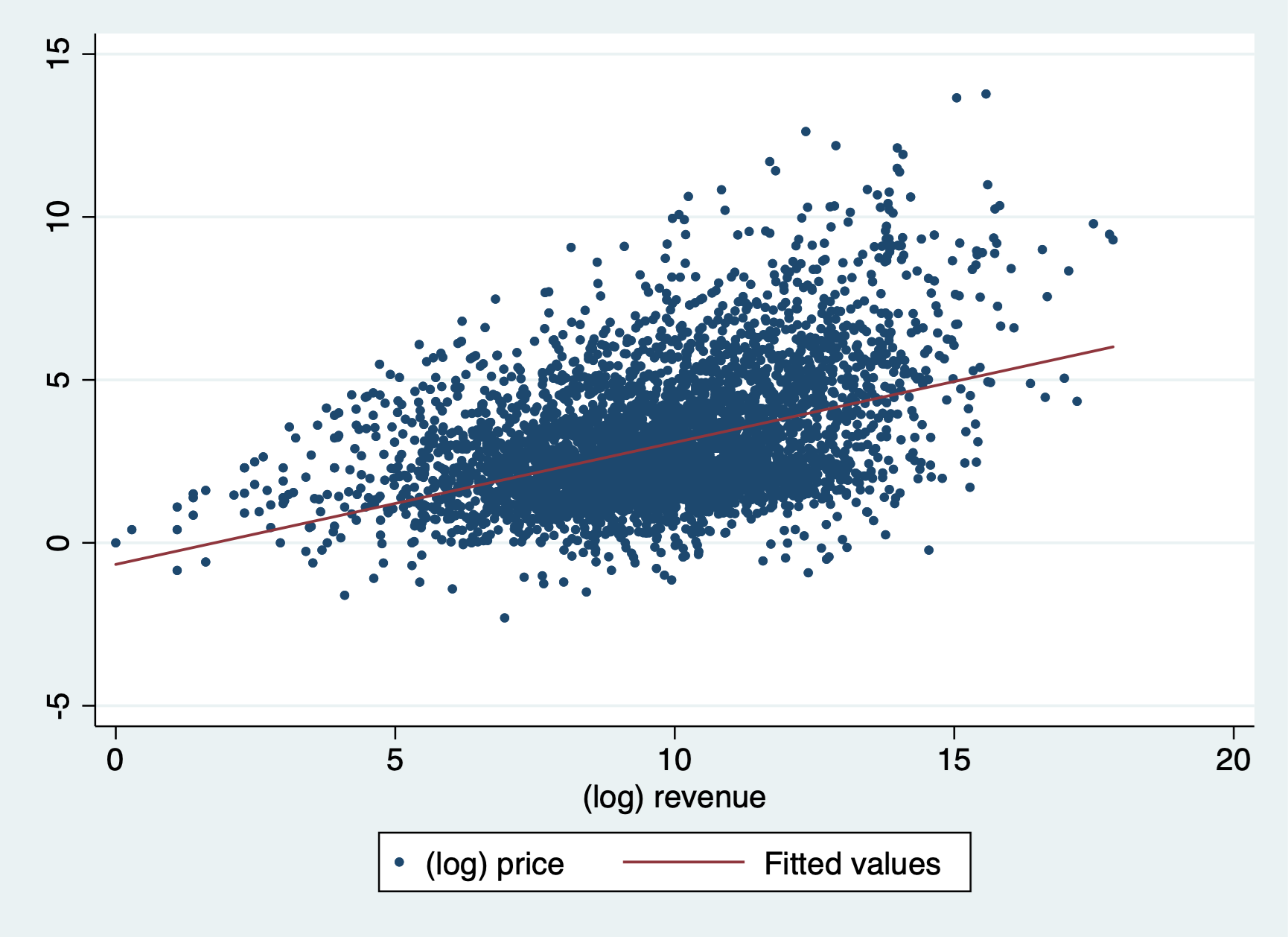}
	\caption{Price and Revenue: HS 30}
	\label{fig:stylized2_hs30}
\end{subfigure}
\hfill
\begin{subfigure}[b]{0.45\textwidth}
	\centering
	\includegraphics[width=\linewidth]{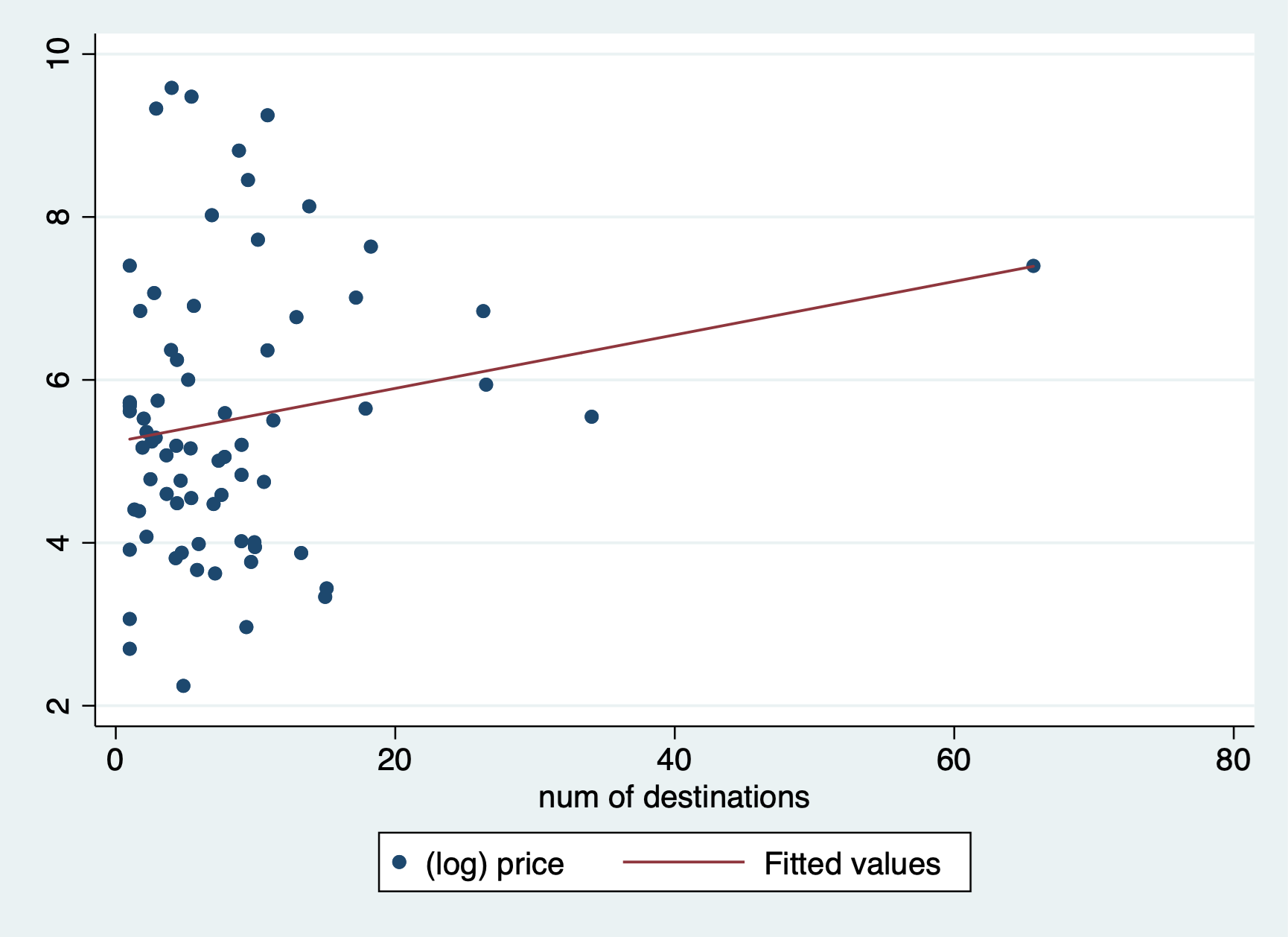}
	\caption{Price and Market Entry: HS 30}
	\label{fig:stylized3_hs30}
\end{subfigure}

\caption{Price and Destination Market Size}
\label{fig:s2}
\floatfoot{Note: The three figures show the correlation between average (log) price across firms exporting to a destination and destination market size (Panel a), average (log) price and average (log) revenue across destinations of each firm (Panel b), and average (log) price of a firm across destinations and average number of markets a firm enters (Panel c), of firms in sector HS 30 (Pharmaceutical). Market size is proxied by population size and the data is from Penn World Table.}
\end{figure}

\begin{table}
\centering
\small\addtolength{\tabcolsep}{-1pt}
\caption{Description of HS 2-digit Industries and Number of Observations}
\centering
\begin{adjustbox}{width=0.65\textwidth,center}
\input{summary1.tex}
\end{adjustbox}
\label{summary1}
\end{table}

\begin{table}[ht]
\centering
\begin{adjustbox}{width=0.8\textwidth}
\small
\caption{Summary of Importing Countries' Products}
\label{summary2}
\input{summary2.tex}

\end{adjustbox}
\end{table}

\begin{table}
\centering
\begin{adjustbox}{width=0.8\textwidth}
\small
\caption{Summary of Exporting Countries' Products}
\label{summary3}
\input{summary3.tex}

\end{adjustbox}
\end{table}

\begin{figure}[ht]
\centering
\begin{subfigure}[b]{0.65\textwidth}
	\centering
	\includegraphics[width=\linewidth]{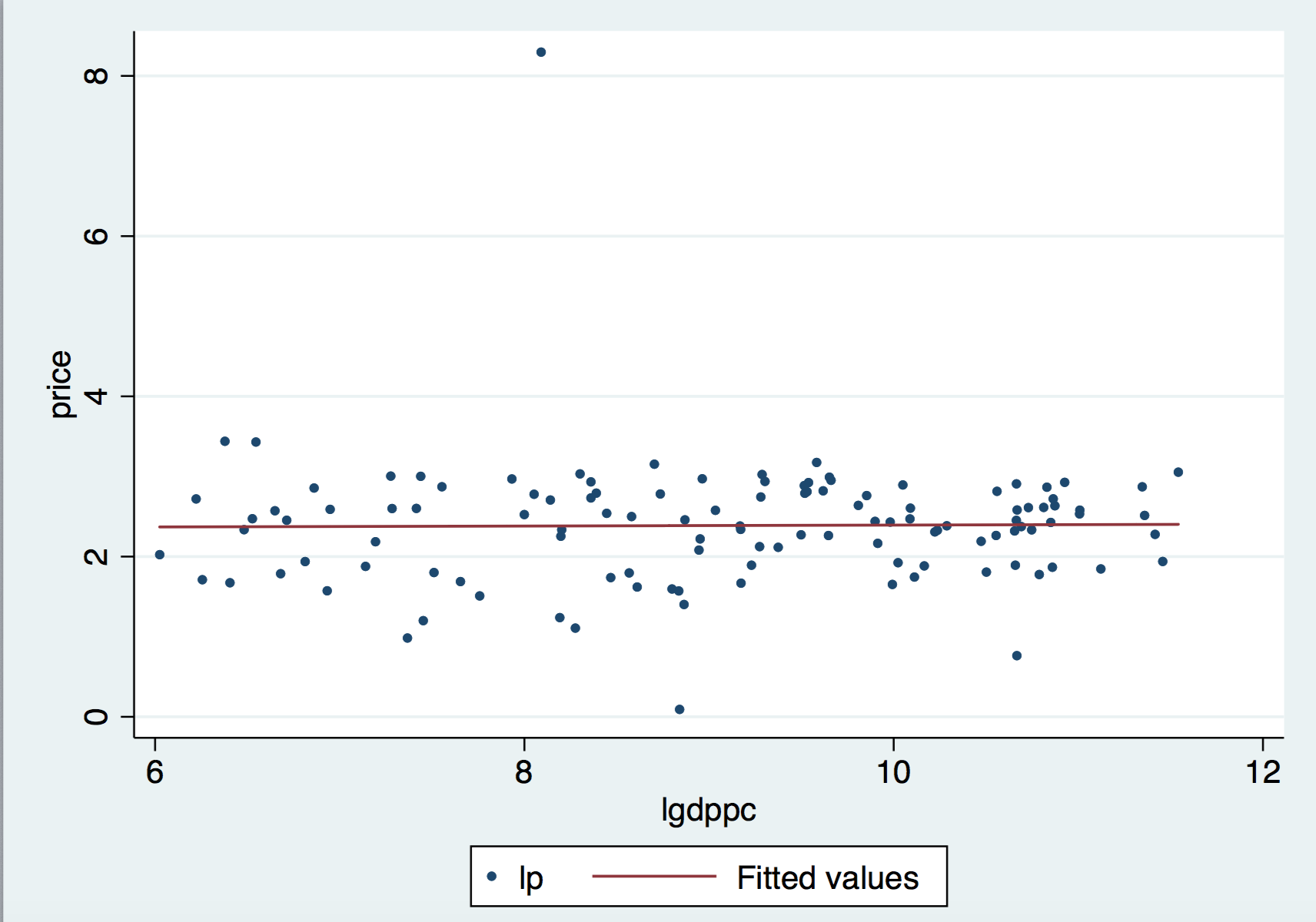}
	\caption{Unit Value and GDP Per Capita: HS 20422}
	\label{fig:20422}
\end{subfigure}
\vfill
\begin{subfigure}[b]{0.65\textwidth}
	\centering
	\includegraphics[width=\linewidth]{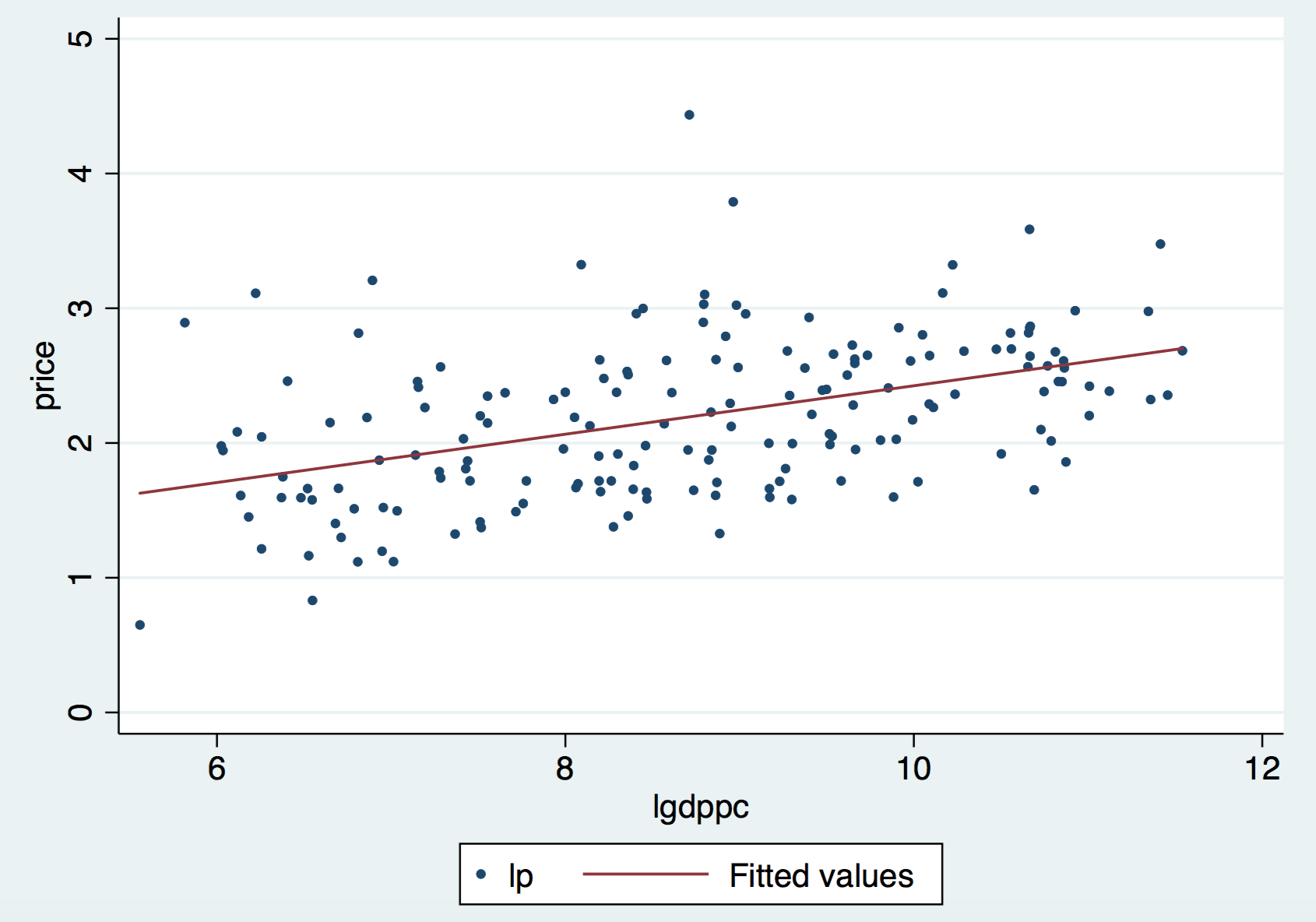}
	\caption{Unit Value and GDP Per Capita: HS 940169}
	\label{fig:940169}
\end{subfigure}

\caption{The Correlation between Price and Destination Income}
\label{fig:price_gdp}
\floatfoot{Note: The two figures show the correlation between price and destination GDP per capita for two sectors: (a) HS 20422 (Meat sheep or goats; (b) fresh, chilled or frozen) and HS 940169 (Seats)}
\end{figure}

\begin{table}[htbp]
\caption{Summary of Coefficients on Trade Cost Variables} \label{trade_costs}

\begin{subtable}{1\textwidth}
\caption{HS 2 - HS 49}
\label{trade_costs1}
\centering
\begin{adjustbox}{width=0.75\textwidth,center}
\small
\input{trade_costs1.tex}

\end{adjustbox}
\end{subtable}
\end{table}

\pagebreak

\begin{table}[htbp]
\ContinuedFloat

\begin{subtable}{1\textwidth}
\caption{HS 50 - HS 97}
\label{trade_costs2}
\centering
\begin{adjustbox}{width=0.75\textwidth,center}
\small
\input{trade_costs2.tex}

\end{adjustbox}
\end{subtable}

\end{table}

\begin{table}
\centering
\small
\caption{Estimated Preference for Quality}
\centering
\begin{adjustbox}{width=0.8\textwidth,center}
\input{parameters.tex}

\end{adjustbox}
\label{kappa}
\end{table}

\begin{table}
\centering
\small
\caption{Estimated Marginal Cost of Quality}
\centering
\begin{adjustbox}{width=0.8\textwidth,center}
\input{parameters1.tex}

\end{adjustbox}
\label{mu}
\end{table}

\begin{table}
\centering
\small
\caption{Preference for Quality and GDP Per Capita}
\centering
\begin{adjustbox}{width=0.9\textwidth,center}
\input{kappa_gdp.tex}

\end{adjustbox}
\label{kappa_gdp}
\end{table}

\begin{figure}
\centering
\caption{Preference for Quality and Income, HS 2 - HS 51}
\label{fig:gdp_kappa}
\begin{subfigure}[b]{1.4\textwidth}
	\includegraphics[scale=0.62]{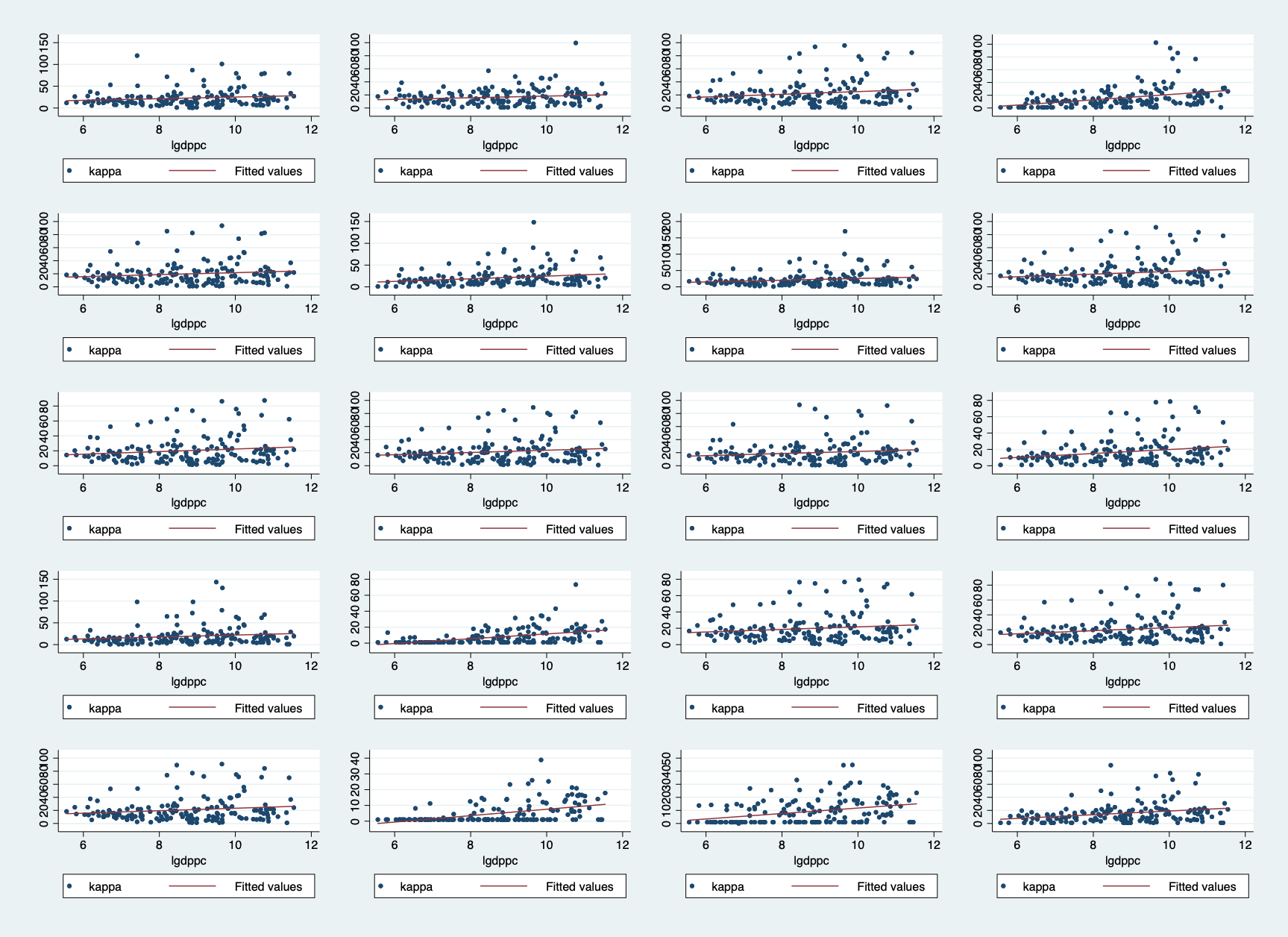}
    \floatfoot{Note: The figures show positive correlations between income and quality preferences for sectors: HS 2, HS 3, HS 4, HS 6,\\ HS 8, HS 11, HS 15, HS 16, HS 19, HS 21, HS 24, HS 32, HS 36, HS 37, HS 40, HS 42, HS 49, HS 50, HS 51}
\end{subfigure}
\end{figure}

\pagebreak

\begin{figure}\ContinuedFloat
\caption{Preference for Quality and Income, HS 52 - HS 97}
\label{fig:gdp_kappa5496}
\begin{subfigure}[b]{1.4\textwidth}
	\includegraphics[scale=0.62]{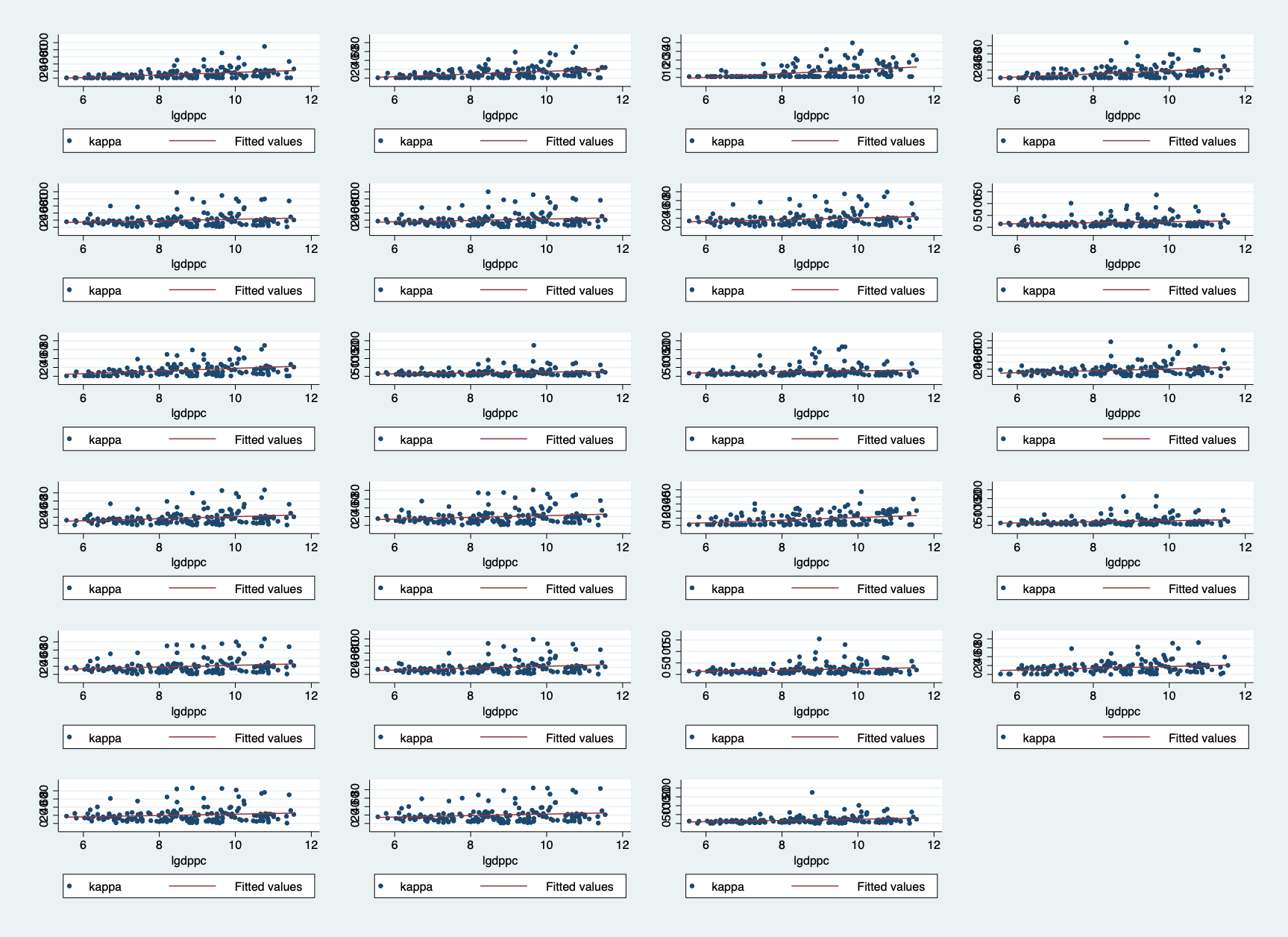}
	
    \floatfoot{Note: The figures show positive correlations between income and quality preferences for sectors: HS 54, HS 55, HS 58, \\HS 59, HS 61, HS 62, HS 65, HS 66, HS 68, HS 69, HS 70, HS 71, HS 74, HS 76, HS 87, HS 88, HS 89, HS 90, HS 91, \\HS 92, HS 93, HS 94, HS 96}
\end{subfigure}

\end{figure}

\begin{table}
\centering
\small
\caption{Cost of Quality and GDP Per Capita}
\centering
\begin{adjustbox}{width=0.9\textwidth,center}
\input{mu_gdp.tex}

\end{adjustbox}
\label{mu_gdp}
\end{table}

\begin{table}[htbp]
\caption{Summary of Exogenous Competitiveness} \label{2-51}

\begin{subtable}{1\textwidth}
\caption{HS 2 - HS 51}
\label{cst1}
\centering
\begin{adjustbox}{width=0.63\textwidth,center}
\small
\input{summary_cm1.tex}

\end{adjustbox}
\end{subtable}
\end{table}

\pagebreak

\begin{table}[htbp]
\ContinuedFloat

\begin{subtable}{1\textwidth}
\caption{HS 52 - HS 96}
\label{cst2}
\centering
\small
\input{summary_cm2.tex}

\end{subtable}

\end{table}

\begin{table}
\centering
\caption{Summary of Cutoffs --- Quality Model}
\centering
\begin{adjustbox}{width=0.85\textwidth,center}
\small
\input{cutoff_quality.tex}

\end{adjustbox}
\label{cutoff_qua}

\end{table}

\begin{table}
\centering
\caption{Summary of Cutoffs ---Melitz and Ottaviano Model}
\begin{adjustbox}{width=0.85\textwidth,center}
\small
\input{cutoff_noqua.tex}

\end{adjustbox}
\label{cutoff_noqua}
\end{table}

\begin{figure}
\caption{Distribution of Cutoffs, HS 2 - HS 51}
	\label{cutoff_dist251}

\begin{subfigure}[b]{1\textwidth}
\includegraphics[width=1\textwidth,height=1.4\textheight,keepaspectratio]{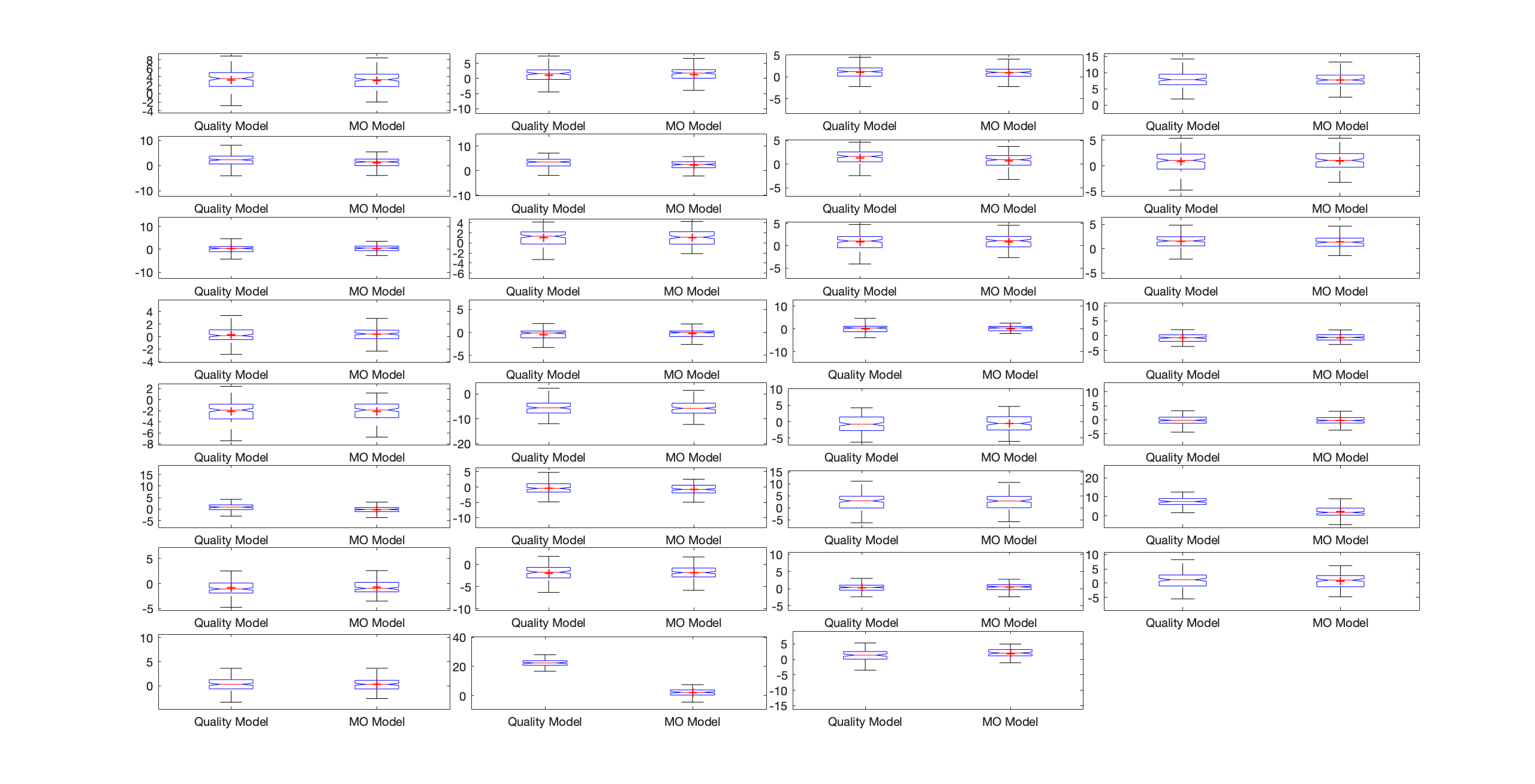}
    \footnotesize{Note: The figures show differences in distribution of cutoffs between the model in this paper (right) and Melitz and Ottaviano (left): HS 2, HS 3, HS 4, HS 6, HS 7, HS 8, HS 9, HS 10, HS 11, HS 15, HS 16, HS 17, HS 18, HS 19, HS 20, HS 21, HS 23, HS 24, HS 30, HS 32, HS 34, HS 35, HS 36, HS 37, HS 38, HS 39, HS 40, HS 42, HS 49, HS 50, HS 51}
\end{subfigure} 
\end{figure}

\pagebreak

\begin{figure}\ContinuedFloat
\caption{Distribution of Cutoffs, HS 52 - HS 97}
	\label{cutoff_dist5296}
    \begin{subfigure}[b]{1\textwidth}
\includegraphics[width=1\textwidth,height=1.4\textheight,keepaspectratio]{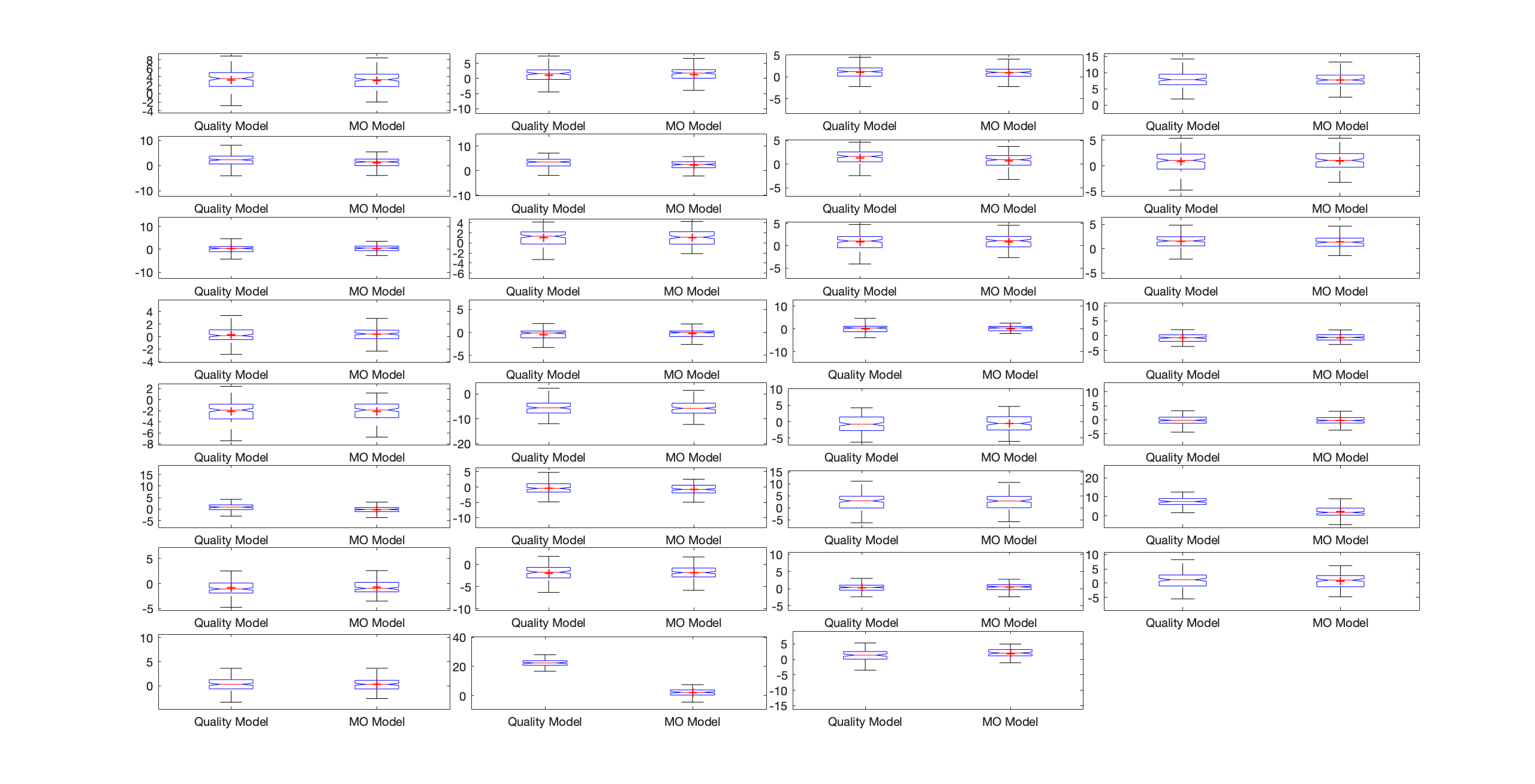}
    \footnotesize{Note: The figures show differences in distribution of cutoffs between the model in this paper (right) and Melitz and Ottaviano (left): HS 52, HS 54, HS 55, HS 58, HS 59, HS 61, HS 62, HS 63, HS 64, HS 65, HS 66, HS 67, HS 68, HS 69, HS 70, HS 71, HS 73, HS 74, HS 76, HS 82, HS 87, HS 88, HS 89, HS 90, HS 91, HS 92, HS 93, HS 94, HS 95, HS 96}
\end{subfigure}
    \end{figure}

\begin{figure}
\centering
\caption{Distribution of of Differences Cutoffs, HS 2 - HS 51}
\label{fig:diff_dist251}
\begin{subfigure}[b]{1\textwidth}	\includegraphics[width=1\textwidth,height=1.4\textheight,keepaspectratio]{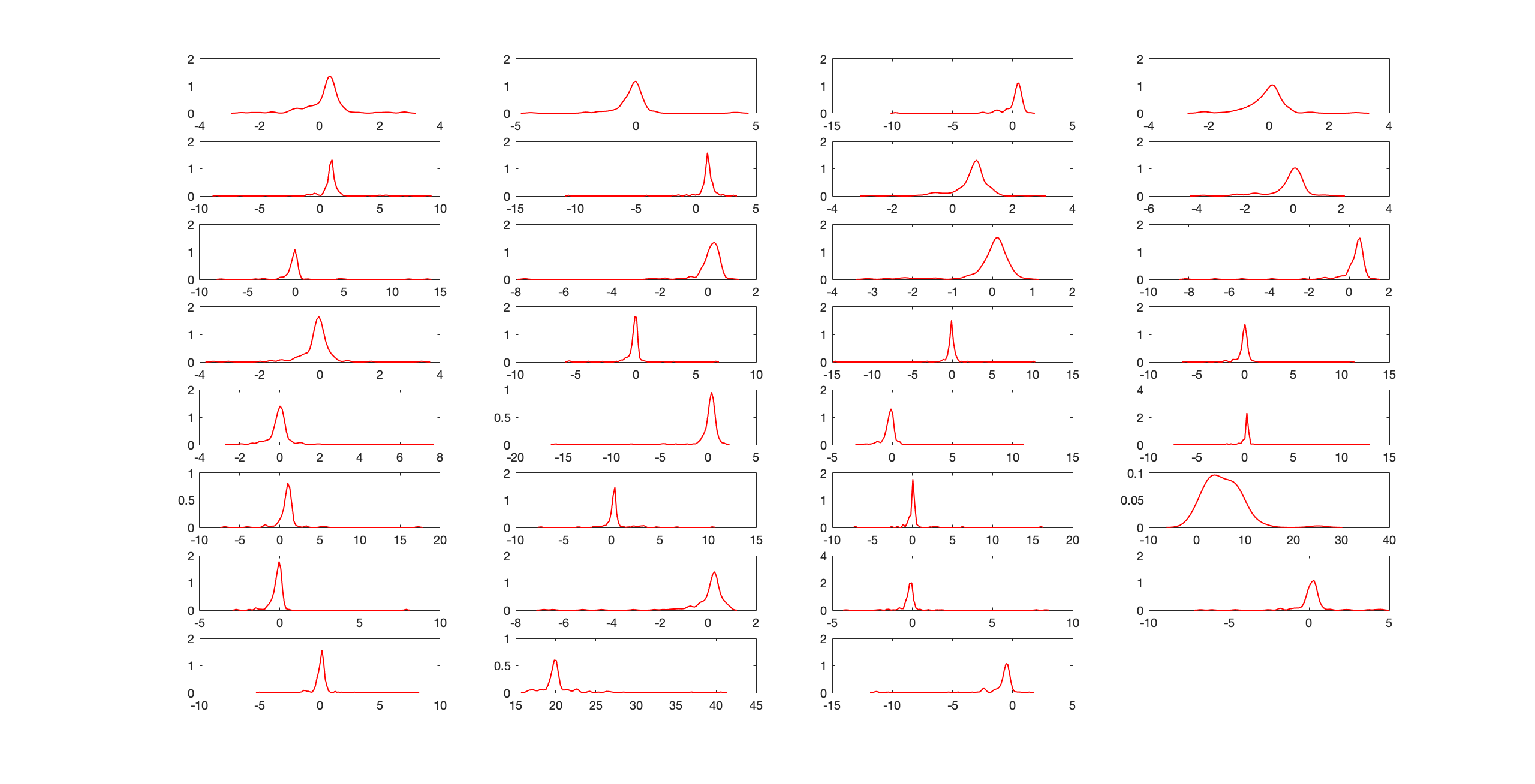}
    \footnotesize{Note: The figures show distribution of differences in cutoffs predicted by model in this paper and Melitz and Ottaviano for sectors: HS 2, HS 3, HS 4, HS 6, HS 7, HS 8, HS 9, HS 10, HS 11, HS 15, HS 16, HS 17, HS 18, HS 19, HS 20, HS 21, HS 23, HS 24, HS 30, HS 32, HS 34, HS 35, HS 36, HS 37, HS 38, HS 39, HS 40, HS 42, HS 49, HS 50, HS 51}
\end{subfigure}
\end{figure}

\pagebreak

\begin{figure}\ContinuedFloat
\caption{Distribution of Differences in Cutoffs, HS 52 - HS 97}
\label{fig:diff_dist5296}
\begin{subfigure}[b]{1\textwidth}
\includegraphics[width=1\textwidth,height=2\textheight,keepaspectratio]{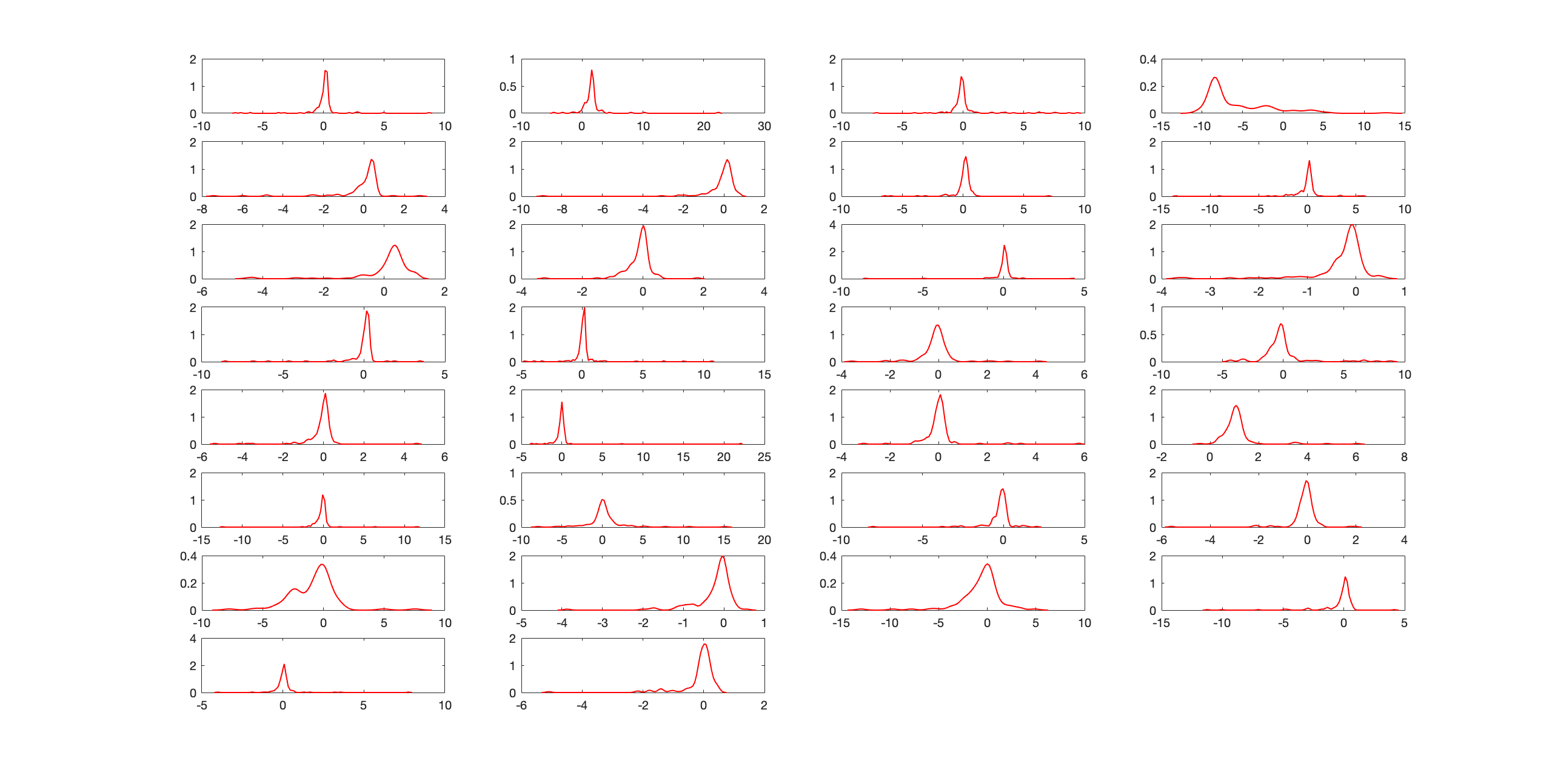}
    \floatfoot{Note: The figures show distribution of differences in cutoffs predicted by model in this paper and Melitz and Ottaviano for sectors: HS 52, HS 54, HS 55, HS 58, HS 59, HS 61, HS 62, HS 63, HS 64, HS 65, HS 66, HS 67, HS 68, HS 69, HS 70, HS 71, HS 73, HS 74, HS 76, HS 81, HS 86, HS 87, HS 88, HS 89, HS 90, HS 91, HS 92, HS 93, HS 94, HS 95}
\end{subfigure}
\end{figure}

\begin{table}
\centering
\small
\caption{International Trade Cost Increase}
\label{counter-factual1}
\begin{adjustbox}{width=0.85\textwidth,center}
\centering
\input{cf1.tex}

\end{adjustbox}
\end{table}

\begin{table}
\centering
\small
\caption{International Trade Cost Increase--MO}
\label{counter-factual1_noqua}
\begin{adjustbox}{width=0.85\textwidth,center}
\input{cf1_noqua.tex}

\end{adjustbox}
\end{table}

\begin{figure}
\centering
\caption{Change in Cutoffs and Quality Preferences, HS 2 - HS 51}
\label{fig:diff_beta}
\begin{subfigure}[b]{1\textwidth}	\includegraphics[width=\textwidth,height=1.4\textheight,keepaspectratio]{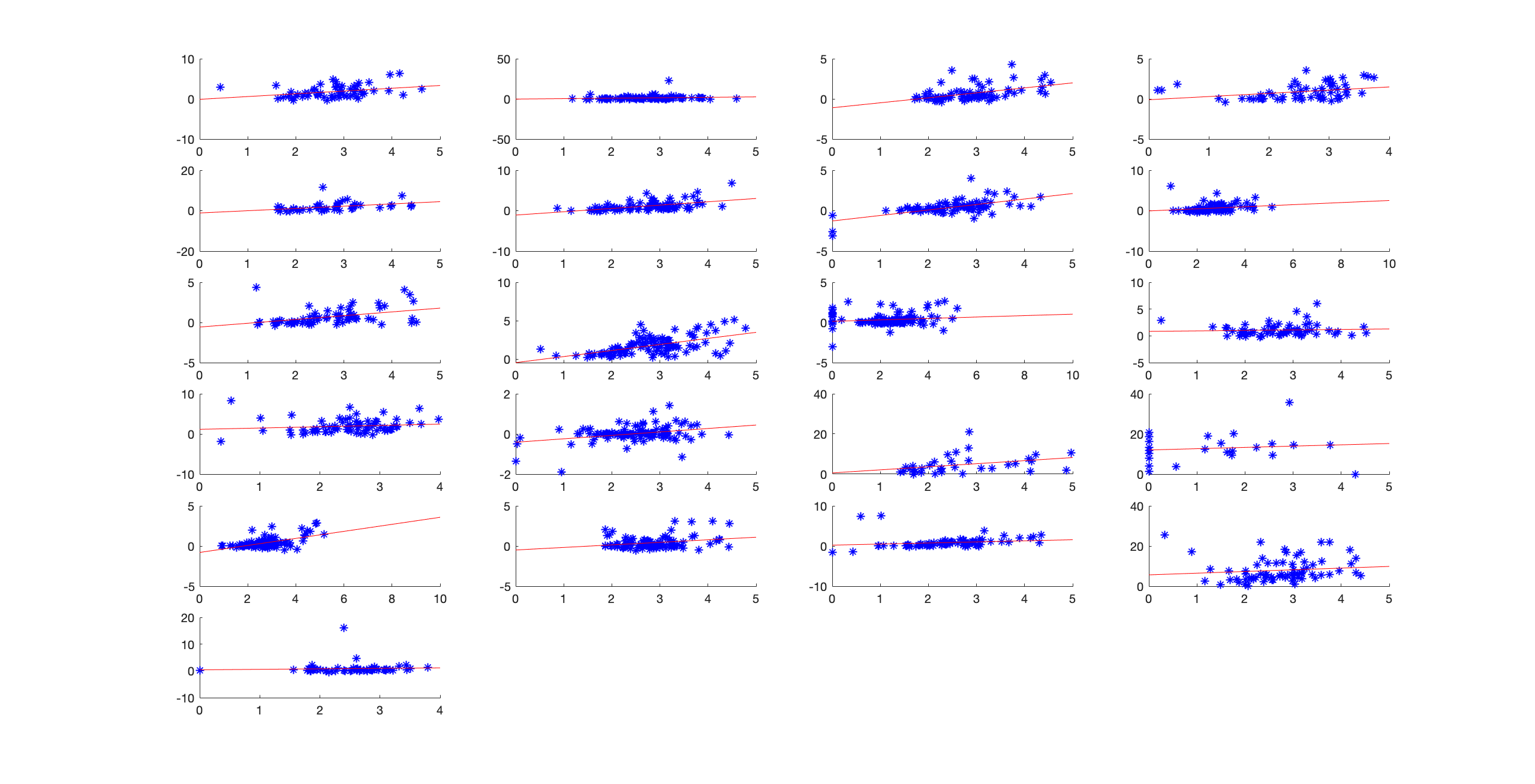}
    \footnotesize{Note: The figures show positive correlations between quality preferences and difference in loss from trade barrier between the model in this paper and Melitz and Ottaviano for sectors: HS 2, HS 3, HS 4, HS 7, HS 8, HS 9, HS 11, HS 15, HS 16, HS 17, HS 23, HS 24, HS 30, HS 35, HS 36, HS 37, HS 38, HS 39, HS 40, HS 42, HS 51}
\end{subfigure}
\end{figure}

\pagebreak

\begin{figure}\ContinuedFloat
\caption{Change in Cutoffs and Quality Preferences, HS 52 - HS 97}
\label{fig:diff_beta5296}
\begin{subfigure}[b]{1\textwidth}
\includegraphics[width=\textwidth,height=2\textheight,keepaspectratio]{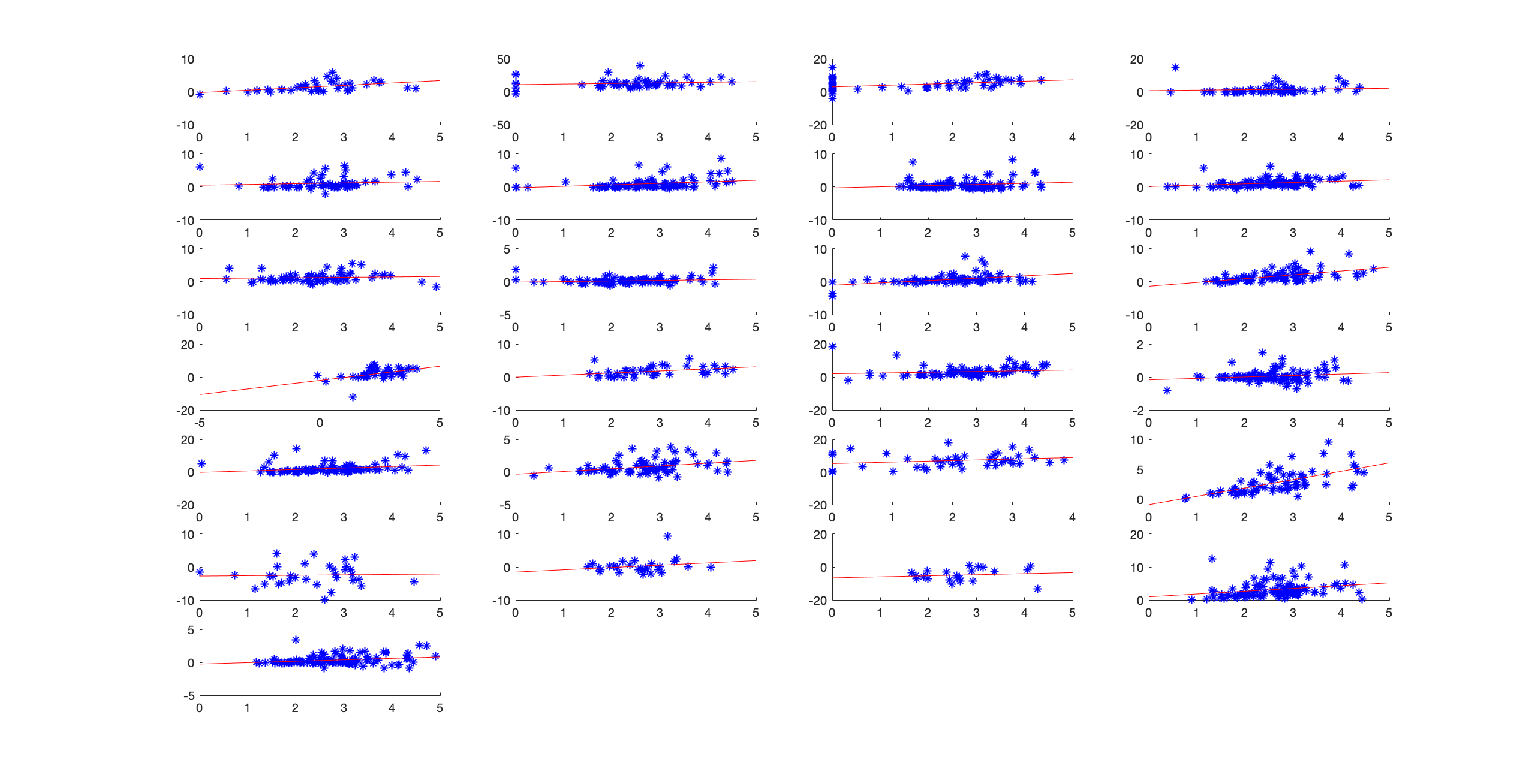}	
    \floatfoot{Note: The figures show positive correlations between quality preferences and difference in loss from trade barrier between the model in this paper and Melitz and Ottaviano: HS 52, HS 54, HS 58, HS 61, HS 62, HS 63, HS 64, HS 65, HS 66, HS 68, HS 69, HS 70, HS 71, HS 73, HS 74, HS 76, HS 82, HS 87, HS 88, HS 90, HS 91, HS 92, HS 93, HS 95, HS 96}
\end{subfigure}
\end{figure}

\begin{table}[htbp]
\caption{Summary of Cutoff Changes} \label{cf2}

\begin{subtable}{1\textwidth}
\caption{Summary of Cutoff Changes}
\label{cf21}
\centering
\begin{adjustbox}{width=0.8\textwidth,center}
\small
\input{cf2_1.tex}

\end{adjustbox}
\end{subtable}
\end{table}

\pagebreak

\begin{table}[htbp]
\ContinuedFloat

\begin{subtable}{1\textwidth}
\caption{Summary of Cutoff Changes (continue)}
\label{cf22}
\centering
\begin{adjustbox}{width=0.8\textwidth,center}
\small
\input{cf2_2.tex}

\end{adjustbox}
\end{subtable}

\end{table}

\begin{sidewaysfigure}
\includegraphics[scale=0.5,height=10cm]{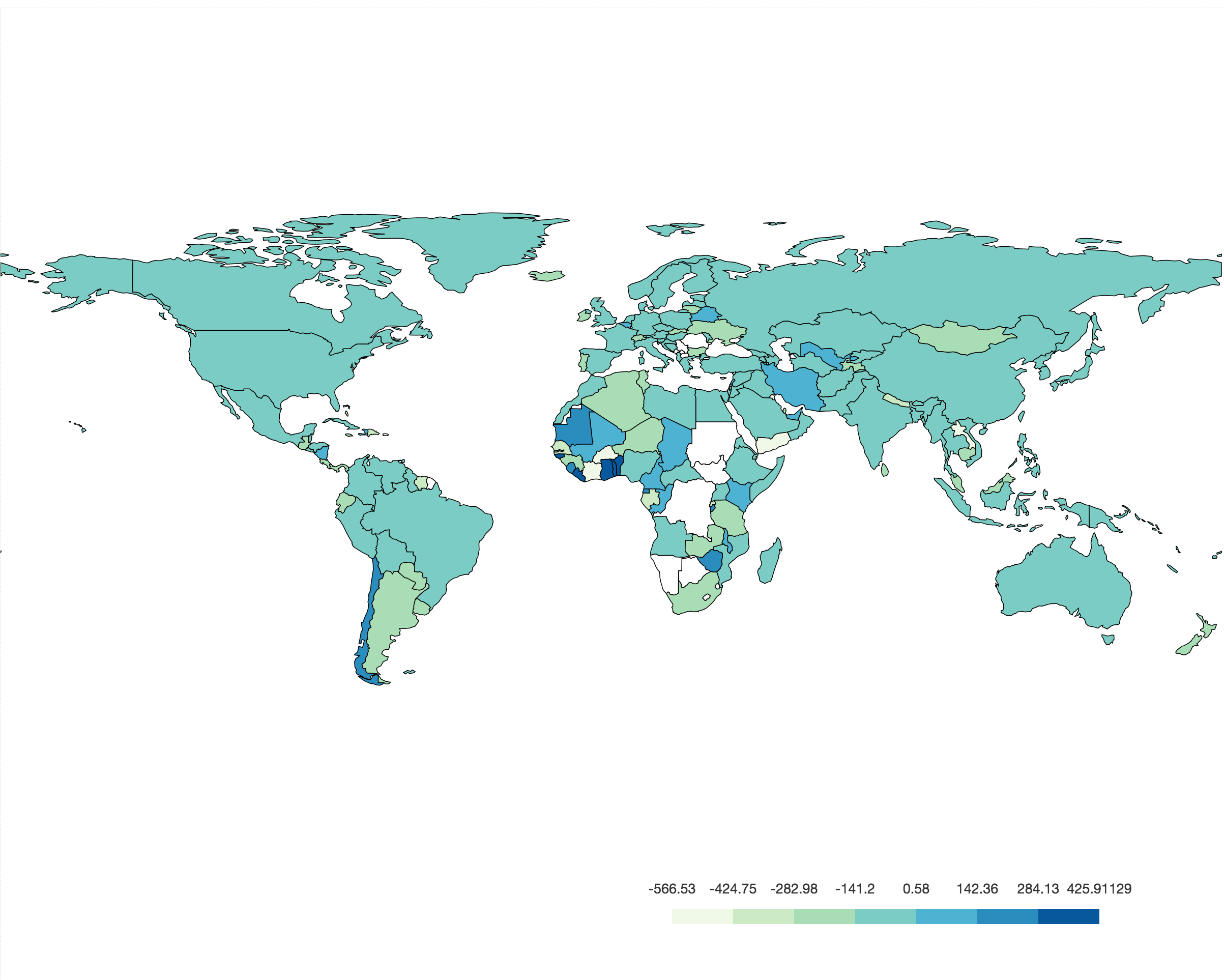}
\caption{Changes in Cutoffs Across Countries}
\label{fig:cutoff_change_distribution}
\floatfoot{The figure illustrates the distribution of rise/fall in endogenous competitiveness across countries}
\end{sidewaysfigure}

\begin{table}
\caption{Gains from Preference Shock and Country Characteristics}
\input{size-gain.tex}

\label{size-gain}

\end{table}

\end{document}

%% file: stylized1.tex
{\tabcolsep=1.3mm
    \begin{tabular}{l||cc|cc}
    \toprule
    \toprule
          & \multicolumn{4}{c}{Dependent Variable: $\ln (price)$} \\
\cmidrule{2-5}          & \multicolumn{2}{c}{$\ln(p_{fhc})$} & \multicolumn{2}{c}{$\ln(p_{hc})$} \\
\cmidrule{2-5}          & (1)   & (2)   & (3)   & (4) \\
    \midrule
    (log) Population   & 0.035*** & 0.051*** & 0.208*** & 0.245*** \\
          & (0.000) & (0.000) & (0.001) & (0.001) \\
    \midrule
    Country-level Controls & No    & Yes   & No    & Yes \\
    Firm FE & Yes   & Yes   & No    & No \\
    Product FE & Yes   & Yes   & Yes   & Yes \\
    Observations & 11,224,467 & 11,038,879 & 449,135 & 440,783 \\
    R-Squred & 0.697 & 0.698 & 0.763 & 0.780 \\
    \bottomrule
    \bottomrule
     \multicolumn{5}{l}{\shortstack[l]{\parbox{12cm}{\vspace{0.1cm} \scriptsize
Notes: *** p < 0.01, ** p < 0.05, * p < 0.1. Robust standard errors are in parentheses. The dependent variable in specifications (1)-(2) is the (log) price at the firm-HS8-country level, and in specifications (3)-(4) is the (log) price at the HS8-country level. Country-level other controls include GDP per capita and distance. All regressions include a constant term.
}}
}
    \end{tabular}}%

%% file: stylized2.tex
{\tabcolsep=1.3mm
    \begin{tabular}{l||cc|cc}
    \toprule
    \toprule
          & \multicolumn{4}{c}{Dependent Variable: $\ln (price)$} \\
\cmidrule{2-5}          & \multicolumn{2}{c}{$\ln(p_{fhc})$} & \multicolumn{2}{c}{$\ln(p_{fh})$} \\
\cmidrule{2-5}          & (1)   & (2)   & (3)   & (4) \\
    \midrule
    $\ln(Revenue)$ & 0.154*** & 0.155*** & 0.206*** & 0.210*** \\
          & (0.000) & (0.000) & (0.000) & (0.000) \\
    $\ln(Revenue)*RD\_Intensity$ &       & 0.048*** &       & 0.009 \\
          &       & (0.005) &       & (0.009) \\
    \midrule
    Firm FE & Yes   & Yes   & Yes   & Yes \\
    Product FE & Yes   & Yes   & Yes   & Yes \\
    Observations & 12,476,096 & 11,070,256 & 4,146,176 & 3,660,845 \\
    R-Squred & 0.696 & 0.695 & 0.722 & 0.725 \\
    \bottomrule
    \bottomrule
    
    \multicolumn{5}{l}{\shortstack[l]{\parbox{13cm}{\vspace{0.1cm} \scriptsize
Notes: *** p < 0.01, ** p < 0.05, * p < 0.1. Robust standard errors are in parentheses. The dependent variable in specifications (1)-(2) is the (log) price at the firm-HS8-country level, and in specifications (3)-(4) is the (log) price at the firm-HS8 level. $RD\_Intensity$ is compiled by \textcite{KROSZNER2007187} at ISIC level, which can be converted to HS 6 codes. All regressions include a constant term.
}}
}
    \end{tabular}}%

%% file: stylized3.tex
{\tabcolsep=1.3mm
    \begin{tabular}{l||c|c|cc}
    \toprule
    \toprule
          & \multicolumn{4}{c}{Dependent Variable: $\ln (price)$} \\
\cmidrule{2-5}          & \multicolumn{2}{c}{$\ln(p_{fh})$} & \multicolumn{2}{c}{$std.(\ln(p_{fh}))$} \\
\cmidrule{2-5}          & (1)   & (2)   & (3)   & (4) \\
    \midrule
    Num\_Destinations & 0.007*** & 0.006*** & 0.006*** & 0.005*** \\
          & (0.000) & (0.000) & (0.000) & (0.000) \\
    Num \_Destinations*RD\_Intensity &       & 0.009*** &       & 0.015*** \\
          &       & (0.001) &       & (0.000) \\
    \midrule
    Firm FE & Yes   & Yes   & Yes   & Yes \\
    Product FE & Yes   & Yes   & Yes   & Yes \\
    Observations & 11,622,117 & 10,314,550 & 9,069,247 & 8,065,339 \\
    R-Squred & 0.681 & 0.680 & 0.459 & 0.463 \\
    \bottomrule
    \bottomrule
    
    \multicolumn{5}{l}{\shortstack[l]{\parbox{15cm}{\vspace{0.1cm} \scriptsize
Notes: *** p < 0.01, ** p < 0.05, * p < 0.1. Robust standard errors are in parentheses. The dependent variable in specifications (1)-(2) is the (log) price at the firm-HS8 level, and in specifications (3)-(4) is the standard deviation of (log) prices at the firm-HS8 level. $RD\_Intensity$ is compiled by \textcite{KROSZNER2007187} at ISIC level, which can be converted to HS 6 codes. All regressions include a constant term.
}}
}
    \end{tabular}}%

%% file: summary1.tex
{\tabcolsep=1.3mm
    \begin{tabular}{llrllrllrll}
    \toprule
    HS 2  & Obs   &       & HS 2  & Obs   &       & HS 2  & Obs   &       & HS 2  & Obs \\
    \midrule
    2     & 30,707 &       & 24    & 6,503 &       & 51    & 1,780 &       & 73    & 27,812 \\
    3     & 58,406 &       & 30    & 28,032 &       & 52    & 3,047 &       & 74    & 6,489 \\
    4     & 27,854 &       & 32    & 3,420 &       & 54    & 900   &       & 76    & 3,604 \\
    6     & 2,998 &       & 33    & 58,712 &       & 55    & 1,928 &       & 82    & 31,914 \\
    7     & 53,205 &       & 34    & 30,956 &       & 57    & 29,694 &       & 83    & 6,406 \\
    8     & 53,451 &       & 35    & 3,713 &       & 58    & 500   &       & 84    & 47,596 \\
    9     & 36,534 &       & 36    & 1,750 &       & 59    & 1,035 &       & 85    & 91,988 \\
    10    & 3,848 &       & 37    & 902   &       & 61    & 234,282 &       & 87    & 21,361 \\
    11    & 4,079 &       & 38    & 15,284 &       & 62    & 265,844 &       & 88    & 1,515 \\
    15    & 13,055 &       & 39    & 31,740 &       & 63    & 74,382 &       & 89    & 8,254 \\
    16    & 29,784 &       & 40    & 16,135 &       & 64    & 66,311 &       & 90    & 31,572 \\
    17    & 11,371 &       & 42    & 59,404 &       & 65    & 14,244 &       & 91    & 29,964 \\
    18    & 12,387 &       & 43    & 3,802 &       & 66    & 6,506 &       & 92    & 16,893 \\
    19    & 39,318 &       & 44    & 12,643 &       & 67    & 4,578 &       & 93    & 3,776 \\
    20    & 76,252 &       & 46    & 6,690 &       & 68    & 1,477 &       & 94    & 61,268 \\
    21    & 34,332 &       & 48    & 46,325 &       & 69    & 16,443 &       & 95    & 52,143 \\
    22    & 47,661 &       & 49    & 31,867 &       & 70    & 19,589 &       & 96    & 46,185 \\
    23    & 2,307 &       & 50    & 395   &       & 71    & 22,075 &       & 97    & 11,678 \\
    \bottomrule
    \multicolumn{11}{l}{\shortstack[l]{\parbox{10.5cm}{\vspace{0.1cm} \scriptsize
Note: This table summarizes the number of bilateral trade transaction observations  within each HS 2 sector. Raw Data source is from CEPII
}}
}

    \end{tabular}}%

%% file: summary2.tex
    \begin{tabular}{llllllllllllll}
    \hline
    Country & \multicolumn{1}{l}{Obs} &       & \multicolumn{1}{l}{Country} & \multicolumn{1}{l}{Obs} &       & \multicolumn{1}{l}{Country} & \multicolumn{1}{l}{Obs} &       & \multicolumn{1}{l}{Country} & \multicolumn{1}{l}{Obs} &       & \multicolumn{1}{l}{Country} & \multicolumn{1}{l}{Obs} \\
    \hline
    AFG   & 902   &       & \multicolumn{1}{l}{CCK} & 86    &       & \multicolumn{1}{l}{GUY} & 935   &       & \multicolumn{1}{l}{MAR} & 1,194 &       & \multicolumn{1}{l}{STP} & 687 \\
    ALB   & 1,063 &       & \multicolumn{1}{l}{COL} & 1,103 &       & \multicolumn{1}{l}{HTI} & 780   &       & \multicolumn{1}{l}{MOZ} & 1,123 &       & \multicolumn{1}{l}{SAU} & 1,152 \\
    DZA   & 1,094 &       & \multicolumn{1}{l}{COM} & 797   &       & \multicolumn{1}{l}{HND} & 966   &       & \multicolumn{1}{l}{OMN} & 1,136 &       & \multicolumn{1}{l}{SEN} & 1,045 \\
    ASM   & 250   &       & \multicolumn{1}{l}{COG} & 1,064 &       & \multicolumn{1}{l}{HKG} & 1,195 &       & \multicolumn{1}{l}{NRU} & 364   &       & \multicolumn{1}{l}{SYC} & 868 \\
    AND   & 1,030 &       & \multicolumn{1}{l}{ZAR} & 998   &       & \multicolumn{1}{l}{HUN} & 1,171 &       & \multicolumn{1}{l}{NPL} & 1,010 &       & \multicolumn{1}{l}{SLE} & 701 \\
    AGO   & 1,127 &       & \multicolumn{1}{l}{COK} & 692   &       & \multicolumn{1}{l}{ISL} & 1,119 &       & \multicolumn{1}{l}{NLD} & 1,201 &       & \multicolumn{1}{l}{IND} & 1,107 \\
    ATG   & 992   &       & \multicolumn{1}{l}{CRI} & 1,108 &       & \multicolumn{1}{l}{IDN} & 1,148 &       & \multicolumn{1}{l}{ABW} & 1,085 &       & \multicolumn{1}{l}{SGP} & 1,184 \\
    AZE   & 1,096 &       & \multicolumn{1}{l}{HRV} & 1,157 &       & \multicolumn{1}{l}{IRN} & 945   &       & \multicolumn{1}{l}{NCL} & 1,094 &       & \multicolumn{1}{l}{SVK} & 1,177 \\
    ARG   & 1,036 &       & \multicolumn{1}{l}{CUB} & 863   &       & \multicolumn{1}{l}{IRQ} & 1,041 &       & \multicolumn{1}{l}{VUT} & 720   &       & \multicolumn{1}{l}{VNM} & 1,145 \\
    AUS   & 1,172 &       & \multicolumn{1}{l}{CYP} & 1,156 &       & \multicolumn{1}{l}{IRL} & 1,182 &       & \multicolumn{1}{l}{NZL} & 1,142 &       & \multicolumn{1}{l}{SVN} & 1,160 \\
    AUT   & 1,186 &       & \multicolumn{1}{l}{CZE} & 1,181 &       & \multicolumn{1}{l}{ISR} & 1,139 &       & \multicolumn{1}{l}{NIC} & 1,015 &       & \multicolumn{1}{l}{SOM} & 611 \\
    BHS   & 1,123 &       & \multicolumn{1}{l}{BEN} & 882   &       & \multicolumn{1}{l}{ITA} & 1,203 &       & \multicolumn{1}{l}{NER} & 821   &       & \multicolumn{1}{l}{ZAF} & 1,164 \\
    BHR   & 1,155 &       & \multicolumn{1}{l}{DNK} & 1,192 &       & \multicolumn{1}{l}{CIV} & 1,093 &       & \multicolumn{1}{l}{NGA} & 1,130 &       & \multicolumn{1}{l}{ZWE} & 1,054 \\
    BGD   & 931   &       & \multicolumn{1}{l}{DMA} & 488   &       & \multicolumn{1}{l}{JAM} & 1,082 &       & \multicolumn{1}{l}{NIU} & 299   &       & \multicolumn{1}{l}{ESP} & 1,209 \\
    ARM   & 1,019 &       & \multicolumn{1}{l}{DOM} & 1,132 &       & \multicolumn{1}{l}{JPN} & 1,214 &       & \multicolumn{1}{l}{NFK} & 289   &       & \multicolumn{1}{l}{SUR} & 958 \\
    BRB   & 1,141 &       & \multicolumn{1}{l}{ECU} & 1,016 &       & \multicolumn{1}{l}{KAZ} & 1,157 &       & \multicolumn{1}{l}{NOR} & 1,171 &       & \multicolumn{1}{l}{SWE} & 1,187 \\
    BEL   & 1,207 &       & \multicolumn{1}{l}{SLV} & 1,062 &       & \multicolumn{1}{l}{JOR} & 1,125 &       & \multicolumn{1}{l}{MNP} & 294   &       & \multicolumn{1}{l}{CHE} & 1,187 \\
    BMU   & 1,104 &       & \multicolumn{1}{l}{GNQ} & 827   &       & \multicolumn{1}{l}{KEN} & 1,085 &       & \multicolumn{1}{l}{FSM} & 596   &       & \multicolumn{1}{l}{SYR} & 823 \\
    BTN   & 246   &       & \multicolumn{1}{l}{ETH} & 1,020 &       & \multicolumn{1}{l}{PRK} & 747   &       & \multicolumn{1}{l}{MHL} & 306   &       & \multicolumn{1}{l}{TJK} & 743 \\
    BOL   & 969   &       & \multicolumn{1}{l}{ERI} & 302   &       & \multicolumn{1}{l}{KOR} & 1,180 &       & \multicolumn{1}{l}{PLW} & 707   &       & \multicolumn{1}{l}{THA} & 1,180 \\
    BIH   & 1,092 &       & \multicolumn{1}{l}{EST} & 1,163 &       & \multicolumn{1}{l}{KWT} & 1,149 &       & \multicolumn{1}{l}{PAK} & 1,014 &       & \multicolumn{1}{l}{TGO} & 896 \\
    BRA   & 1,128 &       & \multicolumn{1}{l}{FLK} & 382   &       & \multicolumn{1}{l}{KGZ} & 983   &       & \multicolumn{1}{l}{PAN} & 1,143 &       & \multicolumn{1}{l}{TKL} & 95 \\
    BLZ   & 897   &       & \multicolumn{1}{l}{FJI} & 1,038 &       & \multicolumn{1}{l}{LAO} & 733   &       & \multicolumn{1}{l}{PNG} & 948   &       & \multicolumn{1}{l}{TON} & 732 \\
    IOT   & 32    &       & \multicolumn{1}{l}{FIN} & 1,171 &       & \multicolumn{1}{l}{LBN} & 1,143 &       & \multicolumn{1}{l}{PRY} & 992   &       & \multicolumn{1}{l}{TTO} & 989 \\
    SLB   & 531   &       & \multicolumn{1}{l}{FRA} & 1,223 &       & \multicolumn{1}{l}{LVA} & 1,164 &       & \multicolumn{1}{l}{PER} & 1,071 &       & \multicolumn{1}{l}{ARE} & 1,187 \\
    VGB   & 630   &       & \multicolumn{1}{l}{PYF} & 1,066 &       & \multicolumn{1}{l}{LBR} & 741   &       & \multicolumn{1}{l}{PHL} & 1,213 &       & \multicolumn{1}{l}{TUN} & 1,076 \\
    BRN   & 1,076 &       & \multicolumn{1}{l}{ATF} & 164   &       & \multicolumn{1}{l}{LBY} & 1,021 &       & \multicolumn{1}{l}{PCN} & 21    &       & \multicolumn{1}{l}{TUR} & 1,143 \\
    BGR   & 1,165 &       & \multicolumn{1}{l}{DJI} & 796   &       & \multicolumn{1}{l}{LTU} & 1,160 &       & \multicolumn{1}{l}{POL} & 1,181 &       & \multicolumn{1}{l}{TKM} & 888 \\
    MMR   & 936   &       & \multicolumn{1}{l}{GAB} & 939   &       & \multicolumn{1}{l}{MAC} & 1,047 &       & \multicolumn{1}{l}{PRT} & 1,192 &       & \multicolumn{1}{l}{TCA} & 519 \\
    BDI   & 698   &       & \multicolumn{1}{l}{GEO} & 1,093 &       & \multicolumn{1}{l}{MDG} & 898   &       & \multicolumn{1}{l}{GNB} & 462   &       & \multicolumn{1}{l}{TUV} & 231 \\
    BLR   & 1,132 &       & \multicolumn{1}{l}{GMB} & 720   &       & \multicolumn{1}{l}{MWI} & 952   &       & \multicolumn{1}{l}{TMP} & 802   &       & \multicolumn{1}{l}{UGA} & 984 \\
    KHM   & 968   &       & \multicolumn{1}{l}{PAL} & 804   &       & \multicolumn{1}{l}{MYS} & 1,184 &       & \multicolumn{1}{l}{QAT} & 1,159 &       & \multicolumn{1}{l}{UKR} & 1,158 \\
    CMR   & 894   &       & \multicolumn{1}{l}{DEU} & 1,208 &       & \multicolumn{1}{l}{MDV} & 1,046 &       & \multicolumn{1}{l}{ROM} & 1,178 &       & \multicolumn{1}{l}{EGY} & 1,112 \\
    CAN   & 1,198 &       & \multicolumn{1}{l}{GHA} & 1,112 &       & \multicolumn{1}{l}{MLI} & 762   &       & \multicolumn{1}{l}{RUS} & 1,187 &       & \multicolumn{1}{l}{GBR} & 1,205 \\
    CPV   & 919   &       & \multicolumn{1}{l}{GIB} & 970   &       & \multicolumn{1}{l}{MLT} & 1,161 &       & \multicolumn{1}{l}{RWA} & 884   &       & \multicolumn{1}{l}{TZA} & 1,035 \\
    CYM   & 751   &       & \multicolumn{1}{l}{KIR} & 549   &       & \multicolumn{1}{l}{MRT} & 843   &       & \multicolumn{1}{l}{SHN} & 354   &       & \multicolumn{1}{l}{USA} & 1,238 \\
    CAF   & 343   &       & \multicolumn{1}{l}{GRC} & 1,184 &       & \multicolumn{1}{l}{MUS} & 1,120 &       & \multicolumn{1}{l}{KNA} & 501   &       & \multicolumn{1}{l}{BFA} & 854 \\
    LKA   & 1,047 &       & \multicolumn{1}{l}{GRL} & 1,000 &       & \multicolumn{1}{l}{MEX} & 1,168 &       & \multicolumn{1}{l}{AIA} & 314   &       & \multicolumn{1}{l}{URY} & 1,065 \\
    TCD   & 589   &       & \multicolumn{1}{l}{GRD} & 538   &       & \multicolumn{1}{l}{TWN} & 1,167 &       & \multicolumn{1}{l}{LCA} & 679   &       & \multicolumn{1}{l}{UZB} & 874 \\
    CHL   & 1,113 &       & \multicolumn{1}{l}{GUM} & 583   &       & \multicolumn{1}{l}{MNG} & 1,007 &       & \multicolumn{1}{l}{SPM} & 538   &       & \multicolumn{1}{l}{VEN} & 1,100 \\
    CHN   & 1,182 &       & \multicolumn{1}{l}{GTM} & 1,103 &       & \multicolumn{1}{l}{MDA} & 1,078 &       & \multicolumn{1}{l}{VCT} & 499   &       & \multicolumn{1}{l}{WLF} & 451 \\
    CXR   & 261   &       & \multicolumn{1}{l}{GIN} & 869   &       & \multicolumn{1}{l}{MSR} & 164   &       & \multicolumn{1}{l}{SMR} & 389   &       & \multicolumn{1}{l}{WSM} & 806 \\
    YEM   & 950   &       &  \multicolumn{1}{l}{ZMB}     &  1,063     &       &       &       &       &       &       &       &       &  \\

    \hline
    
        \multicolumn{14}{l}{\shortstack[l]{\parbox{17.5cm}{\vspace{0.1cm} \scriptsize
Note: This table summarizes the number of import transaction observations  for each country. Raw Data source is from CEPII
}}
}
    
    \end{tabular}%

%% file: summary3.tex

    \begin{tabular}{llllllllllllll}
    \hline
    Country & \multicolumn{1}{l}{Obs} &       & \multicolumn{1}{l}{Country} & \multicolumn{1}{l}{Obs} &       & \multicolumn{1}{l}{Country} & \multicolumn{1}{l}{Obs} &       & \multicolumn{1}{l}{Country} & \multicolumn{1}{l}{Obs} &       & \multicolumn{1}{l}{Country} & \multicolumn{1}{l}{Obs} \\
    \hline
    AFG   & 277   &       & \multicolumn{1}{l}{CCK} & 34    &       & \multicolumn{1}{l}{GUY} & 280   &       & \multicolumn{1}{l}{MAR} & 1,024 &       & \multicolumn{1}{l}{STP} & 53 \\
    ALB   & 622   &       & \multicolumn{1}{l}{COL} & 973   &       & \multicolumn{1}{l}{HTI} & 239   &       & \multicolumn{1}{l}{MOZ} & 360   &       & \multicolumn{1}{l}{SAU} & 888 \\
    DZA   & 296   &       & \multicolumn{1}{l}{COM} & 118   &       & \multicolumn{1}{l}{HND} & 635   &       & \multicolumn{1}{l}{OMN} & 706   &       & \multicolumn{1}{l}{SEN} & 588 \\
    ASM   & 128   &       & \multicolumn{1}{l}{COG} & 178   &       & \multicolumn{1}{l}{HKG} & 1,163 &       & \multicolumn{1}{l}{NRU} & 80    &       & \multicolumn{1}{l}{SYC} & 208 \\
    AND   & 387   &       & \multicolumn{1}{l}{ZAR} & 240   &       & \multicolumn{1}{l}{HUN} & 1,124 &       & \multicolumn{1}{l}{NPL} & 571   &       & \multicolumn{1}{l}{SLE} & 241 \\
    AGO   & 158   &       & \multicolumn{1}{l}{COK} & 67    &       & \multicolumn{1}{l}{ISL} & 626   &       & \multicolumn{1}{l}{NLD} & 1,212 &       & \multicolumn{1}{l}{IND} & 1,167 \\
    ATG   & 219   &       & \multicolumn{1}{l}{CRI} & 881   &       & \multicolumn{1}{l}{IDN} & 1,135 &       & \multicolumn{1}{l}{ABW} & 201   &       & \multicolumn{1}{l}{SGP} & 1,180 \\
    AZE   & 414   &       & \multicolumn{1}{l}{HRV} & 1,029 &       & \multicolumn{1}{l}{IRN} & 737   &       & \multicolumn{1}{l}{NCL} & 377   &       & \multicolumn{1}{l}{SVK} & 1,095 \\
    ARG   & 1,027 &       & \multicolumn{1}{l}{CUB} & 186   &       & \multicolumn{1}{l}{IRQ} & 198   &       & \multicolumn{1}{l}{VUT} & 106   &       & \multicolumn{1}{l}{VNM} & 1,105 \\
    AUS   & 1,190 &       & \multicolumn{1}{l}{CYP} & 923   &       & \multicolumn{1}{l}{IRL} & 1,150 &       & \multicolumn{1}{l}{NZL} & 1,145 &       & \multicolumn{1}{l}{SVN} & 1,092 \\
    AUT   & 1,168 &       & \multicolumn{1}{l}{CZE} & 1,173 &       & \multicolumn{1}{l}{ISR} & 1,004 &       & \multicolumn{1}{l}{NIC} & 497   &       & \multicolumn{1}{l}{SOM} & 75 \\
    BHS   & 301   &       & \multicolumn{1}{l}{BEN} & 230   &       & \multicolumn{1}{l}{ITA} & 1,231 &       & \multicolumn{1}{l}{NER} & 242   &       & \multicolumn{1}{l}{ZAF} & 1,186 \\
    BHR   & 839   &       & \multicolumn{1}{l}{DNK} & 1,185 &       & \multicolumn{1}{l}{CIV} & 564   &       & \multicolumn{1}{l}{NGA} & 608   &       & \multicolumn{1}{l}{ZWE} & 466 \\
    BGD   & 755   &       & \multicolumn{1}{l}{DMA} & 152   &       & \multicolumn{1}{l}{JAM} & 537   &       & \multicolumn{1}{l}{NIU} & 25    &       & \multicolumn{1}{l}{ESP} & 1,242 \\
    ARM   & 509   &       & \multicolumn{1}{l}{DOM} & 854   &       & \multicolumn{1}{l}{JPN} & 1,172 &       & \multicolumn{1}{l}{NFK} & 7     &       & \multicolumn{1}{l}{SUR} & 331 \\
    BRB   & 505   &       & \multicolumn{1}{l}{ECU} & 786   &       & \multicolumn{1}{l}{KAZ} & 743   &       & \multicolumn{1}{l}{NOR} & 1,090 &       & \multicolumn{1}{l}{SWE} & 1,167 \\
    BEL   & 1,204 &       & \multicolumn{1}{l}{SLV} & 765   &       & \multicolumn{1}{l}{JOR} & 804   &       & \multicolumn{1}{l}{MNP} & 110   &       & \multicolumn{1}{l}{CHE} & 1,120 \\
    BMU   & 130   &       & \multicolumn{1}{l}{GNQ} & 25    &       & \multicolumn{1}{l}{KEN} & 937   &       & \multicolumn{1}{l}{FSM} & 44    &       & \multicolumn{1}{l}{SYR} & 683 \\
    BTN   & 44    &       & \multicolumn{1}{l}{ETH} & 508   &       & \multicolumn{1}{l}{PRK} & 432   &       & \multicolumn{1}{l}{MHL} & 62    &       & \multicolumn{1}{l}{TJK} & 230 \\
    BOL   & 443   &       & \multicolumn{1}{l}{ERI} & 50    &       & \multicolumn{1}{l}{KOR} & 1,152 &       & \multicolumn{1}{l}{PLW} & 18    &       & \multicolumn{1}{l}{THA} & 1,179 \\
    BIH   & 801   &       & \multicolumn{1}{l}{EST} & 1,095 &       & \multicolumn{1}{l}{KWT} & 814   &       & \multicolumn{1}{l}{PAK} & 1,014 &       & \multicolumn{1}{l}{TGO} & 478 \\
    BRA   & 1,078 &       & \multicolumn{1}{l}{FLK} & 34    &       & \multicolumn{1}{l}{KGZ} & 492   &       & \multicolumn{1}{l}{PAN} & 960   &       & \multicolumn{1}{l}{TKL} & 146 \\
    BLZ   & 241   &       & \multicolumn{1}{l}{FJI} & 804   &       & \multicolumn{1}{l}{LAO} & 317   &       & \multicolumn{1}{l}{PNG} & 148   &       & \multicolumn{1}{l}{TON} & 84 \\
    IOT   & 33    &       & \multicolumn{1}{l}{FIN} & 1,086 &       & \multicolumn{1}{l}{LBN} & 1,002 &       & \multicolumn{1}{l}{PRY} & 417   &       & \multicolumn{1}{l}{TTO} & 507 \\
    SLB   & 47    &       & \multicolumn{1}{l}{FRA} & 1,254 &       & \multicolumn{1}{l}{LVA} & 1,119 &       & \multicolumn{1}{l}{PER} & 968   &       & \multicolumn{1}{l}{ARE} & 1,173 \\
    VGB   & 198   &       & \multicolumn{1}{l}{PYF} & 297   &       & \multicolumn{1}{l}{LBR} & 71    &       & \multicolumn{1}{l}{PHL} & 1,076 &       & \multicolumn{1}{l}{TUN} & 903 \\
    BRN   & 427   &       & \multicolumn{1}{l}{ATF} & 16    &       & \multicolumn{1}{l}{LBY} & 142   &       & \multicolumn{1}{l}{PCN} & 26    &       & \multicolumn{1}{l}{TUR} & 1,152 \\
    BGR   & 1,109 &       & \multicolumn{1}{l}{DJI} & 72    &       & \multicolumn{1}{l}{LTU} & 1,138 &       & \multicolumn{1}{l}{POL} & 1,177 &       & \multicolumn{1}{l}{TKM} & 141 \\
    MMR   & 508   &       & \multicolumn{1}{l}{GAB} & 213   &       & \multicolumn{1}{l}{MAC} & 633   &       & \multicolumn{1}{l}{PRT} & 1,195 &       & \multicolumn{1}{l}{TCA} & 129 \\
    BDI   & 115   &       & \multicolumn{1}{l}{GEO} & 670   &       & \multicolumn{1}{l}{MDG} & 613   &       & \multicolumn{1}{l}{GNB} & 33    &       & \multicolumn{1}{l}{TUV} & 27 \\
    BLR   & 949   &       & \multicolumn{1}{l}{GMB} & 163   &       & \multicolumn{1}{l}{MWI} & 274   &       & \multicolumn{1}{l}{TMP} & 69    &       & \multicolumn{1}{l}{UGA} & 627 \\
    KHM   & 616   &       & \multicolumn{1}{l}{PAL} & 428   &       & \multicolumn{1}{l}{MYS} & 1,156 &       & \multicolumn{1}{l}{QAT} & 484   &       & \multicolumn{1}{l}{UKR} & 998 \\
    CMR   & 421   &       & \multicolumn{1}{l}{DEU} & 1,226 &       & \multicolumn{1}{l}{MDV} & 153   &       & \multicolumn{1}{l}{ROM} & 1,096 &       & \multicolumn{1}{l}{EGY} & 1,015 \\
    CAN   & 1,186 &       & \multicolumn{1}{l}{GHA} & 645   &       & \multicolumn{1}{l}{MLI} & 252   &       & \multicolumn{1}{l}{RUS} & 1,147 &       & \multicolumn{1}{l}{GBR} & 1,218 \\
    CPV   & 141   &       & \multicolumn{1}{l}{GIB} & 85    &       & \multicolumn{1}{l}{MLT} & 660   &       & \multicolumn{1}{l}{RWA} & 303   &       & \multicolumn{1}{l}{TZA} & 640 \\
    CYM   & 76    &       & \multicolumn{1}{l}{KIR} & 34    &       & \multicolumn{1}{l}{MRT} & 167   &       & \multicolumn{1}{l}{SHN} & 40    &       & \multicolumn{1}{l}{USA} & 1,265 \\
    CAF   & 46    &       & \multicolumn{1}{l}{GRC} & 1,153 &       & \multicolumn{1}{l}{MUS} & 801   &       & \multicolumn{1}{l}{KNA} & 46    &       & \multicolumn{1}{l}{BFA} & 306 \\
    LKA   & 971   &       & \multicolumn{1}{l}{GRL} & 95    &       & \multicolumn{1}{l}{MEX} & 1,131 &       & \multicolumn{1}{l}{AIA} & 36    &       & \multicolumn{1}{l}{URY} & 656 \\
    TCD   & 33    &       & \multicolumn{1}{l}{GRD} & 80    &       & \multicolumn{1}{l}{TWN} & 1,106 &       & \multicolumn{1}{l}{LCA} & 88    &       & \multicolumn{1}{l}{UZB} & 371 \\
    CHL   & 1,036 &       & \multicolumn{1}{l}{GUM} & 353   &       & \multicolumn{1}{l}{MNG} & 292   &       & \multicolumn{1}{l}{SPM} & 81    &       & \multicolumn{1}{l}{VEN} & 573 \\
    CHN   & 1,267 &       & \multicolumn{1}{l}{GTM} & 933   &       & \multicolumn{1}{l}{MDA} & 650   &       & \multicolumn{1}{l}{VCT} & 69    &       & \multicolumn{1}{l}{WLF} & 10 \\
    CXR   & 43    &       & \multicolumn{1}{l}{GIN} & 216   &       & \multicolumn{1}{l}{MSR} & 17    &       & \multicolumn{1}{l}{SMR} & 184   &       & \multicolumn{1}{l}{WSM} & 176 \\
    YEM   & 323   &       &  \multicolumn{1}{l}{ZMB}     &   487    &       &       &       &       &       &       &       &       &  \\

    \hline
    
     \multicolumn{14}{l}{\shortstack[l]{\parbox{17.5cm}{\vspace{0.1cm} \scriptsize
Note: This table summarizes the number of export transaction observations  for each country. Raw Data source is from CEPII
}}
}
    \end{tabular}%

%% file: trade_costs1.tex
    \begin{tabular}{llllllllll}
    \toprule
    HS    & \multicolumn{1}{l}{Contig} & \multicolumn{1}{l}{Comlang} & \multicolumn{1}{l}{Colony} & \multicolumn{1}{l}{Comcol} & \multicolumn{1}{l}{Curcol} & \multicolumn{1}{l}{Smctry} & \multicolumn{1}{l}{LnDist} & \multicolumn{1}{l}{RTA} & \multicolumn{1}{l}{Comcur} \\
    \midrule
    2     & 0.137*** & 0.050*** & 0.073*** & 0.113*** & 0.057 & 0.080** & -0.143*** & 0.121*** & 0.075*** \\
    3     & 0.065*** & 0.044*** & 0.119*** & 0.115*** & 0.003 & 0.009 & -0.173*** & 0.063*** & 0.008 \\
    4     & 0.125*** & 0.075*** & 0.122*** & 0.127*** & 0.226*** & 0.045 & -0.173*** & 0.142*** & 0.083*** \\
    6     & 0.287*** & 0.079** & 0.223*** & -0.038 & 0.340 & -0.145 & -0.205*** & 0.198*** & 0.002 \\
    7     & 0.068*** & 0.061*** & 0.179*** & 0.124*** & 0.090 & -0.011 & -0.207*** & 0.180*** & 0.028 \\
    8     & 0.071*** & 0.071*** & 0.162*** & 0.157*** & 0.008 & -0.022 & -0.190*** & 0.121*** & -0.006 \\
    9     & 0.091*** & 0.116*** & 0.189*** & 0.055** & 0.289* & 0.094** & -0.186*** & 0.114*** & 0.024 \\
    10    & 0.204*** & -0.006 & 0.111** & 0.133*** & 0.257 & 0.034 & -0.175*** & 0.187*** & 0.201*** \\
    11    & 0.202*** & 0.067* & 0.181*** & 0.304*** & 0.154 & -0.063 & -0.217*** & 0.131*** & 0.000 \\
    15    & 0.117*** & 0.093*** & 0.165*** & 0.098*** & 0.083 & 0.044 & -0.207*** & 0.139*** & 0.097*** \\
    16    & 0.105*** & 0.093*** & 0.185*** & 0.165*** & 0.368*** & 0.027 & -0.148*** & 0.119*** & 0.009 \\
    17    & -0.010 & 0.065*** & 0.186*** & 0.142*** & 0.070 & -0.047 & -0.247*** & 0.133*** & 0.131*** \\
    18    & 0.002 & 0.076*** & 0.247*** & 0.163*** & -0.126 & -0.015 & -0.264*** & 0.136*** & 0.049 \\
    19    & 0.027 & 0.145*** & 0.176*** & 0.162*** & 0.078 & 0.027 & -0.249*** & 0.109*** & -0.016 \\
    20    & 0.056*** & 0.120*** & 0.191*** & 0.194*** & -0.044 & 0.033 & -0.181*** & 0.125*** & -0.068*** \\
    21    & 0.019 & 0.123*** & 0.164*** & 0.118*** & 0.219 & 0.043 & -0.216*** & 0.118*** & -0.009 \\
    22    & 0.077*** & 0.114*** & 0.124*** & 0.228*** & 0.152* & 0.082*** & -0.177*** & 0.160*** & -0.067*** \\
    23    & 0.076** & 0.100*** & 0.078** & 0.056 & 0.288 & -0.070 & -0.206*** & 0.107*** & -0.066** \\
    24    & 0.001 & 0.036 & 0.073** & 0.093*** & 0.039 & -0.001 & -0.196*** & 0.190*** & 0.130*** \\
    30    & -0.037** & 0.160*** & 0.097*** & 0.175*** & -0.022 & -0.007 & -0.161*** & 0.095*** & -0.011 \\
    32    & 0.235*** & 0.058* & 0.197*** & 0.201*** & 0.269 & -0.043 & -0.199*** & 0.056** & -0.014 \\
    33    & -0.006 & 0.138*** & 0.129*** & 0.139*** & 0.035 & 0.027 & -0.230*** & 0.128*** & -0.006 \\
    34    & 0.007 & 0.141*** & 0.178*** & 0.119*** & 0.037 & 0.052 & -0.283*** & 0.143***& 0.050* \\
    35    & 0.044 & 0.084*** & 0.225*** & 0.274*** & 0.363** & -0.075 & -0.248*** & 0.112*** & 0.052 \\
    36    & 0.199*** & 0.031 & -0.040 & 0.210*** & 0.235 & 0.197*** & -0.051** & 0.085** & 0.156** \\
    37    & 0.051 & 0.035 & 0.144** & -0.091 & 0.513 & 0.213** & -0.054* & 0.195*** & 0.013 \\
    38    & 0.036 & 0.054*** & 0.093*** & 0.094*** & 0.317** & -0.019 & -0.193*** & 0.093*** & 0.061** \\
    39    & 0.045** & 0.136*** & 0.137*** & 0.142*** & 0.119 & 0.068* & -0.224*** & 0.131*** & -0.036 \\
    40    & 0.085*** & 0.093*** & 0.210*** & 0.067** & 0.386** & -0.025 & -0.197*** & 0.048 & -0.023 \\
    42    & 0.037 & 0.095*** & 0.178*** & 0.017 & 0.173 & 0.065 & -0.206*** & 0.077*** & -0.064*** \\
    43    & 0.129*** & 0.114*** & 0.099** & 0.073 & 0.671*** & -0.159* & -0.161*** & -0.004 & -0.031 \\
    44    & 0.110*** & 0.143*** & 0.216*** & 0.061 & 0.315* & -0.002 & -0.203*** & 0.060*** & -0.008 \\
    46    & 0.139*** & 0.136*** & 0.128*** & 0.025 & 0.538** & 0.007 & -0.230*** & 0.036 & 0.072* \\
    48    & 0.016 & 0.146*** & 0.106*** & 0.175*** & -0.029 & -0.021 & -0.270*** & 0.131*** & -0.014 \\
    49    & -0.010 & 0.317*** & 0.215*** & 0.151*** & 0.085 & 0.071* & -0.247*** & 0.115*** & -0.025 \\
    \bottomrule
    \multicolumn{10.5}{l}{\shortstack[l]{\parbox{16.5cm}{\vspace{0.1cm} \scriptsize
Note: This table summarizes the estimated coefficients on proxy variables of trade costs, for sectors HS 2 to HS 49. Contig = 1 of the both are contiguous; Comlang =1 if both share the same language; Colony =1 if ever had colonial relation; Comcol =1 if having common colonizer; Curcol =1 if currently in colonial relation; Smctry =1 if were/are the same country; RTA =1 if having regional trade agreement; Comcur =1 if using the same currency. ***, **, * denotes significance at 1\%, 5\% and 10\% respectively.
}}
}
    
    \end{tabular}%

%% file: trade_costs2.tex
    \begin{tabular}{llllllllll}
    \toprule
    HS    & \multicolumn{1}{l}{Contig} & \multicolumn{1}{l}{Comlang} & \multicolumn{1}{l}{Colony} & \multicolumn{1}{l}{Comcol} & \multicolumn{1}{l}{Curcol} & \multicolumn{1}{l}{Smctry} & \multicolumn{1}{l}{LnDist} & \multicolumn{1}{l}{RTA} & \multicolumn{1}{l}{Comcur} \\
    \midrule
    50    & 0.020 & 0.120 & 0.074 & -0.026 & -0.053 & 0.188 & -0.163*** & -0.032 & 0.171 \\
    51    & 0.134* & 0.082 & 0.128** & 0.392** & -0.298 & -0.047 & -0.155*** & 0.125*** & 0.080 \\
    52    & 0.133*** & 0.001 & 0.082 & -0.048 & -0.184 & -0.037 & -0.214*** & 0.075** & 0.112* \\
    54    & 0.160** & 0.056 & 0.114 & 0.188 &       & -0.202 & -0.173 & -0.052*** & -0.007 \\
    55    & 0.132*** & 0.154*** & -0.032 & 0.026 & -0.498* & 0.107 & -0.181*** & 0.113** & 0.009 \\
    57    & 0.050* & 0.097*** & 0.155*** & 0.123*** & -0.053 & -0.061 & -0.216*** & 0.108*** & -0.013 \\
    58    & 0.104 & 0.157 & 0.164 & -0.248 & -0.227 & 0.089 & -0.131*** & -0.028 & -0.002 \\
    59    & 0.059 & 0.146** & 0.050 & 0.004 & -0.622* & -0.017 & -0.143*** & 0.071 & 0.087 \\
    61    & 0.046** & 0.119*** & 0.144*** & 0.059*** & 0.120 & 0.016 & -0.197*** & 0.098*** & -0.003 \\
    62    & 0.043* & 0.136*** & 0.160*** & 0.041** & 0.135 & 0.010 & -0.200*** & 0.081*** & -0.029 \\
    63    & 0.063*** & 0.135*** & 0.177*** & 0.084*** & 0.056 & 0.022 & -0.219*** & 0.128*** & -0.002 \\
    64    & 0.088*** & 0.110*** & 0.150*** & 0.102*** & 0.125 & -0.014 & -0.217*** & 0.087*** & 0.022 \\
    65    & 0.077*** & 0.111*** & 0.232*** & 0.110*** & 0.405** & -0.048 & -0.205*** & 0.096*** & 0.027 \\
    66    & 0.129*** & 0.097*** & 0.185*** & 0.078 & 0.451** & 0.001 & -0.228*** & 0.109*** & 0.063* \\
    68    & 0.132** & 0.077 & -0.052 & -0.100 & 1.140** & -0.076 & -0.178*** & 0.017 & 0.129* \\
    69    & 0.110*** & 0.122*** & 0.186*** & 0.124*** & 0.531** & -0.006 & -0.184*** & 0.107*** & -0.079*** \\
    70    & 0.090*** & 0.119*** & 0.110*** & 0.153*** & 0.346 & -0.049 & -0.193*** & 0.098*** & 0.012 \\
    71    & 0.011 & 0.140*** & 0.242*** & 0.108*** & 0.070 & 0.084 & -0.176*** & 0.109*** & 0.001 \\
    74    & 0.066** & 0.101*** & 0.208*** & 0.273*** & -0.168 & 0.047 & -0.223*** & 0.130*** & -0.084** \\
    76    & 0.149*** & 0.074*** & 0.232*** & 0.182*** & 0.242 & -0.016 & -0.235*** & 0.144*** & 0.047 \\
    82    & 0.051* & 0.122*** & 0.166*** & 0.060* & 0.104 & 0.034 & -0.196*** & 0.094*** & 0.011 \\
    83    & 0.139*** & 0.120*** & 0.269*** & -0.038 & 0.155 & -0.081 & -0.187*** & 0.079*** & -0.044 \\
    84    & 0.044** & 0.094*** & 0.111*** & 0.122*** & 0.359* & 0.055 & -0.193*** & 0.104*** & -0.017 \\
    85    & 0.019 & 0.135*** & 0.133*** & 0.135*** & 0.094 & 0.013 & -0.185*** & 0.109*** & 0.016 \\
    87    & 0.161*** & 0.070*** & 0.143*** & 0.159*** & 0.457*** & 0.005 & -0.142*** & 0.126*** & 0.026 \\
    88    & 0.051 & 0.079* & -0.147*** & -0.106 & -1.882*** & 0.124* & -0.046** & 0.057 & -0.049 \\
    89    & 0.091*** & 0.025 & 0.116*** & 0.203*** & -0.330** & 0.029 & -0.133*** & 0.067*** & -0.052* \\
    90    & 0.058** & 0.065*** & 0.137*** & 0.141*** & 0.223 & -0.009 & -0.135*** & 0.075*** & -0.044* \\
    91    & 0.036 & 0.086*** & 0.180*** & 0.177*** & 0.330*** & 0.038 & -0.170*** & 0.029 & 0.089 ***\\
    92    & 0.089*** & 0.068*** & 0.091*** & 0.206*** & 0.721*** & 0.064 & -0.137*** & 0.049** & 0.016 \\
    93    & 0.060* & 0.099*** & 0.063* & 0.342*** & -0.180 & 0.127* & -0.077*** & 0.045 & 0.031 \\
    94    & 0.034 & 0.167*** & 0.127*** & 0.078*** & 0.139 & 0.062 & -0.231*** & 0.098*** & -0.074*** \\
    95    & 0.074*** & 0.110*** & 0.193*** & 0.059** & 0.480*** & -0.025 & -0.176*** & 0.065*** & -0.013 \\
    96    & 0.064** & 0.101*** & 0.186*** & 0.092*** & 0.660*** & 0.032 & -0.206*** & 0.104*** & -0.027 \\
    97    & 0.040 & 0.110*** & 0.099*** & 0.026 & 0.358** & 0.127** & -0.102*** & 0.056*** & -0.102*** \\
    \bottomrule
    
    \multicolumn{10.5}{l}{\shortstack[l]{\parbox{16.5cm}{\vspace{0.1cm} \scriptsize
Note: This table summarizes the estimated coefficients on proxy variables of trade costs, for sectors HS 50 to HS 97. Contig = 1 of the both are contiguous; Comlang =1 if both share the same language; Colony =1 if ever had colonial relation; Comcol =1 if having common colonizer; Curcol =1 if currently in colonial relation; Smctry =1 if were/are the same country; RTA =1 if having regional trade agreement; Comcur =1 if using the same currency. ***, **, * denotes significance at 1\%, 5\% and 10\% respectively.
}}
}
    \end{tabular}%

%% file: parameters.tex
{\tabcolsep=1.3mm
    \begin{tabular}{llllll|rrrrrr}
    \toprule
    HS    & Obs   & Mean  & Std.Dev & Min   & Max   & \multicolumn{1}{l}{HS} & \multicolumn{1}{l}{Obs} & \multicolumn{1}{l}{Mean} & \multicolumn{1}{l}{Std.Dev} & \multicolumn{1}{l}{Min} & \multicolumn{1}{l}{Max} \\
    \midrule
    2     & 165   & 22.246 & 19.613 & 0.467 & 120.217 & \multicolumn{1}{l}{54} & \multicolumn{1}{l}{165} & \multicolumn{1}{l}{11.385} & \multicolumn{1}{l}{13.524} & \multicolumn{1}{l}{1.000} & \multicolumn{1}{l}{89.390} \\
    3     & 165   & 16.230 & 12.954 & 0.642 & 99.572 & \multicolumn{1}{l}{55} & \multicolumn{1}{l}{165} & \multicolumn{1}{l}{11.407} & \multicolumn{1}{l}{12.325} & \multicolumn{1}{l}{1.000} & \multicolumn{1}{l}{70.688} \\
    4     & 165   & 21.944 & 19.275 & 0.023 & 95.676 & \multicolumn{1}{l}{58} & \multicolumn{1}{l}{165} & \multicolumn{1}{l}{5.814} & \multicolumn{1}{l}{7.687} & \multicolumn{1}{l}{1.000} & \multicolumn{1}{l}{39.734} \\
    6     & 165   & 15.513 & 16.725 & 1.000 & 102.454 & \multicolumn{1}{l}{59} & \multicolumn{1}{l}{165} & \multicolumn{1}{l}{13.130} & \multicolumn{1}{l}{14.915} & \multicolumn{1}{l}{1.000} & \multicolumn{1}{l}{87.704} \\
    7     & 165   & 16.686 & 13.018 & 1.000 & 91.702 & \multicolumn{1}{l}{61} & \multicolumn{1}{l}{165} & \multicolumn{1}{l}{19.649} & \multicolumn{1}{l}{18.534} & \multicolumn{1}{l}{1.000} & \multicolumn{1}{l}{98.075} \\
    8     & 165   & 19.338 & 17.456 & 1.000 & 93.673 & \multicolumn{1}{l}{62} & \multicolumn{1}{l}{165} & \multicolumn{1}{l}{20.201} & \multicolumn{1}{l}{19.074} & \multicolumn{1}{l}{1.000} & \multicolumn{1}{l}{100.413} \\
    9     & 165   & 21.288 & 19.457 & 1.000 & 96.457 & \multicolumn{1}{l}{63} & \multicolumn{1}{l}{165} & \multicolumn{1}{l}{21.464} & \multicolumn{1}{l}{18.519} & \multicolumn{1}{l}{1.000} & \multicolumn{1}{l}{90.493} \\
    10    & 165   & 21.508 & 20.186 & 1.000 & 132.176 & \multicolumn{1}{l}{64} & \multicolumn{1}{l}{165} & \multicolumn{1}{l}{17.749} & \multicolumn{1}{l}{13.834} & \multicolumn{1}{l}{1.000} & \multicolumn{1}{l}{77.016} \\
    11    & 165   & 20.349 & 20.691 & 1.000 & 148.433 & \multicolumn{1}{l}{65} & \multicolumn{1}{l}{165} & \multicolumn{1}{l}{17.903} & \multicolumn{1}{l}{15.941} & \multicolumn{1}{l}{1.000} & \multicolumn{1}{l}{79.300} \\
    15    & 165   & 21.821 & 21.960 & 1.000 & 170.171 & \multicolumn{1}{l}{66} & \multicolumn{1}{l}{165} & \multicolumn{1}{l}{19.995} & \multicolumn{1}{l}{20.808} & \multicolumn{1}{l}{1.000} & \multicolumn{1}{l}{136.851} \\
    16    & 165   & 20.542 & 18.605 & 1.000 & 91.212 & \multicolumn{1}{l}{67} & \multicolumn{1}{l}{165} & \multicolumn{1}{l}{15.871} & \multicolumn{1}{l}{14.654} & \multicolumn{1}{l}{1.000} & \multicolumn{1}{l}{79.738} \\
    17    & 165   & 21.330 & 20.083 & 1.000 & 120.062 & \multicolumn{1}{l}{68} & \multicolumn{1}{l}{165} & \multicolumn{1}{l}{13.762} & \multicolumn{1}{l}{13.674} & \multicolumn{1}{l}{1.000} & \multicolumn{1}{l}{69.466} \\
    18    & 165   & 23.452 & 24.309 & 1.000 & 154.248 & \multicolumn{1}{l}{69} & \multicolumn{1}{l}{165} & \multicolumn{1}{l}{19.893} & \multicolumn{1}{l}{21.115} & \multicolumn{1}{l}{1.000} & \multicolumn{1}{l}{174.533} \\
    19    & 165   & 19.730 & 17.928 & 1.000 & 87.576 & \multicolumn{1}{l}{70} & \multicolumn{1}{l}{165} & \multicolumn{1}{l}{26.733} & \multicolumn{1}{l}{31.433} & \multicolumn{1}{l}{1.000} & \multicolumn{1}{l}{166.876} \\
    20    & 165   & 13.773 & 9.373 & 1.000 & 42.538 & \multicolumn{1}{l}{71} & \multicolumn{1}{l}{165} & \multicolumn{1}{l}{17.061} & \multicolumn{1}{l}{16.126} & \multicolumn{1}{l}{0.876} & \multicolumn{1}{l}{97.138} \\
    21    & 165   & 21.046 & 18.647 & 1.000 & 89.231 & \multicolumn{1}{l}{73} & \multicolumn{1}{l}{165} & \multicolumn{1}{l}{20.877} & \multicolumn{1}{l}{18.926} & \multicolumn{1}{l}{1.000} & \multicolumn{1}{l}{115.750} \\
    23    & 165   & 22.211 & 39.912 & 1.000 & 417.558 & \multicolumn{1}{l}{74} & \multicolumn{1}{l}{165} & \multicolumn{1}{l}{17.933} & \multicolumn{1}{l}{16.758} & \multicolumn{1}{l}{1.000} & \multicolumn{1}{l}{87.612} \\
    24    & 165   & 19.376 & 17.441 & 1.000 & 93.222 & \multicolumn{1}{l}{77} & \multicolumn{1}{l}{165} & \multicolumn{1}{l}{19.266} & \multicolumn{1}{l}{16.757} & \multicolumn{1}{l}{1.000} & \multicolumn{1}{l}{81.233} \\
    30    & 165   & 17.782 & 15.580 & 1.000 & 88.500 & \multicolumn{1}{l}{82} & \multicolumn{1}{l}{165} & \multicolumn{1}{l}{21.521} & \multicolumn{1}{l}{21.301} & \multicolumn{1}{l}{1.000} & \multicolumn{1}{l}{153.153} \\
    32    & 165   & 16.370 & 15.530 & 1.000 & 78.624 & \multicolumn{1}{l}{85} & \multicolumn{1}{l}{165} & \multicolumn{1}{l}{9.549} & \multicolumn{1}{l}{56.190} & \multicolumn{1}{l}{0.978} & \multicolumn{1}{l}{590.516} \\
    34    & 165   & 15.013 & 11.013 & 0.000 & 54.480 & \multicolumn{1}{l}{87} & \multicolumn{1}{l}{165} & \multicolumn{1}{l}{21.242} & \multicolumn{1}{l}{19.659} & \multicolumn{1}{l}{1.000} & \multicolumn{1}{l}{121.543} \\
    35    & 165   & 16.151 & 23.239 & 1.000 & 252.413 & \multicolumn{1}{l}{88} & \multicolumn{1}{l}{165} & \multicolumn{1}{l}{8.442} & \multicolumn{1}{l}{8.623} & \multicolumn{1}{l}{0.000} & \multicolumn{1}{l}{47.360} \\
    36    & 165   & 18.633 & 21.963 & 1.000 & 143.656 & \multicolumn{1}{l}{89} & \multicolumn{1}{l}{165} & \multicolumn{1}{l}{21.767} & \multicolumn{1}{l}{25.082} & \multicolumn{1}{l}{1.000} & \multicolumn{1}{l}{164.926} \\
    37    & 165   & 7.529 & 9.940 & 1.000 & 73.404 & \multicolumn{1}{l}{90} & \multicolumn{1}{l}{165} & \multicolumn{1}{l}{19.339} & \multicolumn{1}{l}{17.718} & \multicolumn{1}{l}{1.000} & \multicolumn{1}{l}{87.486} \\
    38    & 165   & 25.404 & 29.797 & 1.000 & 178.241 & \multicolumn{1}{l}{91} & \multicolumn{1}{l}{165} & \multicolumn{1}{l}{19.110} & \multicolumn{1}{l}{18.261} & \multicolumn{1}{l}{1.000} & \multicolumn{1}{l}{98.730} \\
    39    & 165   & 21.551 & 19.905 & 1.000 & 149.762 & \multicolumn{1}{l}{92} & \multicolumn{1}{l}{165} & \multicolumn{1}{l}{20.431} & \multicolumn{1}{l}{22.115} & \multicolumn{1}{l}{1.000} & \multicolumn{1}{l}{154.758} \\
    40    & 165   & 19.242 & 17.163 & 0.609 & 79.498 & \multicolumn{1}{l}{93} & \multicolumn{1}{l}{165} & \multicolumn{1}{l}{14.949} & \multicolumn{1}{l}{13.256} & \multicolumn{1}{l}{0.000} & \multicolumn{1}{l}{71.716} \\
    42    & 165   & 19.754 & 17.691 & 1.000 & 87.856 & \multicolumn{1}{l}{94} & \multicolumn{1}{l}{165} & \multicolumn{1}{l}{20.294} & \multicolumn{1}{l}{18.134} & \multicolumn{1}{l}{1.000} & \multicolumn{1}{l}{87.038} \\
    49    & 165   & 20.422 & 18.574 & 1.000 & 90.914 & \multicolumn{1}{l}{95} & \multicolumn{1}{l}{165} & \multicolumn{1}{l}{19.526} & \multicolumn{1}{l}{17.907} & \multicolumn{1}{l}{1.000} & \multicolumn{1}{l}{84.246} \\
    50    & 165   & 4.876 & 6.636 & 1.000 & 38.863 & \multicolumn{1}{l}{96} & \multicolumn{1}{l}{165} & \multicolumn{1}{l}{23.772} & \multicolumn{1}{l}{26.979} & \multicolumn{1}{l}{1.000} & \multicolumn{1}{l}{191.981} \\
    51    & 165   & 8.913 & 9.538 & 0.000 & 44.815 & \multicolumn{1}{l}{97} & \multicolumn{1}{l}{165} & \multicolumn{1}{l}{18.996} & \multicolumn{1}{l}{20.996} & \multicolumn{1}{l}{1.000} & \multicolumn{1}{l}{174.981} \\
    52    & 165   & 15.120 & 15.745 & 1.000 & 89.140 &       &       &       &       &       &  \\
    \bottomrule
     \multicolumn{12}{l}{\shortstack[l]{\parbox{14cm}{\vspace{0.1cm} \scriptsize
Note: This table summarizes estimated preference for quality ($\kappa_{i,s}$) across all countries $i$ within each sector $s$. 
}}
}
    \end{tabular}}%

%% file: parameters1.tex
{\tabcolsep=1.3mm

    \begin{tabular}{llllll|rrrrrl}
    \hline
    HS    & Obs   & Mean  & Std.Dev & Min   & Max   & \multicolumn{1}{l}{HS} & \multicolumn{1}{l}{Obs} & \multicolumn{1}{l}{Mean} & \multicolumn{1}{l}{Std.Dev} & \multicolumn{1}{l}{Min} & \multicolumn{1}{l}{Max} \\
    \hline
    2     & 165   & 5.098 & 3.600 & 0.000 & 18.030 & \multicolumn{1}{l}{54} & 165   & 1.948 & 1.847 & 0.000 & 13.054 \\
    3     & 165   & 3.650 & 2.453 & 0.000 & 24.217 & \multicolumn{1}{l}{55} & 165   & 1.637 & 2.557 & 0.000 & 31.159 \\
    4     & 165   & 5.912 & 6.072 & 0.000 & 67.456 & \multicolumn{1}{l}{58} & 165   & 1.385 & 1.097 & 0.000 & 9.063 \\
    6     & 165   & 3.938 & 8.218 & 0.000 & 102.771 & \multicolumn{1}{l}{59} & 165   & 2.252 & 2.422 & 0.000 & 12.221 \\
    7     & 165   & 3.871 & 1.716 & 0.000 & 8.806 & \multicolumn{1}{l}{61} & 165   & 2.742 & 1.486 & 0.000 & 10.675 \\
    8     & 165   & 4.242 & 2.156 & 0.000 & 16.841 & \multicolumn{1}{l}{62} & 165   & 3.051 & 1.714 & 0.000 & 14.742 \\
    9     & 165   & 4.271 & 3.949 & 0.000 & 47.431 & \multicolumn{1}{l}{63} & 165   & 5.839 & 3.079 & 0.000 & 21.517 \\
    10    & 165   & 4.413 & 3.849 & 0.000 & 31.479 & \multicolumn{1}{l}{64} & 165   & 4.525 & 3.065 & 0.000 & 30.028 \\
    11    & 165   & 3.126 & 3.263 & 0.000 & 27.200 & \multicolumn{1}{l}{65} & 165   & 2.413 & 1.688 & 0.000 & 9.500 \\
    15    & 165   & 3.524 & 2.351 & 0.000 & 13.085 & \multicolumn{1}{l}{66} & 165   & 2.196 & 2.542 & 0.000 & 24.453 \\
    16    & 165   & 3.553 & 2.603 & 0.000 & 21.460 & \multicolumn{1}{l}{67} & 165   & 1.988 & 1.576 & 0.000 & 10.410 \\
    17    & 165   & 3.284 & 2.201 & 0.000 & 16.038 & \multicolumn{1}{l}{68} & 165   & 1.816 & 1.683 & 0.000 & 9.698 \\
    18    & 165   & 3.512 & 2.380 & 0.000 & 12.900 & \multicolumn{1}{l}{69} & 165   & 2.756 & 1.681 & 0.000 & 10.819 \\
    19    & 165   & 2.469 & 1.931 & 0.000 & 15.552 & \multicolumn{1}{l}{70} & 165   & 3.908 & 2.803 & 0.000 & 14.830 \\
    20    & 165   & 3.956 & 1.884 & 0.000 & 12.288 & \multicolumn{1}{l}{71} & 165   & 3.625 & 2.377 & 0.000 & 17.847 \\
    21    & 165   & 4.053 & 2.107 & 0.000 & 17.662 & \multicolumn{1}{l}{73} & 165   & 3.714 & 2.075 & 0.000 & 11.556 \\
    23    & 165   & 1.843 & 3.430 & 0.000 & 38.749 & \multicolumn{1}{l}{74} & 165   & 2.301 & 1.510 & 0.000 & 7.564 \\
    24    & 165   & 4.121 & 2.994 & 0.000 & 18.264 & \multicolumn{1}{l}{77} & 165   & 2.951 & 2.109 & 0.000 & 16.693 \\
    30    & 165   & 3.648 & 1.906 & 0.000 & 9.800 & \multicolumn{1}{l}{82} & 165   & 3.261 & 1.908 & 0.000 & 9.653 \\
    32    & 165   & 1.966 & 1.857 & 0.000 & 11.108 & \multicolumn{1}{l}{85} & 165   & 5.292 & 45.482 & 0.247 & 585.700 \\
    34    & 165   & 4.615 & 2.607 & 0.000 & 16.672 & \multicolumn{1}{l}{87} & 165   & 3.965 & 2.265 & 0.000 & 10.860 \\
    35    & 165   & 3.204 & 3.066 & 0.000 & 21.819 & \multicolumn{1}{l}{88} & 165   & 1.664 & 1.817 & 0.000 & 20.027 \\
    36    & 165   & 1.765 & 1.771 & 0.000 & 10.764 & \multicolumn{1}{l}{89} & 165   & 3.404 & 2.598 & 0.000 & 18.004 \\
    37    & 165   & 1.275 & 0.879 & 0.000 & 6.082 & \multicolumn{1}{l}{90} & 165   & 3.089 & 2.180 & 0.000 & 16.907 \\
    38    & 165   & 3.278 & 4.996 & 0.000 & 52.698 & \multicolumn{1}{l}{91} & 165   & 2.967 & 2.330 & 0.000 & 14.631 \\
    39    & 165   & 5.994 & 2.513 & 0.000 & 12.972 & \multicolumn{1}{l}{92} & 165   & 2.533 & 1.790 & 0.000 & 10.225 \\
    40    & 165   & 2.759 & 2.512 & 0.000 & 24.370 & \multicolumn{1}{l}{93} & 165   & 2.431 & 2.326 & 0.000 & 12.901 \\
    42    & 165   & 3.930 & 1.835 & 0.000 & 13.093 & \multicolumn{1}{l}{94} & 165   & 3.932 & 1.809 & 0.000 & 11.029 \\
    49    & 165   & 3.972 & 2.692 & 0.000 & 24.601 & \multicolumn{1}{l}{95} & 165   & 3.248 & 2.012 & 0.000 & 20.826 \\
    50    & 165   & 1.134 & 0.630 & 0.000 & 4.231 & \multicolumn{1}{l}{96} & 165   & 3.588 & 2.238 & 0.000 & 17.202 \\
    51    & 165   & 2.859 & 2.968 & 0.000 & 18.538 & \multicolumn{1}{l}{97} & 165   & 3.853 & 1.848 & 0.000 & 11.709 \\
    52    & 165   & 1.764 & 1.545 & 0.000 & 9.746 &       &       &       &       &       &  \\
    \hline
    \multicolumn{12}{l}{\shortstack[l]{\parbox{14cm}{\vspace{0.1cm} \scriptsize
Note: This table summarizes estimated cost for quality ($\mu_{i,s}$) across all countries $i$ within each sector $s$. 
}}
}
    \end{tabular}}%

%% file: kappa_gdp.tex

    \begin{tabular}{ll|ll|llrr}
    \toprule
    HS    & Coefficients & HS    & \multicolumn{1}{l}{Coefficients} & HS    & Coefficients & \multicolumn{1}{l}{HS} & \multicolumn{1}{l}{Coefficients} \\
    \midrule
    2     & 1.879$^{*}$ & 23    & 1.438 & 54    & \multicolumn{1}{l|}{3.506$^{***}$} & \multicolumn{1}{l}{74} & \multicolumn{1}{l}{2.561$^{***}$} \\
    3     & 1.238$^{*}$ & 24    & 1.613$^{*}$ & 55    & \multicolumn{1}{l|}{2.989$^{***}$} & \multicolumn{1}{l}{76} & \multicolumn{1}{l}{2.044$^{**}$} \\
    4     & 2.067$^{**}$ & 30    & 0.867 & 58    & \multicolumn{1}{l|}{2.163$^{***}$} & \multicolumn{1}{l}{82} & \multicolumn{1}{l}{1.692} \\
    6     & 4.128$^{***}$ & 32    & 2.431$^{***}$ & 59    & \multicolumn{1}{l|}{3.903$^{***}$} & \multicolumn{1}{l}{85} & \multicolumn{1}{l}{1.231} \\
    7     & 0.062 & 34    & 0.040 & 61    & \multicolumn{1}{l|}{2.009$^{**}$} & \multicolumn{1}{l}{87} & \multicolumn{1}{l}{1.319} \\
    8     & 1.518$^{*}$ & 35    & -0.207 & 62    & \multicolumn{1}{l|}{1.960$^{**}$} & \multicolumn{1}{l}{88} & \multicolumn{1}{l}{1.852$^{***}$} \\
    9     & 0.985 & 36    & 2.228$^{*}$ & 63    & \multicolumn{1}{l|}{1.314} & \multicolumn{1}{l}{89} & \multicolumn{1}{l}{3.152$^{**}$} \\
    10    & 1.451 & 37    & 3.016$^{***}$ & 64    & \multicolumn{1}{l|}{0.240} & \multicolumn{1}{l}{90} & \multicolumn{1}{l}{2.128$^{**}$} \\
    11    & 3.033$^{***}$ & 38    & 2.144 & 65    & \multicolumn{1}{l|}{1.917$^{**}$} & \multicolumn{1}{l}{91} & \multicolumn{1}{l}{2.698$^{***}$} \\
    15    & 2.623$^{**}$ & 39    & 0.328 & 66    & \multicolumn{1}{l|}{2.082$^{*}$} & \multicolumn{1}{l}{92} & \multicolumn{1}{l}{2.889$^{**}$} \\
    16    & 2.029$^{**}$ & 40    & 1.568$^{*}$ & 67    & \multicolumn{1}{l|}{1.955$^{**}$} & \multicolumn{1}{l}{93} & \multicolumn{1}{l}{2.025$^{***}$} \\
    17    & 1.656 & 42    & 2.062$^{**}$ & 68    & \multicolumn{1}{l|}{2.993$^{***}$} & \multicolumn{1}{l}{94} & \multicolumn{1}{l}{1.649$^{*}$} \\
    18    & 2.106 & 49    & 1.891$^{*}$ & 69    & \multicolumn{1}{l|}{2.422$^{**}$} & \multicolumn{1}{l}{95} & \multicolumn{1}{l}{1.800$^{*}$} \\
    19    & 1.733$^{*}$ & 50    & 2.032$^{***}$ & 70    & \multicolumn{1}{l|}{2.857$^{*}$} & \multicolumn{1}{l}{96} & \multicolumn{1}{l}{2.146} \\
    20    & 0.186 & 51    & 2.090$^{***}$ & 71    & 2.679$^{***}$ & \multicolumn{1}{l}{97} & \multicolumn{1}{l}{3.422$^{***}$} \\
    21    & 1.692$^{*}$ & 52    & 2.829$^{***}$ & 73    & 1.580 &       &  \\
    \bottomrule
    
    \multicolumn{8}{l}{\shortstack[l]{\parbox{12cm}{\vspace{0.1cm} \scriptsize
Note: This table the estimated coefficients on (log) GDP per capita with dependent variable being the $\kappa_{i}$ for each sector $s$. ${}^{*}$, ${}^{**}$, and ${}^{***}$ indicate statistical significance at the 10\%, 5\%, and 1\% level, respectively
}}
}
    \end{tabular}%

%% file: mu_gdp.tex

    \begin{tabular}{ll|ll|llrr}
    \toprule
    HS    & \multicolumn{1}{l}{Coefficients} & HS    & Coefficients & HS    & Coefficients & \multicolumn{1}{l}{HS} & \multicolumn{1}{l}{Coefficients} \\
    \midrule
    2     & 0.494$^{***}$ & 23    & 0.264 & 54    & \multicolumn{1}{l|}{0.395$^{***}$} & \multicolumn{1}{l}{74} & \multicolumn{1}{l}{0.246$^{***}$} \\
    3     & 0.219$^{*}$ & 24    & 0.210 & 55    & \multicolumn{1}{l|}{0.278$^{**}$} & \multicolumn{1}{l}{77} & \multicolumn{1}{l}{0.115} \\
    4     & 0.521 & 30    & -0.028 & 58    & \multicolumn{1}{l|}{0.186$^{***}$} & \multicolumn{1}{l}{82} & \multicolumn{1}{l}{0.209$^{**}$} \\
    6     & 0.338 & 32    & 0.169$^{*}$ & 59    & \multicolumn{1}{l|}{0.715$^{***}$} & \multicolumn{1}{l}{85} & \multicolumn{1}{l}{1.208} \\
    7     & 0.044 & 34    & 0.151 & 61    & \multicolumn{1}{l|}{-0.205$^{***}$} & \multicolumn{1}{l}{87} & \multicolumn{1}{l}{-0.247$^{**}$} \\
    8     & -0.146 & 35    & 0.404$^{***}$ & 62    & \multicolumn{1}{l|}{-0.076} & \multicolumn{1}{l}{88} & \multicolumn{1}{l}{0.425$^{***}$} \\
    9     & -0.055 & 36    & -0.076 & 63    & \multicolumn{1}{l|}{0.279$^{*}$} & \multicolumn{1}{l}{89} & \multicolumn{1}{l}{0.166} \\
    10    & 0.192 & 37    & 0.163$^{***}$ & 64    & \multicolumn{1}{l|}{-0.199} & \multicolumn{1}{l}{90} & \multicolumn{1}{l}{0.060} \\
    11    & 0.337$^{*}$ & 38    & -0.184 & 65    & \multicolumn{1}{l|}{-0.063} & \multicolumn{1}{l}{91} & \multicolumn{1}{l}{-0.139} \\
    15    & 0.312$^{**}$ & 39    & 0.453$^{***}$ & 66    & \multicolumn{1}{l|}{0.157} & \multicolumn{1}{l}{92} & \multicolumn{1}{l}{0.120} \\
    16    & 0.185 & 40    & 0.187 & 67    & \multicolumn{1}{l|}{0.070} & \multicolumn{1}{l}{93} & \multicolumn{1}{l}{0.584$^{***}$} \\
    17    & 0.053 & 42    & 0.121 & 68    & \multicolumn{1}{l|}{0.427$^{***}$} & \multicolumn{1}{l}{94} & \multicolumn{1}{l}{0.046} \\
    18    & 0.328$^{***}$ & 49    & 0.147 & 69    & \multicolumn{1}{l|}{-0.150} & \multicolumn{1}{l}{95} & \multicolumn{1}{l}{0.077} \\
    19    & -0.176$^{*}$ & 50    & 0.053 & 70    & \multicolumn{1}{l|}{0.156} & \multicolumn{1}{l}{96} & \multicolumn{1}{l}{0.208$^{*}$} \\
    20    & 0.110 & 51    & 0.767$^{***}$ & 71    & \multicolumn{1}{l|}{-0.105} & \multicolumn{1}{l}{97} & \multicolumn{1}{l}{-0.128} \\
    21    & -0.113 & 52    & 0.191$^{**}$ & 73    & 0.155 &       &  \\
    \bottomrule
    \multicolumn{8}{l}{\shortstack[l]{\parbox{12cm}{\vspace{0.1cm} \scriptsize
Note: This table the estimated coefficients on (log) GDP per capita with dependent variable being the $\mu_{i}$ for each sector $s$. ${}^{*}$, ${}^{**}$, and ${}^{***}$ indicate statistical significance at the 10\%, 5\%, and 1\% level, respectively
}}
}
    \end{tabular}%

%% file: summary_cm1.tex
    \begin{tabular}{llllll}
    \toprule
    \multicolumn{1}{l}{HS} & \multicolumn{1}{l}{Obs} & \multicolumn{1}{l}{Mean} & \multicolumn{1}{l}{Std.Dev} & \multicolumn{1}{l}{Min} & \multicolumn{1}{l}{Max} \\
    \midrule
    2     & 150   & 24.900 & 30.751 & 10.451 & 308.543 \\
    3     & 138   & 1.609 & 4.517 & 0.674 & 49.741 \\
    4     & 148   & 0.577 & 1.377 & 0.221 & 16.706 \\
    6     & 129   & 5.67E+03 & 5.83E+03 & 1.55E+03 & 3.96E+04 \\
    7     & 155   & 1.989 & 1.036 & 1.290 & 8.073 \\
    8     & 148   & 136.358 & 64.449 & 71.293 & 398.349 \\
    9     & 150   & 0.529 & 0.369 & 0.336 & 3.782 \\
    10    & 152   & 1.442 & 1.383 & 0.395 & 4.362 \\
    11    & 157   & 0.197 & 0.057 & 0.121 & 0.581 \\
    15    & 154   & 0.824 & 0.822 & 0.251 & 3.487 \\
    16    & 151   & 0.685 & 0.463 & 0.367 & 3.840 \\
    17    & 156   & 0.361 & 0.272 & 0.197 & 2.481 \\
    18    & 158   & 0.050 & 0.021 & 0.028 & 0.117 \\
    19    & 163   & 0.060 & 0.025 & 0.037 & 0.226 \\
    20    & 142   & 0.429 & 0.242 & 0.297 & 2.184 \\
    21    & 155   & 0.057 & 0.034 & 0.038 & 0.337 \\
    23    & 160   & 0.011 & 0.002 & 0.007 & 0.023 \\
    24    & 147   & 0.004 & 0.001 & 0.002 & 0.008 \\
    30    & 158   & 0.325 & 0.625 & 0.168 & 7.216 \\
    32    & 159   & 9.104 & 2.486 & 7.390 & 33.174 \\
    34    & 145   & 6251.746 & 3636.260 & 4089.427 & 40193.020 \\
    35    & 152   & 0.060 & 0.024 & 0.029 & 0.183 \\
    36    & 153   & 436.416 & 30.690 & 404.460 & 637.010 \\
    37    & 143   & 3.28E+07 & 2.19E+06 & 3.10E+07 & 4.92E+07 \\
    38    & 159   & 0.060 & 0.063 & 0.029 & 0.437 \\
    39    & 147   & 0.015 & 0.061 & 0.004 & 0.662 \\
    40    & 158   & 0.361 & 0.475 & 0.161 & 5.763 \\
    42    & 158   & 12.000 & 64.131 & 4.345 & 792.150 \\
    49    & 164   & 0.117 & 0.159 & 0.048 & 0.682 \\
    50    & 107   & 4.01E+12 & 7.98E+11 & 3.32E+12 & 1.13E+13 \\
    51    & 95    & 0.442 & 0.305 & 0.151 & 1.777 \\
    \bottomrule
    
\multicolumn{6}{l}{\shortstack[l]{\parbox{10.5cm}{\vspace{0.1cm} \scriptsize
Note: This table summarizes the estimated exogenous competitiveness from quality model for industries from HS 2 to HS 51
}}
}
    \end{tabular}%

%% file: summary_cm2.tex

    \begin{tabular}{llllll}
    \toprule
    \multicolumn{1}{l}{HS} & \multicolumn{1}{l}{Obs} & \multicolumn{1}{l}{Mean} & \multicolumn{1}{l}{Std.Dev} & \multicolumn{1}{l}{Min} & \multicolumn{1}{l}{Max} \\
    \midrule
    52    & 154   & 0.354 & 0.245 & 0.229 & 3.183 \\
    54    & 122   & 9051.833 & 1944.386 & 7409.059 & 27034.280 \\
    55    & 149   & 0.474 & 0.070 & 0.361 & 0.904 \\
    58    & 149   & 1.407 & 0.154 & 1.155 & 2.506 \\
    59    & 116   & 0.781 & 0.261 & 0.345 & 1.080 \\
    61    & 156   & 3.372 & 1.820 & 2.151 & 19.650 \\
    62    & 155   & 1.013 & 0.408 & 0.679 & 3.317 \\
    63    & 157   & 1.537 & 1.890 & 0.846 & 21.568 \\
    64    & 148   & 0.550 & 0.372 & 0.323 & 2.306 \\
    65    & 153   & 4.423 & 5.596 & 2.232 & 38.725 \\
    66    & 157   & 80.861 & 25.629 & 50.781 & 334.302 \\
    67    & 148   & 6.701 & 9.744 & 2.679 & 75.814 \\
    68    & 143   & 10.399 & 3.691 & 6.252 & 41.790 \\
    69    & 161   & 107.424 & 1266.542 & 5.332 & 16078.190 \\
    70    & 158   & 11.084 & 10.113 & 4.999 & 117.903 \\
    71    & 146   & 1.06E+05 & 3.95E+05 & 2.86E+03 & 1.85E+06 \\
    73    & 159   & 0.384 & 0.244 & 0.240 & 2.430 \\
    74    & 154   & 59.894 & 20.527 & 43.466 & 234.731 \\
    76    & 158   & 0.077 & 0.029 & 0.042 & 0.223 \\
    82    & 159   & 24.072 & 15.373 & 14.861 & 171.604 \\
    85    & 162   & 5242.618 & 2099.089 & 4312.667 & 27509.740 \\
    87    & 158   & 2.802 & 0.958 & 2.006 & 6.494 \\
    88    & 133   & 524.241 & 19.825 & 501.511 & 644.742 \\
    89    & 150   & 14.529 & 10.022 & 7.053 & 87.736 \\
    90    & 161   & 3.463 & 25.910 & 0.802 & 329.955 \\
    91    & 159   & 9.40E+04 & 2.33E+04 & 7.56E+04 & 2.71E+05 \\
    92    & 150   & 5.985 & 2.880 & 3.919 & 21.810 \\
    93    & 135   & 2769.903 & 6592.229 & 1639.503 & 78574.520 \\
    94    & 164   & 0.218 & 0.081 & 0.153 & 0.619 \\
    95    & 154   & 3.116 & 3.515 & 1.331 & 37.424 \\
    96    & 157   & 0.197 & 0.083 & 0.127 & 0.558 \\
    \bottomrule
    
    \multicolumn{6}{l}{\shortstack[l]{\parbox{10.5cm}{\vspace{0.1cm} \scriptsize
Note: Note: This table summarizes the estimated exogenous competitiveness from quality model for industries HS 52 to HS 96
}}
}
    \end{tabular}%

%% file: cutoff_quality.tex

    \begin{tabular}{llllllrllllll}
    \toprule
    HS & Obs   & Mean  & Std.Dev & Min   & Max   &       & HS & Obs   & Mean  & Std.Dev & Min   & Max \\
\cmidrule{1-6}\cmidrule{8-13}    2     & 150   & 3.193 & 2.258 & -3.773 & 8.843 &       &       &       &       &       &       &  \\
    3     & 138   & 1.097 & 2.869 & -10.654 & 7.451 &       & 52    & 154   & 0.527 & 2.504 & -6.717 & 10.367 \\
    4     & 148   & 1.008 & 1.617 & -7.744 & 4.466 &       & 54    & 104   & 2.720 & 2.819 & -4.052 & 17.634 \\
    6     & 128   & 7.627 & 2.840 & -1.662 & 14.552 &       & 55    & 148   & 0.321 & 3.007 & -7.431 & 7.327 \\
    7     & 154   & 2.160 & 2.769 & -8.700 & 10.649 &       & 58    & 144   & -4.662 & 3.801 & -12.224 & 14.796 \\
    8     & 147   & 3.079 & 2.453 & -9.172 & 13.662 &       & 59    & 116   & 1.582 & 1.895 & -6.410 & 4.840 \\
    9     & 150   & 1.271 & 1.890 & -6.190 & 4.611 &       & 61    & 154   & 0.670 & 3.494 & -7.308 & 9.192 \\
    10    & 152   & 0.793 & 2.024 & -5.359 & 5.373 &       & 62    & 155   & -0.568 & 3.570 & -8.134 & 7.563 \\
    11    & 157   & 0.293 & 2.377 & -11.553 & 12.827 &       & 63    & 157   & 1.753 & 3.012 & -8.383 & 10.598 \\
    15    & 154   & 1.069 & 1.622 & -6.561 & 4.172 &       & 64    & 148   & 0.235 & 3.075 & -7.296 & 7.520 \\
    16    & 151   & 0.788 & 1.743 & -6.752 & 4.771 &       & 65    & 153   & 1.883 & 2.421 & -3.978 & 9.374 \\
    17    & 156   & 1.478 & 1.650 & -5.644 & 5.916 &       & 66    & 157   & 1.921 & 2.542 & -3.859 & 14.613 \\
    18    & 158   & 0.264 & 1.434 & -3.614 & 5.453 &       & 67    & 141   & 0.923 & 3.630 & -9.452 & 7.234 \\
    19    & 163   & -0.456 & 1.475 & -5.850 & 6.520 &       & 68    & 141   & 3.829 & 2.277 & -5.093 & 7.577 \\
    20    & 142   & -0.012 & 2.088 & -13.366 & 11.607 &       & 69    & 161   & 2.516 & 2.881 & -5.794 & 13.437 \\
    21    & 155   & -0.772 & 1.703 & -7.983 & 10.165 &       & 70    & 158   & 2.477 & 2.091 & -3.990 & 7.993 \\
    23    & 160   & -2.093 & 1.907 & -7.440 & 2.399 &       & 71    & 122   & 6.911 & 3.253 & -0.570 & 15.437 \\
    24    & 146   & -5.573 & 3.395 & -19.292 & 3.481 &       & 73    & 159   & 1.360 & 1.467 & -3.445 & 4.580 \\
    30    & 156   & -0.648 & 2.535 & -6.231 & 9.380 &       & 74    & 153   & 0.526 & 2.202 & -5.759 & 18.850 \\
    32    & 158   & -0.210 & 2.195 & -7.951 & 12.375 &       & 76    & 158   & 0.590 & 1.368 & -4.148 & 5.923 \\
    34    & 138   & 0.856 & 2.276 & -6.754 & 17.849 &       & 82    & 159   & 2.959 & 2.190 & -2.797 & 8.213 \\
    35    & 152   & -0.364 & 2.131 & -12.268 & 5.465 &       & 87    & 156   & 1.713 & 2.887 & -7.842 & 13.780 \\
    36    & 141   & 2.766 & 3.501 & -4.758 & 14.432 &       & 88    & 123   & 0.098 & 3.029 & -6.952 & 15.822 \\
    37    & 84    & 7.881 & 2.917 & 1.732 & 25.094 &       & 89    & 149   & 2.490 & 2.154 & -4.532 & 8.526 \\
    38    & 159   & -0.886 & 1.577 & -4.756 & 6.580 &       & 90    & 156   & 0.835 & 1.440 & -2.976 & 4.344 \\
    39    & 147   & -1.963 & 2.057 & -9.644 & 3.091 &       & 91    & 98    & 5.958 & 3.803 & -2.000 & 25.677 \\
    40    & 158   & 0.287 & 1.675 & -5.473 & 9.903 &       & 92    & 146   & 1.197 & 3.382 & -7.977 & 8.698 \\
    42    & 154   & 0.914 & 2.770 & -8.523 & 8.218 &       & 93    & 105   & 6.577 & 2.748 & -3.537 & 18.500 \\
    49    & 163   & 0.449 & 1.914 & -4.101 & 9.980 &       & 94    & 164   & 0.662 & 2.157 & -10.345 & 8.008 \\
    50    & 106   & 22.518 & 3.090 & 15.606 & 38.232 &       & 95    & 153   & 0.347 & 2.618 & -7.854 & 11.375 \\
    51    & 94    & 0.590 & 3.265 & -14.942 & 5.291 &       & 96    & 157   & -0.042 & 2.060 & -5.586 & 4.974 \\
    \bottomrule
    \multicolumn{13}{l}{\shortstack[l]{\parbox{17cm}{\vspace{0.1cm} \scriptsize
Note: This table summarizes the estimated exogenous competitiveness $c_{M}^{i}$ across countries within each sector $s$ from quality model
}}
}
    \end{tabular}%

%% file: cutoff_noqua.tex

    \begin{tabular}{lllllllllllll}
    \toprule
    HS & Obs   & Mean  & Std.Dev & Min   & Max   &       & HS & Obs   & Mean  & Std.Dev & Min   & Max \\
    \midrule
    2     & 150   & 3.011 & 2.214 & -3.992 & 8.409 &       &       &       &       &       &       &  \\
    3     & 138   & 1.249 & 2.645 & -8.513 & 6.687 &       & 52    & 154   & 0.495 & 2.501 & -6.882 & 7.866 \\
    4     & 148   & 0.868 & 1.228 & -2.975 & 4.239 &       & 54    & 104   & 1.041 & 2.077 & -4.777 & 5.708 \\
    6     & 128   & 7.710 & 2.764 & -0.780 & 15.091 &       & 55    & 148   & 0.065 & 2.562 & -6.670 & 4.982 \\
    7     & 154   & 1.216 & 1.979 & -4.815 & 7.150 &       & 58    & 144   & 1.049 & 2.572 & -5.458 & 13.692 \\
    8     & 147   & 2.261 & 2.126 & -2.952 & 12.880 &       & 59    & 116   & 1.700 & 1.351 & -2.378 & 4.444 \\
    9     & 150   & 0.654 & 1.684 & -4.379 & 3.718 &       & 61    & 154   & 0.792 & 3.299 & -6.633 & 10.001 \\
    10    & 152   & 0.997 & 1.848 & -5.317 & 5.391 &       & 62    & 155   & -0.611 & 3.498 & -8.208 & 8.613 \\
    11    & 157   & 0.436 & 1.352 & -4.034 & 3.497 &       & 63    & 157   & 1.879 & 2.569 & -4.476 & 7.188 \\
    15    & 154   & 1.047 & 1.439 & -2.104 & 4.277 &       & 64    & 148   & 0.124 & 2.744 & -6.358 & 8.317 \\
    16    & 151   & 0.849 & 1.620 & -6.824 & 4.562 &       & 65    & 153   & 2.020 & 2.333 & -3.684 & 9.270 \\
    17    & 156   & 1.341 & 1.236 & -2.432 & 5.477 &       & 66    & 157   & 1.943 & 2.580 & -4.759 & 14.359 \\
    18    & 158   & 0.404 & 1.234 & -2.701 & 4.046 &       & 67    & 141   & 1.217 & 3.488 & -6.993 & 7.586 \\
    19    & 163   & -0.273 & 1.013 & -2.728 & 2.819 &       & 68    & 141   & 3.971 & 1.908 & -1.216 & 7.532 \\
    20    & 142   & 0.173 & 1.198 & -2.071 & 2.534 &       & 69    & 161   & 2.326 & 2.653 & -5.742 & 8.726 \\
    21    & 155   & -0.625 & 1.104 & -2.981 & 1.894 &       & 70    & 158   & 2.570 & 2.003 & -3.190 & 8.366 \\
    23    & 160   & -2.125 & 1.851 & -7.661 & 1.200 &       & 71    & 122   & 6.987 & 3.883 & -5.886 & 15.516 \\
    24    & 146   & -5.581 & 3.025 & -15.509 & 2.437 &       & 73    & 159   & 1.500 & 1.253 & -2.480 & 4.756 \\
    30    & 156   & -0.443 & 2.427 & -5.975 & 4.664 &       & 74    & 153   & 0.488 & 1.452 & -4.534 & 6.413 \\
    32    & 158   & -0.225 & 1.452 & -4.578 & 3.050 &       & 76    & 158   & 0.557 & 1.131 & -4.030 & 4.030 \\
    34    & 138   & -0.095 & 1.354 & -5.192 & 3.013 &       & 82    & 159   & 1.858 & 2.116 & -3.206 & 5.999 \\
    35    & 152   & -0.751 & 1.718 & -6.521 & 2.556 &       & 87    & 156   & 1.902 & 2.518 & -5.229 & 7.500 \\
    36    & 141   & 2.621 & 3.361 & -7.049 & 10.561 &       & 88    & 123   & -0.144 & 1.899 & -5.253 & 5.816 \\
    37    & 84    & 2.459 & 2.584 & -4.488 & 7.855 &       & 89    & 149   & 2.773 & 2.055 & -2.797 & 10.367 \\
    38    & 159   & -0.736 & 1.397 & -3.502 & 6.090 &       & 90    & 156   & 0.969 & 1.361 & -2.965 & 4.320 \\
    39    & 147   & -1.879 & 1.561 & -7.064 & 3.002 &       & 91    & 98    & 6.833 & 3.035 & 2.149 & 21.273 \\
    40    & 158   & 0.401 & 1.295 & -5.337 & 5.052 &       & 92    & 146   & 1.443 & 3.264 & -6.880 & 8.791 \\
    42    & 154   & 0.815 & 2.659 & -7.365 & 9.969 &       & 93    & 105   & 7.317 & 2.357 & -2.368 & 17.926 \\
    49    & 163   & 0.303 & 1.433 & -3.599 & 5.719 &       & 94    & 164   & 1.006 & 1.677 & -4.750 & 5.649 \\
    50    & 106   & 1.955 & 2.781 & -6.967 & 7.527 &       & 95    & 153   & 0.268 & 2.484 & -7.657 & 5.498 \\
    51    & 94    & 1.768 & 2.238 & -5.315 & 7.984 &       & 96    & 157   & 0.107 & 2.070 & -4.670 & 5.063 \\
    \bottomrule
     \multicolumn{13}{l}{\shortstack[l]{\parbox{17cm}{\vspace{0.1cm} \scriptsize
Note: This table summarizes the estimated exogenous competitiveness $c_{M}^{i}$ across countries within each sector $s$ from Melitz and Ottaviano (2008) model
}}
}
    \end{tabular}%

%% file: cf1.tex

    \begin{tabular}{lllllrlllll}
    \toprule
    HS    & Mean  & Std.Dev & Min   & Max   &       & HS    & Mean  & Std.Dev & Min   & Max \\
    \midrule
    2     & 2.801 & 2.284 & -0.003 & 9.085 &       &       &       &       &       &  \\
    3     & 1.983 & 3.815 & -0.327 & 30.816 &       & 52    & 1.935 & 1.768 & -0.085 & 6.716 \\
    4     & 0.944 & 1.178 & -0.302 & 4.733 &       & 54    & 13.115 & 6.858 & 0.000 & 39.165 \\
    6     & 3.914 & 4.133 & -0.038 & 24.256 &       & 55    & 1.217 & 1.876 & -0.353 & 15.144 \\
    7     & 1.248 & 1.236 & -0.383 & 3.983 &       & 58    & 5.617 & 3.677 & -0.053 & 15.358 \\
    8     & 3.591 & 3.468 & -0.348 & 15.524 &       & 59    & 0.752 & 1.322 & -0.602 & 8.314 \\
    9     & 1.411 & 1.341 & -0.200 & 7.070 &       & 61    & 3.032 & 4.738 & -0.220 & 18.739 \\
    10    & 3.196 & 3.037 & -0.001 & 24.547 &       & 62    & 2.988 & 6.235 & -0.179 & 45.542 \\
    11    & 0.731 & 0.837 & -0.244 & 5.103 &       & 63    & 1.520 & 2.268 & -0.452 & 12.386 \\
    15    & 1.002 & 1.428 & -0.433 & 8.380 &       & 64    & 1.173 & 2.023 & -0.294 & 9.771 \\
    16    & 1.277 & 1.521 & -0.255 & 6.770 &       & 65    & 1.856 & 1.703 & -0.095 & 10.980 \\
    17    & 1.973 & 1.369 & 0.212 & 6.722 &       & 66    & 3.150 & 2.721 & -0.032 & 14.317 \\
    18    & 1.362 & 1.105 & -0.296 & 5.531 &       & 67    & 3.587 & 5.030 & -0.466 & 23.328 \\
    19    & 0.986 & 1.816 & -0.398 & 12.198 &       & 68    & 0.839 & 1.219 & -0.439 & 6.521 \\
    20    & 1.298 & 1.960 & -0.505 & 13.862 &       & 69    & 2.108 & 2.579 & -0.254 & 13.567 \\
    21    & 0.583 & 1.015 & -0.247 & 8.883 &       & 70    & 2.675 & 2.609 & -0.302 & 12.418 \\
    23    & 0.900 & 1.212 & -0.261 & 5.901 &       & 71    & 5.430 & 3.946 & 0.033 & 15.622 \\
    24    & 2.055 & 2.091 & -0.006 & 14.665 &       & 73    & 2.160 & 1.785 & -0.249 & 7.655 \\
    30    & 2.829 & 2.432 & -0.263 & 12.214 &       & 74    & 3.756 & 3.876 & -0.023 & 22.587 \\
    32    & 4.943 & 4.259 & -0.022 & 30.004 &       & 76    & 0.148 & 0.388 & -0.505 & 1.933 \\
    34    & 28.687 & 13.970 & -0.074 & 71.330 &       & 82    & 2.988 & 3.571 & -0.027 & 17.884 \\
    35    & 0.276 & 0.407 & -0.184 & 2.016 &       & 87    & 1.815 & 1.789 & -0.712 & 7.665 \\
    36    & 12.454 & 9.820 & 0.216 & 45.854 &       & 88    & 8.191 & 4.952 & -0.054 & 22.213 \\
    37    & 14.389 & 5.979 & -0.053 & 36.342 &       & 89    & 1.434 & 1.475 & -0.113 & 8.867 \\
    38    & 0.662 & 0.893 & -0.174 & 3.920 &       & 90    & 3.195 & 2.180 & 0.030 & 10.248 \\
    39    & 0.596 & 0.981 & -0.526 & 4.460 &       & 91    & 8.818 & 6.094 & -0.065 & 24.584 \\
    40    & 1.112 & 1.552 & -0.368 & 8.873 &       & 92    & 2.605 & 3.745 & -0.320 & 22.862 \\
    42    & 8.844 & 6.304 & 0.023 & 28.939 &       & 93    & 3.810 & 2.954 & -0.278 & 10.813 \\
    49    & 0.311 & 0.429 & -0.114 & 3.005 &       & 94    & 1.109 & 2.610 & -0.709 & 17.924 \\
    50    & 7.222 & 5.252 & -0.017 & 23.817 &       & 95    & 3.633 & 3.020 & 0.040 & 16.968 \\
    51    & 1.071 & 2.407 & -0.309 & 16.746 &       & 96    & 1.074 & 1.305 & -0.088 & 5.768 \\
    \bottomrule
     \multicolumn{11}{l}{\shortstack[l]{\parbox{14cm}{\vspace{0.1cm} \scriptsize
Note: This table summarizes the effects on cutoffs (endogenous competitiveness) of a 5\% increase in international trade costs under quality model
}}
}
    \end{tabular}%

%% file: cf1_noqua.tex

    \begin{tabular}{lrrrrrlrrrr}
    \toprule
    HS    & \multicolumn{1}{l}{Mean} & \multicolumn{1}{l}{Std. Dev.} & \multicolumn{1}{l}{Min} & \multicolumn{1}{l}{Max} &       & HS    & \multicolumn{1}{l}{Mean} & \multicolumn{1}{l}{Std. Dev.} & \multicolumn{1}{l}{Min} & \multicolumn{1}{l}{Max} \\
    \midrule
    2     & 1.047 & 1.272 & -0.293 & 7.608 &       &       &       &       &       &  \\
    3     & 0.586 & 1.051 & -0.158 & 7.937 &       & 52    & 0.607 & 0.931 & -0.010 & 6.390 \\
    4     & 0.446 & 0.787 & -0.661 & 5.415 &       & 54    & 0.646 & 1.266 & -0.399 & 9.755 \\
    6     & 2.593 & 3.507 & -0.452 & 27.547 &       & 55    & 0.597 & 1.095 & -0.370 & 7.856 \\
    7     & 0.663 & 1.209 & -0.107 & 11.402 &       & 58    & 1.168 & 1.486 & -0.130 & 7.156 \\
    8     & 1.616 & 1.636 & -0.137 & 8.661 &       & 59    & 0.543 & 0.847 & -0.399 & 4.301 \\
    9     & 0.389 & 0.565 & -0.138 & 2.651 &       & 61    & 1.189 & 1.491 & -0.031 & 7.398 \\
    10    & 0.497 & 0.697 & -0.526 & 5.233 &       & 62    & 1.363 & 1.891 & -0.061 & 10.884 \\
    11    & 0.381 & 0.725 & -0.150 & 7.001 &       & 63    & 0.633 & 0.750 & -0.508 & 3.724 \\
    15    & 0.449 & 0.662 & -0.034 & 4.073 &       & 64    & 0.695 & 0.876 & -0.075 & 4.262 \\
    16    & 0.659 & 1.185 & -0.463 & 8.025 &       & 65    & 0.791 & 2.017 & -0.125 & 22.371 \\
    17    & 0.364 & 0.561 & -0.010 & 3.607 &       & 66    & 2.033 & 2.519 & -0.131 & 19.504 \\
    18    & 0.232 & 0.320 & -0.009 & 1.939 &       & 67    & 4.231 & 4.888 & -0.462 & 19.710 \\
    19    & 0.327 & 0.420 & -0.010 & 2.236 &       & 68    & 0.727 & 1.212 & -0.159 & 8.095 \\
    20    & 0.407 & 0.533 & -0.361 & 2.586 &       & 69    & 1.127 & 1.502 & -0.173 & 8.505 \\
    21    & 0.412 & 0.670 & -0.013 & 4.640 &       & 70    & 0.518 & 0.886 & -0.394 & 5.239 \\
    23    & 0.481 & 0.764 & -0.012 & 4.653 &       & 71    & 3.212 & 4.667 & -0.223 & 24.524 \\
    24    & 0.923 & 1.302 & -0.292 & 10.050 &       & 73    & 0.457 & 0.665 & -0.036 & 4.666 \\
    30    & 0.731 & 0.952 & -0.179 & 4.721 &       & 74    & 0.336 & 0.673 & -0.031 & 4.304 \\
    32    & 0.400 & 0.613 & -0.253 & 3.211 &       & 76    & 0.240 & 0.422 & -0.104 & 3.735 \\
    34    & 0.419 & 0.636 & -0.006 & 3.691 &       & 82    & 0.858 & 1.104 & -0.362 & 5.341 \\
    35    & 0.525 & 0.824 & -0.009 & 5.677 &       & 87    & 0.766 & 1.347 & -0.749 & 6.729 \\
    36    & 7.257 & 6.698 & 0.096 & 36.632 &       & 88    & 1.041 & 1.260 & -0.185 & 8.286 \\
    37    & 2.159 & 1.844 & -0.065 & 8.031 &       & 89    & 4.455 & 3.799 & -0.039 & 17.700 \\
    38    & 0.330 & 0.518 & -0.084 & 3.462 &       & 90    & 0.339 & 0.426 & -0.378 & 2.027 \\
    39    & 0.284 & 0.562 & -0.013 & 4.220 &       & 91    & 9.938 & 5.588 & -0.051 & 36.798 \\
    40    & 0.401 & 0.740 & -0.023 & 5.858 &       & 92    & 2.672 & 4.312 & -0.399 & 24.148 \\
    42    & 0.952 & 1.652 & -0.084 & 13.658 &       & 93    & 8.151 & 4.063 & 1.295 & 18.768 \\
    49    & 0.300 & 0.418 & -0.217 & 2.798 &       & 94    & 0.488 & 0.711 & -0.368 & 5.283 \\
    50    & 0.792 & 1.732 & -0.020 & 13.111 &       & 95    & 0.649 & 0.984 & -0.273 & 6.777 \\
    51    & 0.840 & 1.897 & -0.131 & 12.506 &       & 96    & 0.644 & 0.762 & -0.107 & 3.713 \\
    \bottomrule
     \multicolumn{11}{l}{\shortstack[l]{\parbox{14cm}{\vspace{0.1cm} \scriptsize
Note: This table summarizes the effects on cutoffs (endogenous competitiveness) of a 5\% increase in international trade costs under Melitz and Ottaviano (2008) model
}}
}
    \end{tabular}%

%% file: cf2_1.tex

    \begin{tabular}{lrrlrrrrrrr}
    \toprule
    Country & \multicolumn{1}{l}{$\Delta$} &       & Country & \multicolumn{1}{l}{$\Delta$} &       & \multicolumn{1}{l}{Country} & \multicolumn{1}{l}{$\Delta$} &       & \multicolumn{1}{l}{Country} & \multicolumn{1}{l}{$\Delta$} \\
    \midrule
    ABW   & 0     &       & DMA   & 0     &       & \multicolumn{1}{l}{BOL} & -92.825249 &       & \multicolumn{1}{l}{GTM} & -142.23793 \\
    AFG   & 0     &       & DNK   & -114.35749 &       & \multicolumn{1}{l}{BRA} & -88.041534 &       & \multicolumn{1}{l}{GUM} & 0 \\
    AGO   & -5.9250102 &       & DOM   & -290.53134 &       & \multicolumn{1}{l}{BRB} & 0     &       & \multicolumn{1}{l}{GUY} & 0 \\
    AIA   & 0     &       & DZA   & -211.93385 &       & \multicolumn{1}{l}{BRN} & -105.61481 &       & \multicolumn{1}{l}{HKG} & 0 \\
    ALB   & 0     &       & ECU   & -150.62788 &       & \multicolumn{1}{l}{BTN} & -30.520798 &       & \multicolumn{1}{l}{HND} & -71.087296 \\
    AND   & 0     &       & EGY   & -139.12598 &       & \multicolumn{1}{l}{CAF} & -6.0313959 &       & \multicolumn{1}{l}{HRV} & -35.601143 \\
    ARE   & 110.83006 &       & ERI   & 0     &       & \multicolumn{1}{l}{CAN} & -134.60129 &       & \multicolumn{1}{l}{HTI} & 10.083085 \\
    ARG   & -172.65367 &       & ESP   & -3.0844131 &       & \multicolumn{1}{l}{CCK} & 0     &       & \multicolumn{1}{l}{HUN} & -71.353348 \\
    ARM   & -83.855202 &       & EST   & -137.4808 &       & \multicolumn{1}{l}{CHE} & -221.01123 &       & \multicolumn{1}{l}{IDN} & -61.682568 \\
    ASM   & 0     &       & ETH   & -14.915351 &       & \multicolumn{1}{l}{CHL} & 148.51651 &       & \multicolumn{1}{l}{IND} & -79.18856 \\
    ATF   & 0     &       & FIN   & -139.43924 &       & \multicolumn{1}{l}{CHN} & -18.930056 &       & \multicolumn{1}{l}{IOT} & 0 \\
    ATG   & 0     &       & FJI   & -1.2493064 &       & \multicolumn{1}{l}{CIV} & -566.53101 &       & \multicolumn{1}{l}{IRL} & -159.63408 \\
    AUS   & -85.625969 &       & FLK   & 0     &       & \multicolumn{1}{l}{CMR} & 20.903 &       & \multicolumn{1}{l}{IRN} & 125.87421 \\
    AUT   & -35.475746 &       & FRA   & -48.415028 &       & \multicolumn{1}{l}{COG} & 9.3548021 &       & \multicolumn{1}{l}{IRQ} & -94.298141 \\
    AZE   & -88.898872 &       & FSM   & 0     &       & \multicolumn{1}{l}{COK} & 0     &       & \multicolumn{1}{l}{ISL} & -217.56667 \\
    BDI   & 197.47388 &       & GAB   & -317.1246 &       & \multicolumn{1}{l}{COL} & -32.865913 &       & \multicolumn{1}{l}{ISR} & -117.73706 \\
    BEL   & 4.9142151 &       & GBR   & -89.882896 &       & \multicolumn{1}{l}{COM} & 0     &       & \multicolumn{1}{l}{ITA} & -4.3799672 \\
    BEN   & 340.07468 &       & GEO   & -77.644836 &       & \multicolumn{1}{l}{CPV} & 0     &       & \multicolumn{1}{l}{JAM} & -181.83876 \\
    BFA   & -499.91898 &       & GHA   & 313.25208 &       & \multicolumn{1}{l}{CRI} & -150.97874 &       & \multicolumn{1}{l}{JOR} & -117.16895 \\
    BGD   & -52.478401 &       & GIB   & 0     &       & \multicolumn{1}{l}{CUB} & 0     &       & \multicolumn{1}{l}{JPN} & -34.355942 \\
    BGR   & -193.67613 &       & GIN   & -221.29001 &       & \multicolumn{1}{l}{CXR} & 0     &       & \multicolumn{1}{l}{KAZ} & -3.359056 \\
    BHR   & 0     &       & GMB   & -127.25304 &       & \multicolumn{1}{l}{CYM} & 0     &       & \multicolumn{1}{l}{KEN} & 114.01509 \\
    BHS   & -313.15955 &       & GNB   & 425.91129 &       & \multicolumn{1}{l}{CYP} & -286.89645 &       & \multicolumn{1}{l}{KGZ} & -32.405098 \\
    BIH   & -70.668152 &       & GNQ   & 164.41464 &       & \multicolumn{1}{l}{CZE} & -132.38187 &       & \multicolumn{1}{l}{KHM} & -155.89485 \\
    BLR   & 9.4817276 &       & GRC   & -42.363346 &       & \multicolumn{1}{l}{DEU} & -97.001007 &       & \multicolumn{1}{l}{KIR} & 0 \\
    BLZ   & -127.64666 &       & GRD   & 0     &       & \multicolumn{1}{l}{DJI} & -76.750755 &       & \multicolumn{1}{l}{KNA} & 0 \\
    BMU   & 0     &       & GRL   & 0     &       &       &       &       &       &  \\
    \bottomrule
    
     \multicolumn{11}{l}{\shortstack[l]{\parbox{14cm}{\vspace{0.1cm} \scriptsize
Note: This table summarizes the changes in endogenous cutoffs following an increase in preference for quality
}}
}
    \end{tabular}%

%% file: cf2_2.tex

    \begin{tabular}{lrrlrrrrrrr}
    \toprule
    Country & \multicolumn{1}{l}{$\Delta$} &       & Country & \multicolumn{1}{l}{$\Delta$} &       & \multicolumn{1}{l}{Country} & \multicolumn{1}{l}{$\Delta$} &       & \multicolumn{1}{l}{Country} & \multicolumn{1}{l}{$\Delta$} \\
    \midrule
    KOR   & -35.787136 &       & QAT   & -219.32658 &       & \multicolumn{1}{l}{MYS} & -158.87752 &       & \multicolumn{1}{l}{TKM} & -52.555267 \\
    KWT   & -460.9559 &       & ROM   & 0     &       & \multicolumn{1}{l}{NCL} & 0     &       & \multicolumn{1}{l}{TMP} & 0 \\
    LAO   & -436.81769 &       & RUS   & -98.661285 &       & \multicolumn{1}{l}{NER} & -230.2971 &       & \multicolumn{1}{l}{TON} & 0 \\
    LBN   & 20.976509 &       & RWA   & -396.1434 &       & \multicolumn{1}{l}{NFK} & 0     &       & \multicolumn{1}{l}{TTO} & -104.69793 \\
    LBR   & 292.86221 &       & SAU   & -45.471489 &       & \multicolumn{1}{l}{NGA} & -42.240189 &       & \multicolumn{1}{l}{TUN} & -165.00168 \\
    LBY   & 0     &       & SEN   & -349.76016 &       & \multicolumn{1}{l}{NIC} & 74.577103 &       & \multicolumn{1}{l}{TUR} & -73.441292 \\
    LCA   & 0     &       & SGP   & 0     &       & \multicolumn{1}{l}{NIU} & 0     &       & \multicolumn{1}{l}{TUV} & 0 \\
    LKA   & -184.27283 &       & SHN   & 0     &       & \multicolumn{1}{l}{NLD} & -138.3261 &       & \multicolumn{1}{l}{TWN} & -55.560108 \\
    LTU   & -159.97792 &       & SLB   & 0     &       & \multicolumn{1}{l}{NOR} & -83.566177 &       & \multicolumn{1}{l}{TZA} & -234.38101 \\
    LVA   & -20.440449 &       & SLE   & 194.28214 &       & \multicolumn{1}{l}{NPL} & -314.22827 &       & \multicolumn{1}{l}{UGA} & -58.329792 \\
    MAC   & 0     &       & SLV   & -65.059845 &       & \multicolumn{1}{l}{NRU} & 0     &       & \multicolumn{1}{l}{UKR} & -150.4944 \\
    MAR   & -100.28398 &       & SMR   & 0     &       & \multicolumn{1}{l}{NZL} & -180.14423 &       & \multicolumn{1}{l}{URY} & -185.035 \\
    MDA   & -37.104824 &       & SOM   & 0     &       & \multicolumn{1}{l}{OMN} & -0.6131459 &       & \multicolumn{1}{l}{USA} & -83.864113 \\
    MDG   & -21.992657 &       & SPM   & 0     &       & \multicolumn{1}{l}{PAK} & -30.568909 &       & \multicolumn{1}{l}{UZB} & 110.60154 \\
    MDV   & 0     &       & STP   & 0     &       & \multicolumn{1}{l}{PAL} & 0     &       & \multicolumn{1}{l}{VCT} & 0 \\
    MEX   & -26.129211 &       & SUR   & -325.828 &       & \multicolumn{1}{l}{PAN} & -197.16428 &       & \multicolumn{1}{l}{VEN} & -17.546169 \\
    MHL   & 0     &       & SVK   & -167.24568 &       & \multicolumn{1}{l}{PCN} & 0     &       & \multicolumn{1}{l}{VGB} & 0 \\
    MLI   & 7.2255492 &       & SVN   & -107.09841 &       & \multicolumn{1}{l}{PER} & -114.9763 &       & \multicolumn{1}{l}{VNM} & -68.709572 \\
    MLT   & 0     &       & SWE   & -112.68818 &       & \multicolumn{1}{l}{PHL} & -77.56852 &       & \multicolumn{1}{l}{VUT} & 0 \\
    MMR   & -69.190582 &       & SYC   & 0     &       & \multicolumn{1}{l}{PLW} & 0     &       & \multicolumn{1}{l}{WLF} & 0 \\
    MNG   & -217.73524 &       & SYR   & -77.680817 &       & \multicolumn{1}{l}{PNG} & 0     &       & \multicolumn{1}{l}{WSM} & 0 \\
    MNP   & 0     &       & TCA   & 0     &       & \multicolumn{1}{l}{POL} & -64.078629 &       & \multicolumn{1}{l}{YEM} & -515.06506 \\
    MOZ   & -138.54242 &       & TCD   & 105.06513 &       & \multicolumn{1}{l}{PRK} & 0     &       & \multicolumn{1}{l}{ZAF} & -143.84947 \\
    MRT   & 179.3062 &       & TGO   & 301.63589 &       & \multicolumn{1}{l}{PRT} & -181.69063 &       & \multicolumn{1}{l}{ZAR} & 0 \\
    MSR   & 0     &       & THA   & -39.003437 &       & \multicolumn{1}{l}{PRY} & -254.91846 &       & \multicolumn{1}{l}{ZMB} & -184.53088 \\
    MUS   & 0     &       & TJK   & -232.86932 &       & \multicolumn{1}{l}{PYF} & 0     &       & \multicolumn{1}{l}{ZWE} & 150.78076 \\
    MWI   & 37.507915 &       & TKL   & 0     &       &       &       &       &       &  \\
    \bottomrule
   \multicolumn{11}{l}{\shortstack[l]{\parbox{14cm}{\vspace{0.1cm} \scriptsize
Note: This table summarizes the changes in endogenous cutoffs following an increase in preference for quality 
}}
}
    \end{tabular}%

%% file: size-gain.tex
{\tabcolsep=1.3mm
    \begin{tabular}{rrrr}
    \toprule
    \toprule
    \multicolumn{1}{l}{Dependent Variable:} & \multicolumn{3}{c}{Changes in Cost Cutoff} \\
    \midrule
          & \multicolumn{1}{c}{(1)} & \multicolumn{1}{c}{(2)} & \multicolumn{1}{c}{(3)} \\
    \multicolumn{1}{l}{(log) Population} & -6.188** & -8.510*** & -8.929*** \\
          & (3.139) & (2.044) & (3.299) \\
    \multicolumn{1}{l}{Initial Preference} &       & -0.795 & -0.634 \\
          &       & (0.606) & (0.700) \\
    \multicolumn{1}{l}{(log) GDP per capita} &       &       & -13.617 \\
          &       &       & (9.317) \\
   \multicolumn{1}{l}{Number of Observations} &   168    &  168     & 160 \\
    \bottomrule
    \bottomrule
    \multicolumn{4}{l}{\shortstack[l]{\parbox{10cm}{\vspace{0.1cm} \scriptsize
Notes: *** p < 0.01, ** p < 0.05, * p < 0.1. Robust standard errors are in parentheses. The dependent variable in specifications (1)-(3) is the level of aggregate cutoff change at the country-level. Country-level control variables are the initial preferences for quality and (log) GDP per capita.
}}
}
    \end{tabular}}%